\documentclass{article}

\usepackage{amsmath}
\usepackage{amsthm}
\usepackage{amsfonts}
\usepackage{appendix}
\usepackage{bbold}
\usepackage{url}
\usepackage{hyperref}
\hypersetup{
	colorlinks,
	linkcolor={red!100!black},
	citecolor={blue!100!black},
	urlcolor={blue!80!black}
} 							
\usepackage[dvipsnames]{xcolor}
\usepackage{tikz}
\usepackage{pgfplots}
\usepackage{pgf}
\usetikzlibrary{patterns}
\usepackage{float}
\usepackage{graphics}      
\usepackage{graphicx}      
\usepackage[caption=false]{subfig}
\usepackage{enumerate}
\usepackage{nicefrac}
\usepackage{physics}
\usepackage{algorithmicx}
\usepackage[Algorithm,ruled]{algorithm} 
\usepackage{algorithm}
\usepackage{algcompatible}

\usepackage{arydshln}

\theoremstyle{plain}
\newtheorem{thm}{Theorem}
\newtheorem{claim}{Claim}
\newtheorem{cse}{Case}

\newtheorem{lma}[thm]{Lemma}

\newtheorem{rmk}[thm]{Remark}
\newtheorem*{thm*}{Theorem}
\newtheorem*{prop*}{Proposition}
\newtheorem*{lma*}{Lemma}
\newtheorem*{cor*}{Corollary}
\newtheorem*{rmk*}{Remark}

\theoremstyle{definition}
\newtheorem{defi}[thm]{Definition}
\newtheorem*{conj*}{Conjecture}

\theoremstyle{definition}
\newtheorem{question}[thm]{Question}

\usepackage{fullpage}

\usepackage{todonotes}

\usepackage{authblk}

\renewcommand{\tr}[1]{\mathrm{Tr}#1}

\newcommand{\D}[2]{\frac{\mathrm{d}#1}{\mathrm{d}#2}}

\definecolor{plot1}{RGB}{48, 79, 254}
\definecolor{plot2}{RGB}{0, 184, 212}
\definecolor{plot3}{RGB}{0, 200, 83}
\definecolor{plot4}{RGB}{255, 214, 0}
\definecolor{plot5}{RGB}{255, 109, 0}
\definecolor{plot6}{RGB}{221, 44, 0}

\newcommand{\MARKSZ}{0.5mm}
\newcommand{\MARKFORM}{oplus*}
\newcommand{\FONTSZ}{\small}

\begin{document}

\title{\vspace{-35pt}Device-independent Randomness Amplification and Privatization}

\author[1]{Max Kessler\thanks{kesslerm@student.ethz.ch}}
\author[1]{Rotem Arnon-Friedman\thanks{rotema@itp.phys.ethz.ch}}
\affil[1]{Institute for Theoretical Physics, ETH-Z\"urich, CH-8093, Z\"urich, Switzerland}
\date{}

\maketitle

\vspace{-35pt}

\begin{abstract}
	Randomness is an essential resource in computer science. In most applications perfect, and sometimes private, randomness is needed, while it is not even clear that such a resource exists. It is well known that the tools of classical computer science do not allow us to create perfect and secret randomness from a single weak public source. Quantum physics, on the other hand, allows for such a process, even in the most paranoid cryptographic sense termed ``quantum device-independent cryptography''.  In this work we propose and prove the security of a new device-independent protocol that takes any single public Santha-Vazirani source as input and creates a secret close to uniform string in the presence of a quantum~adversary.  
	
	Our work is the first to achieve randomness amplification with all the following properties: (1)~amplification and ``privatization'' of a public Santha-Vazirani source with arbitrary bias (2) the use of a device with only two components (compared to polynomial number of components) (3) non-vanishing extraction rate and (4) maximal noise tolerance. In particular, this implies that our protocol is the first protocol that can possibly be implemented  with reachable parameters.
	We are able to achieve these by combining three new tools: a particular family of Bell inequalities, a proof technique to lower bound entropy in the device-independent setting, and a special framework for quantum-proof multi-source extractors. 
\end{abstract}

\section{Introduction}

Randomness is widely used in computer science; it is essential for cryptography and (at the least) beneficial for many other scenarios, e.g., when designing  efficient algorithms or proving the existence of certain functions and combinatorial objects of interest, via the probabilistic method~\cite{vadhan2012pseudorandomness}. 

Unfortunately, we cannot know for sure that randomness even exists; it might as well be that everything in nature is completely deterministic and fixed in advance. Furthermore, even if we assume the existence of some sources of randomness in nature, it is not clear at all that there are sources of \emph{perfect} randomness. Physical sources of randomness, such as radioactive decay or thermal noise, can be used to produce unpredictable bit strings, but those are usually partially biased and correlated bits. Even worse, how unpredictable these sources of randomness are depends on the knowledge of the observer regarding the physical system. For a person who can keep track of all microscopic degrees of freedom the outcomes can be completely predictable.

The question addressed in this work is familiar --- can we reduce the amount of perfect randomness required for one's task of interest? In particular, we are interested in the cryptographic point of view. That is, when we say perfect randomness, for example, we mean that it should be uniform even with respect to  some prior knowledge or side information of a malicious party or an adversary.\footnote{This is the most demanding context to consider randomness in. A positive answer to the these questions in the cryptographic sense also implies a positive answer in applications where a malicious party is not of interest. The opposite direction is, of course, not true.} We then ask:

\begin{question}
	Can \emph{perfect} randomness be created from \emph{weak or short} randomness?
\end{question}
\begin{question}
	 Can \emph{private, secret,} randomness be created from \emph{public} randomness?
\end{question}

By weak randomness we mean that the produced bits can be correlated and biased (though not completely deterministic). One such source, investigated in many  works and of relevance for the current one, is the so called ``Santha-Vazirani source'', or SV-source,~\cite{SV} --- a source that produces a sequence of bits, where each bit has \emph{some} randomness given all previous ones. This source is a special type of the more general ``min-entropy source''~\cite{min-entropy-source} (both defined formally below). 
Public randomness means that anyone can see the random string once it is produced. This is the case, for example, for the random numbers produced by the NIST randomness beacon\footnote{\url{http://www.nist.gov/itl/csd/ct/nist_beacon.cfm}}; they are publicly available over the internet. 

``Classical'' computer science addresses the first question by considering \emph{pseudorandom generators} and \emph{randomness extractors}. Pseudorandom generators take a short perfectly random seed and generate from it a longer string of bits that no efficient algorithm can distinguish from a uniformly random string (see, e.g.,~\cite{goldreich2010primer} for a survey). Thus, for the existence of pseudorandom generators we must make some assumptions regarding the complexity of certain computational tasks~\cite{diffie1976new,shamir1983generation}. Hence, they cannot be used when considering an all-powerful adversary. 

Randomness extractors are functions that take a weak random source as an input and return an almost-uniform string as the output (see~\cite{nisan1999extracting}). Extractors are ``information-theoretically secure'' in the sense that, in contrast to pseudorandom generators, they do not require the use of computational  assumptions. However, as widely known, no function can take a single SV-source and create close to uniform randomness out of it~\cite{SV}. We therefore ought to consider extractors which either take an additional independent (short) random seed as input or several independent weak sources of randomness. These are called seeded extractors and multi-source extractors, respectively. (See~\cite{de2012trevisan,kasher2010two} for examples of extractors that work even in the presence of a quantum adversary).

The answer to the second question seems obvious and intuitive --- if everything is known in public (i.e., the initial source of randomness and the procedure, or protocol, used to manipulate it) then there is no way to create some private, secret, information out of it.

Quantum physics allows us to tackle the above questions from another angle and derive different conclusions, without making assumptions regarding computational complexity or the number of independent sources~\cite{bera2016randomness,acin2016certified}. By preparing certain quantum states, e.g., a photon in a particular configuration, and measuring them one can generate perfectly random bits which, according to the laws of physics, were not known to anybody in advance.\footnote{Note, however, that quantum physics (as well as any other physical theory) cannot exclude the possibility that there is no randomness in nature to begin with. To prove that the outcome of a measurement performed on a quantum state is random, for example, we must first assume that we have the ability to choose the different states and measurements we would like to perform.}

Taking such an approach to answer the above questions is ``unfair'' and unsatisfactory. Firstly, one can argue that allowing the use of a source of, say, photons is like allowing the use of private unbiased coins. (And allowing the use of entangled photons, distributed among several parties, is like allowing shared randomness). Secondly, and significantly more importantly, when trying to implement such a source of randomness we find that creating perfect quantum states and measurements is practically impossible. In the cryptographic setting, imperfections and noise in the implementation are being exploited to gain information on the generated randomness~\cite{gerhardt2011full}. 

To solve these issues (and many others) the quantum cryptography community took one step further~\cite{ekert2014ultimate}. In the so called \emph{device-independent} approach we let the adversary prepare the quantum devices used to generate the desired randomness. The honest parties interact with the device prepared by the adversary to test it and abort the protocol if its behaviour does not fit some predefined requirements.
Then, the entire procedure is known to the adversary and there are no ``hidden private coins''. Furthermore, we can no longer assume anything about the inner-workings of the device. Hence, if we are able to prove that the produced outcomes are secure to use, then the statement is inherently independent of the physical device and therefore robust to imperfections in the implementation.

In the device-independent scenario it might be that the adversary programmed the device to output a certain fixed string which is completely known to her. Thus, at first sight, it seems impossible to prove that the outputs are random from the perspective of the adversary. As known for quite some time now, the solution is to base device-independent protocols on the violation of Bell inequalities~\cite{ekert1991quantum,mayers1998quantum,barrett2005no,acin2016certified}. 

A Bell inequality~\cite{bell1964einstein} can be thought of as a game played by the honest parties using a device that includes two non-communicaiting components (the most famous one being the CHSH inequality~\cite{CHSH} or CHSH game; see~\cite{brunner2014bell} for a review on Bell inequalities and non-locality).
The game has a special property~--- some quantum, non-local, strategies can win the game with probability $\omega_q$ greater than any classical, local, strategy.
Hence, if the honest parties observe that using their device they win the game with probability $\omega_q$ they conclude it must be quantum (further details are given in Section~\ref{sec:nl-games}). Otherwise they abort the protocol.
Experiments have verified the quantum advantage in such ``Bell games'' in a loophole-free way~\cite{hensen2015loophole,shalm2015strong,giustina2015significant} 
(in particular, this means that the experiments were executed without making assumptions that could otherwise be exploited by the adversary in the cryptographic setting).
It is well established that  the higher the winning probability in a game is, the higher the amount of secret randomness which was produced in the process. We show this in Section~\ref{sec:single-round} for our scenario of interest.

In this work we suggest a new quantum device-independent cryptographic protocol that uses a \emph{single public SV source} as input and produces \emph{secret close to uniform randomness}, even with respect to a quantum adversary. We state the concrete result and compare it to previous works in the following.  

\subsection{Results and contributions}

We focus in this work on the amplification of an SV-source. An SV-source with bias $\mu\in(0,0.5)$ has the following property: for each bit produced by the source $b_i$, $\Pr[b_i=0|b_1,\dotsc,b_{i-1}]\in[\mu,1-\mu]$, where $b_1,\dotsc,b_{i-1}$ are all the previous bits produced by the source. Such sources describe physical processes in which the bits are produced one after the other. Hence, the bias of each bit can depend (adversarially) on the previous bits, but not on the bits that will be produced in the future. 
Many of the processes in nature produce a sequence of bits, one bit after the other; the chronological order then implies that each bit can only depend on the past and not on the future. Thus, an SV-source can be used to describe such process in a realistic way. 

The first challenge when dealing with randomness amplification is to find an interesting (and relevant) setting to consider and devise a protocol that can be proven secure in that setting. Previous works considered different protocols and there is no ``standard model''.\footnote{Though it is always the case that some Bell game is repeated many times, as in all device-independent protocols (e.g., device-independent quantum key distribution and randomness expansion).} We first describe the scenario that we focus on and its relevance. Then we state our result and explain the main steps and ideas of the proof.

\begin{figure}
	\includegraphics[width=115mm]{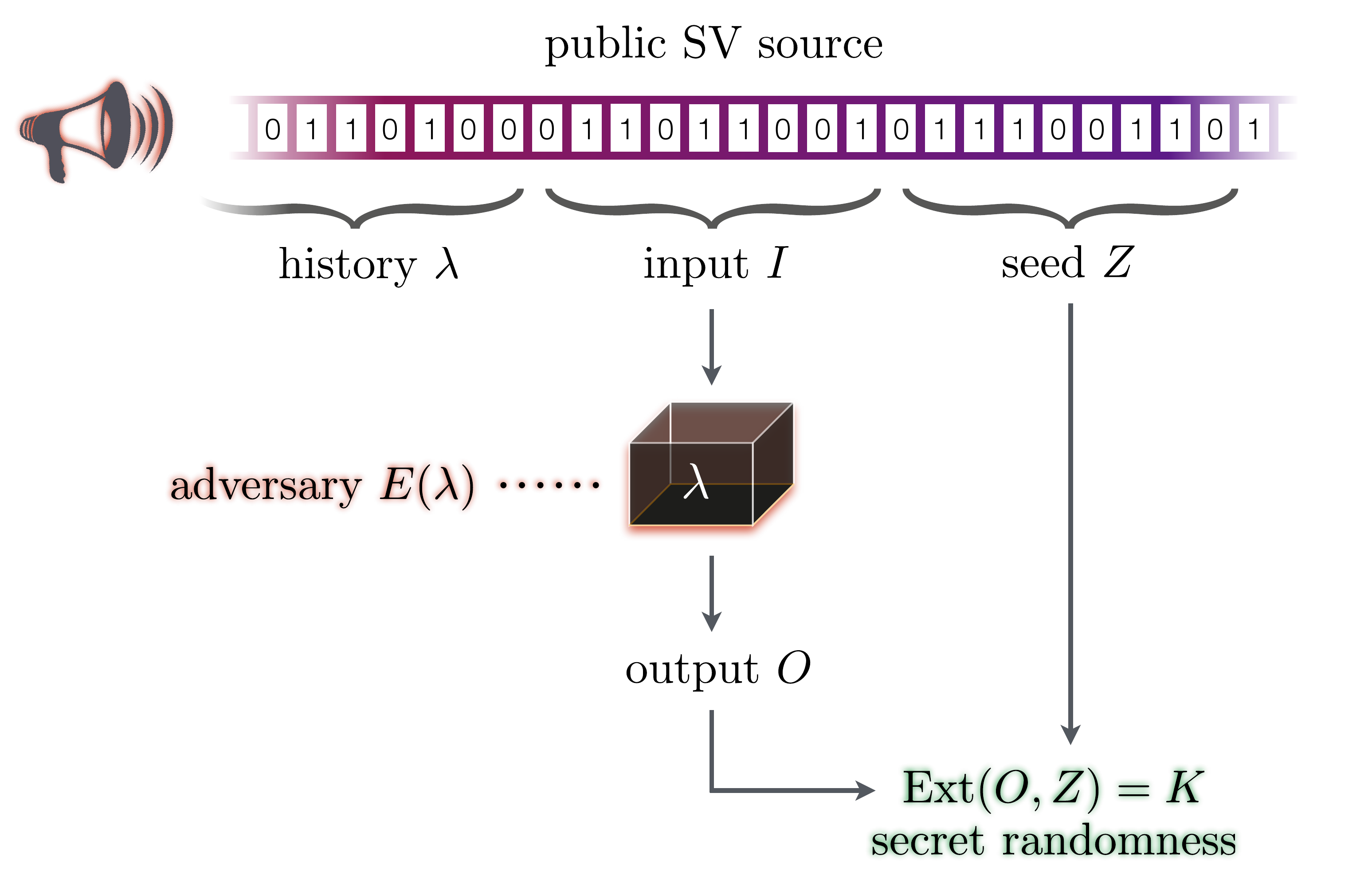}
	\centering
	\caption{An illustration of the considered setting. We start with a public SV source and a device which was created by the adversary (the black box in the figure). The goal is to produce a secret, close to uniform, string $K$. The bits produced by the SV-source when running the protocol, $I$ and $Z$, and the device can be correleated via the previous bits of the source, $\lambda$, and the adversary $E$.
	Our protocol is such that the honest party first uses some of the bits, $I$, as input to the device. The output of the device is denoted by $O$. Then, a special type of randomness extractor is applied to $O$ and additional bits $Z$ from the source. The result is the output randomness $K$.} 
	\label{fig:setting_intro}
\end{figure}

The setting that we consider is illustrated in Figure~\ref{fig:setting_intro}. 
We start with an arbitrary, public, SV-source with bias $\mu\in(0,0.5)$. $\lambda$ denotes all the bits produced before the adversary, Eve, prepares the device for the honest party, Alice. $\lambda$ can also include any other piece of classical information from the past that might be of relevance to Eve.
Eve then creates the device, denoted by the black box in the figure, depending on $\lambda$. She can keep quantum side information $E=E(\lambda)$ correlated with the device for herself; this side information can later be used by Eve to gain information about the final random string. Once Alice holds the device she can use it together with additional bits produced by the source, $I$ and $Z$ in the figure, to create her final secret random string~$K$.

The SV-source can be controlled by an untrusted party but we assume that every bit, when produced, has some randomness conditioned on all side information. Mathematically, for the first bit of $I$, $I_1$, for example, we have $\frac{1}{2} - \mu \leq \Pr[I_1 = 0|\lambda] \leq \frac{1}{2} + \mu$.

In particular, in the above explained scenario it holds that, given the history $\lambda$ and Eve's knowledge $E$, the device $D$ and the sequence of bits $I \circ Z$ are independent. That is,\footnote{This should be understood on the intuitive level, as we did not define the device $D$ in a mathematical way. The exact setting is modelled formally in Section~\ref{sec:assumptions}.}
\begin{equation}\label{eq:intro_mutual_info_setting}
	\mathrm{I}(D:I \circ Z|\lambda E) =0 \;,
\end{equation}
where $\mathrm{I(\bullet:\bullet|\bullet)}$ is the conditional mutual information.

We remark that the considered scenario is relevant for actual implementations of randomness amplification protocols: the chronological order of events is such that Eve can prepare the device depending only on past information (the history) but not on the bits which will be produced after delivering the device to Alice. This implies that all correlations between the following bits produced by the source and the device are due to past events and Eve's side information. Thus, Equation~\eqref{eq:intro_mutual_info_setting} holds. Several previous works, e.g.,~\cite{CRrand,Gallego2013,brandao2016realistic}, considered similar settings as well.

The main contribution of our work is a construction of a \emph{device-independent randomness amplification protocol} that uses a single public SV-source to create secret and close to uniform randomness, with respect to all of the knowledge that the adversary has:

\begin{thm}[Informal]\label{thm:informal}
	Given any public SV-source with bias $\mu \in (0,0.5)$ there exists a protocol, requiring a two-component device, such that:
	\begin{enumerate}
		\item (Soundness) For any device $D$ used to implement the protocol such that Equation~\eqref{eq:intro_mutual_info_setting} holds, either the protocol aborts with overwhelming probability or an $\varepsilon$-close to uniform (given the adversary's knowledge) string $K$ is produced. 
		\item (Completeness) There exists an honest implementation of the device such that the protocol aborts with negligible probability when using this device, even in the presence of noise.
	\end{enumerate} 
\end{thm}

The formal statement is given in Theorem~\ref{thm:formal}.  The soundness, or security, parameter $\varepsilon$ depends on the bias of the source, $\mu$, as well as the parameters of an extractor used in our protocol to create $K$. For certain choices of parameters the protocol can be made explicit.

Theorem~\ref{thm:informal} improves upon the prior state-of-the-art in several significant aspects (see Section~\ref{sec:related_works} and Table~\ref{tab:work_comp} for comparison with previous works):
\begin{enumerate}
	\item \textbf{Device requirement} -- we only require that the device includes two components (the lowest possible), compared to a polynomial number in previous works that considered a public weak source of randomness.. \\
	This means that the black box  in Figure~\ref{fig:setting_intro} consists of two separated parts.\footnote{One can imagine the two components as being two computers or, alternatively,  two provers in a multi-prover interactive proof system.} Having two components is a necessary requirement for protocols based on Bell inequalities. As we explain in Section~\ref{sec:related_works}, previous works that considered a public weak source had to use, at the least, polynomial number of components, which is not realistic. Other works that allowed a constant number of devices could not derive a result for an arbitrary bias $\mu$, a public SV-source, and/or quantum adversaries.
	\item \textbf{Extraction rate (efficiency)} -- for a large range of parameters we can extract a linear number of bits\footnote{To be more precise -- for a large range of parameters (the full details are given in Remark~\ref{rmk:RAP-secrecy}) there is an explicit extractor that can be used in our protocol to extract a linear number of bits. If one is interested in an explicit protocol for all parameters, there are two options: 1)~A simple modification of our protocol, which requires the use of a device with 4 components, can be used to extract a sub-linear number of bits using a three-source extractor. (A similar thing was previously done in~\cite[Theorem 2]{brandao2016realistic} but the resulting protocol requires 8 components and the security proof uses an additional assumption of a private SV-source; see Section~\ref{sec:related_works}). 2)~Using the current protocol (with only two components) one can extract a logrithmic number of bits. If, in the future, new (classical) two-source extractors with better parameters are developed, they can be used in our protocol to achieve better extraction rates without modifying the protocol or its security proof.} while maintaining cryptographic security level, compared to a vanishing extraction rate in previous works that considered a public weak source of randomness. \\ 
	Using an extractor with sufficiently good parameters $\varepsilon$ can be made exponentially small in the number of bits taken from the SV-source while extracting a linear number of bits. Previous works could not achieve this, \emph{independently} of the extractor used in the protocol.
	\item \textbf{Robustness} -- we are able to tolerate the maximal amount of noise, compared to low noise levels in previous works that considered a public weak source of randomness. \\
	The completeness statement holds for any amount of noise in the implementation which still results in a violation of the Bell inequality.\footnote{This can be seen, for example, from Figure~\ref{fig:single_round_bound} below which shows that non-zero entropy can be certified as long as there is a violation of the Bell inequality.} This is the maximal possible amount one can hope to tolerate. 
\end{enumerate}

Apart from randomness amplification, our protocol can also be used as a main building block for device-independent randomness expansion and key distribution using weak sources of randomness. More details are given in Section~\ref{sec:conclustions}.

Theorem~\ref{thm:informal} \emph{cannot} be derived by improving previously known techniques (as explained in Section~\ref{sec:related_works}). To prove it we present a completely new proof, which can be of independent interest.  Our proof uses three different tools which were developed recently and were not used before in the context of randomness amplification. 
One particular example for an independent technical contribution is the proof given in Section~\ref{sec:single-round}, where we investigate a new type of Bell inequalities and show, for the first time, that they can also be used in a cryptographic setting.
Another contribution is presenting a first application of a special type of extractors that were recently introduced in~\cite{Extractors}. The existence of such extractors is what allows us to produce randomness, in the presence of a quantum adversary, when starting with a single public SV-source.

\subsection{Main steps in the proof}\label{sec:proof_steps}

Our protocol is stated as Protocol~\ref{alg:RAP} in Section~\ref{sec:ra_protocol}. The protocol is simple: the device is used sequentially with the inputs $I$ from the SV-source to create the outputs $O$. Once all the outputs are produced Alice calculates the average violation of a specific Bell inequality from the raw data and aborts if the violation is not sufficiently high. If she does not abort then a special type of extractor is applied to $O$ together with additional bits from the source $Z$. 

\subsubsection*{Step 1: Choosing the ``correct'' Bell inequality}

\begin{figure}
	\includegraphics[width=40mm]{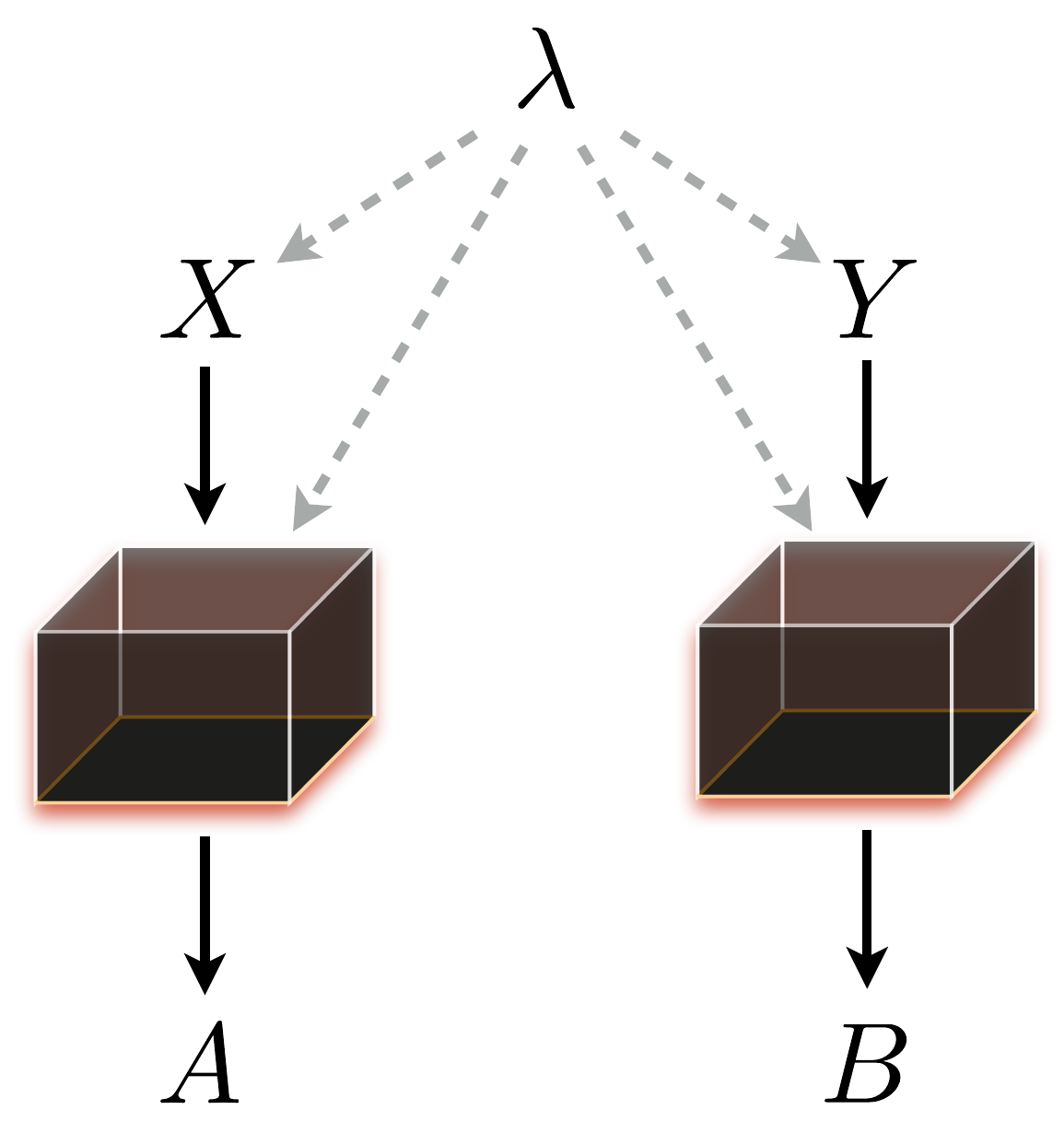}
	\centering
	\caption{Correlations between the device and the inputs. The two comonents of the device are denoted by the black boxes. The inputs to the two components, $X$ and $Y$, come from the SV-source. The outputs are denoted by $A$ and $B$. The device and the inputs can be correlated via the history $\lambda$, as denoted by the dashed arrows. A violation of an MDL inequality certifies that the device cannot be classical in this setting.} 
	\label{fig:intro_mdl}
\end{figure}

As all device-independent protocols, our protocol is based on the violation of a given Bell inequality above a certain threshold.  This way Alice can make sure that the device implements a quantum non-local strategy. 
All previous protocols use the CHSH Bell inequality or other well known inequalities.

We use a recently developed family of Bell inequalities (with two parties, two inputs, and two outputs) which fits perfectly to the scenario of randomness amplification. As explained above, in our setting,  the device and the inputs $I$ can be correlated via $\lambda$. The Bell inequalities developed in~\cite{MDL}, called ``measurement dependent locality (MDL) inequalities'', are adapted to the situation illustrated in Figure~\ref{fig:intro_mdl} for \emph{any} bias of the source. They therefore accommodate the dependency between the device and the side information. In contrast, the violation of the CHSH inequality cannot be used to ``verify quantumness'' above some threshold for the bias (see further details in Section~\ref{sec:MDL}). 
Other Bell inequalities which were used in the context of randomness amplification and allowed for an arbitrary bias of the SV-source require a device with more than two components~\cite{Gallego2013,brandao2016realistic,ramanathan2015randomness}.

We note that, for the completeness of our protocol, it is crucial that for any bias of the source there is a quantum strategy (i.e., quantum state and measurements) that violate the inequality. This is indeed the case as shown in~\cite{MDL}. When proving completeness we also explain how the maximal violation within quantum physics can be found numerically. 

The rest of the steps in the proof deal with the soundness proof.

\subsubsection*{Step 2: Certifying randomness from the MDL violation after a single use of the device}

The analysis done in~\cite{MDL} for the MDL inequalities only ensures that a violation of the inequality implies that the device must be non-local, i.e., it cannot be implemented by a classical strategy. While this is important for the study of fundamental questions in physics, it is not sufficient in the cryptographic setting. A quantitive bound on how random the output of the device must look to an adversary was missing. 

The first part of our proof is devoted to deriving a relation between the violation of the MDL inequality and the amount of knowledge Eve can gain regarding the output in a single use of the device. Specifically, we prove a lower-bound on the von Neumann entropy of the output given all side information:
\begin{equation}\label{eq:intro_von_nuemann}
	H(O_i|I_i E, \lambda) \geq t \;,
\end{equation}
where $I_i$ and $O_i$ are the inputs and outputs when using the device for the $i$'th time and $t\geq 0$ depends on the bias of the source and the observed violation of the MDL inequality (see Lemma~\ref{lma:holevo_bound} for the exact bound and Figure~\ref{fig:single_round_bound} for a plot). 
The conditional von Neumann entropy is just one way of quantifying the amount of secret randomness, but as we will show below, this is the relevant quantity for us. 

A bound similar to Equation~\eqref{eq:intro_von_nuemann}, but for the CHSH inequality, was proven in~\cite{Pironio}. In the case of the CHSH inequality the inputs are assumed to be chosen uniformly and independently of the device and hence one cannot use the result of~\cite{Pironio} directly for randomness amplification. We find a way to connect the two scenarios and derive a bound as in Equation~\eqref{eq:intro_von_nuemann} for the MDL inequality from that of the CHSH inequality.

The resulting bound is non-trivial as long as the MDL inequality is violated (while if there is no violation the conditional entropy must be 0, since the device might be using a classical strategy). Combined with the following step, this property allows us to tolerate maximal amount of noise in the honest implementation of the device used in the protocol. 

\subsubsection*{Step 3: Bounding the total amount of min-entropy after multiple uses of the device}

To bound the amount of extractable randomness from the outputs of the device $O$ we need to lower bound the total conditional smooth min-entropy\footnote{The smooth min-entropy is a standrd quantity related to the, more commonly known, min-entropy; the formal definition is given in Section~\ref{sec:entropies-markov}. The important thing to know at this stage is that it tightly determines how much randomness Alice can extract from $O$ in the presence of a quantum adversary~\cite{konig2009operational}.} $H^{\varepsilon_s}_{\min}(O|IE,\lambda)$, for $\varepsilon_s\in(0,1)$, given that our protocol did not abort.

If the different uses of the device in the protocol were independent and identical, getting a bound on $H^{\varepsilon_s}_{\min}(O|IE,\lambda)$ is rather easy. On the intuitive level, the total amount of entropy in that case is the sum of the entropies in each round of the protocol~\cite{tomamichel2009fully,devetak2005distillation}. However, as the adversary is the one preparing the device, there is no reason to believe that the device behaves in an independent and identical way. The analysis is therefore more delicate. 

To overcome this difficulty we uses a new information-theoretic tool, called the entropy accumulation theorem~\cite{EAT}, to bound the total amount of smooth
min-entropy, in a sequential processes, using the von~Neumann entropy of a single step of the process. More precisely, we use the framework
developed in~\cite{RotemEAT} for proving security of device-independent cryptographic protocols using the entropy accumulation theorem. In~\cite{RotemEAT} the entropy accumulation theorem was used to prove security of device-independent key distribution and randomness expansion protocols. We adapt the different steps to our scenario of randomness amplification with the MDL inequalities. 


To prove a lower bound on $H^{\varepsilon_s}_{\min}(O|IE,\lambda)$ we start by showing that for any SV-source and device, the sequential process defined by the rounds of our protocol and the actions of the device fulfil the prerequisites of the entropy accumulation theorem. Next, using Equation~\eqref{eq:intro_von_nuemann} we devise a ``min-tradeoff function''. This function quantifies the ``worst-case von Neumann entropy'' that is accumulated in a single round of the protocol, while taking into account the observed violation of the MDL inequality. Once this function is constructed we can apply the techniques of~\cite{EAT,RotemEAT} to derive a bound on $H^{\varepsilon_s}_{\min}(O|IE,\lambda)$. 
The first order term of the lower bound on $H^{\varepsilon_s}_{\min}(O|IE,\lambda)$ is $n H(O_i|I_i E, \lambda)$, where $n$ is the number of rounds of the protocol. That is, $H^{\varepsilon_s}_{\min}(O|IE,\lambda)\in \Omega(n)$, which is optimal.  
For more details, see Section~\ref{sec:soundness}.

\subsubsection*{Step 4: Extracting the randomness}

Once a bound on the conditional smooth min-entropy is derived we need to extract the randomness using an extractor. However, since only a single SV-source is available, there is no additional independent source of randomness. Thus, standard seeded or multi-source extractors cannot be used. 

In the last step of our proof we show that the setting that we consider (as in Figure~\ref{fig:setting_intro} above) implies that a newly developed model for quantum-proof multi-source extractors can be used~\cite{Extractors}. The model presented in~\cite{Extractors}, termed the ``Markov model'', deals with extraction from multiple weak sources which are independent only given some side information, possibly quantum. Each of the sources must have sufficient amount of entropy conditioned on that side information. It was proven in~\cite{Extractors} that any (strong) multi-source extractor is also a (strong) quantum-proof multi-source extractor in the Markov model, with some loss in parameters (the exact statements which we use are presented in Section~\ref{sec:extractors}).  

We show that the considered setting implies that 
\[
	\mathrm{I}(O:Z|IE,\lambda) = 0 \;,
\]
meaning that given $I, E$, and $\lambda$, $O$ and $Z$ are independent. Furthermore, the previous step of our proof ensures that $O$ has sufficient amount of entropy conditioned on $IE\lambda$. The same is true for $Z$ since it is taken directly from the SV-source. We can therefore use a strong quantum-proof two source extractor in the Markov model to create the final string $K=\mathrm{Ext}(O,Z)$, which is close to uniform even given $ZIE\lambda$. This implies the security of our protocol. 

The use of this special type of extractors~\cite{Extractors} is what allows us to start with nothing but a single public SV-source and consider quantum side-information. Previous models for quantum-proof multi-source extractors~\cite{kasher2010two,chung2014multi} do not allow for the side information considered in the current setting. Moreover, a \emph{strong} extractor is crucial here since the seed $Z$ is public (as it comes from the public SV-source). 

We remark that $I$ and $Z$ cannot be used directly as the sources for the extractors, although they both have high min-entropy given $\lambda$ and $E$. The reason is that they are not independent given $\lambda E$. The use of the device is therefore necessary in order to create a string $O$ which is ``decoupled'' from $Z$. 

The combination of all the steps above proves the soundness of our protocol.

\subsection{Previous works}\label{sec:related_works}

We now discuss the different works and assumptions and compare them to the current work. See also Table~\ref{tab:work_comp}.

\subsubsection*{Public SV-source}

Colbeck and Renner were the first to consider the task of randomness amplification~\cite{CRrand} and give a ``proof of concept''. There, the relation between the knowledge that an adversary has about a final single bit was bounded using the expected Bell violation . They showed that using a public SV-source with bounded bias ($\mu=0.058$) and a two-component device a single close to uniform bit can be created in the presence of both quantum and non-signalling (super quantum) adversaries. The number of measurements, however, grew with their security parameter and only one bit was produced. Hence any protocol based on such approach would have resulted in a vanishing extraction rate.

Following that, \cite{Gallego2013} improved on the above result by considering a protocol that can accommodate arbitrary bias of the SV-source and tolerate some noise.
Instead of restricting the analysis to quantum adversaries~\cite{Gallego2013} focused on the stronger non-signalling adversaries. Unfortunately, the protocol required the use of many devices --- polynomial in the number of bits used from the source. One can imagine this as requiring a polynomial number of laboratories separated in space, each of which runs a quantum experiment. This is of course unrealistic in any implementation. 

To see why the proof technique of~\cite{Gallego2013} could not be extended to get results similar to ours note the following. First, to deal with an arbitrary bias of the SV-source a 5-party Bell inequality was used. This implies that any protocol based on their Bell inequality would require, at the least, 5 devices (otherwise the violation is meaningless). Second, the final randomness is extracted using a deterministic process, which is only possible since their protocol requires a polynomial number of devices (for details see the discussion in~\cite[Supplementary information C]{Gallego2013}).
To reduce the number of devices one would have to construct strong randomness extractors which are secure in the presence of non-signalling adversaries, but there are indications that such do not exist~\cite{arnon2012limits}.

\subsubsection*{Private SV-source}

In~\cite{brandao2016realistic,ramanathan2015randomness} a protocol using a constant number of devices was constructed, also when considering non-signalling adversaries. In addition, as in our work, the protocol is robust to noise and achieves a non-zero extraction rate.
The crucial difference between~\cite{brandao2016realistic,ramanathan2015randomness} and the current work is that the security proof of~\cite{brandao2016realistic,ramanathan2015randomness} assumes that the SV-source must be private, i.e., no information about the bits produced by the source can leak to the adversary at any point (also after the end of the protocol). 

One might argue that this is not such a strong requirement, especially since we anyhow assume that the final randomness created by the protocol is kept secret. However, there is one critical difference: it is implied by the security definition of randomness amplification protocols (sometimes termed composable; see Section~\ref{sec:secur_def}) that if part of the produced randomness is leaked to the adversary the rest of the bits are still close to uniform. In contrast, when proving security with a private source it is not clear at all what happens when some information about the source is leaked to the adversary. It is nowhere proven (or conjectured) that if partial information about the used source (even a single bit) is leaked the entropy of the produced string remains somewhat high. 

The proof of~\cite{brandao2016realistic,ramanathan2015randomness} cannot be used to get a protocol which can take a public SV-source as input. The reason is that the assumption regarding the privacy of the source is used in order to simplify the security criterion and argue that a classical multi-source extractor can be used to extract the randomness, although a non-signalling adversary is present. 
To allow for a public source one will need a strong multi-source extractor which is secure in the presence of a non-signalling adversary, but as mentioned above it is not clear that such exists. 

We also remark that the simplification of the security definition to a classical one due to the use of private source enabled the analysis of the total amount of min-entropy in the outputs of the device. The same analysis cannot be used as is when considering the case of a public source or when trying to bound the smooth min-entropy as we do here. 
Moreover, in~\cite{brandao2016realistic,ramanathan2015randomness} as well, Bell inequalities with more than two parties are used. Thus, such protocols cannot lead to a protocol that requires only two components as ours. 
 
\begin{table}
\centering
\begin{tabular}{c||c|c|c|c|c|c|c}
	Work & Source & Adversary &  \# Devices &  Public source? &Arbitrary bias? & Robust? & Efficient?  \\
	\hline\hline
	\cite{CRrand} & SV & Q \& NS &  2  & \checkmark & $\times$ & $\times$ & zero  \\
	\cite{Gallego2013} & SV & NS &  poly  & \checkmark & \checkmark & \checkmark  & zero \\
	\cite{brandao2016realistic} & SV &  NS &  4   & $\times$ & \checkmark & \checkmark & \checkmark \\
	{\color{Aquamarine}Current} & {\color{Aquamarine}SV} & {\color{Aquamarine}Q} &  {\color{Aquamarine}2}  & {\color{Aquamarine}\checkmark} & {\color{Aquamarine}\checkmark} & {\color{Aquamarine}\checkmark} & {\color{Aquamarine}\checkmark} \\ 
	\hdashline
	\cite{chung2014physical} & min-entropy & Q &  poly  & \checkmark & \checkmark & slightly  & zero \\ 
	\cite{chunggeneral} & min-entropy & NS &  exp  & \checkmark & \checkmark & slightly  & zero 
\end{tabular}
\caption{Comparison of the different works. Q and NS stand for a quantum and non-signalling (super-quantum) adversary respectively. The number of devices is with respect to the number of bits used from the weak source of randomness. For a more detailed comparison of previous works see also~\cite[Supplementary Information]{brandao2016realistic} and~\cite[Table 1]{acin2016certified}.}\label{tab:work_comp}
\end{table}

\subsubsection*{Public min-entropy source}

In two more recent works~\cite{chung2014physical,chunggeneral} a protocol that can amplify a public min-entropy source was suggested and its security was proven. \cite{chung2014physical} assumed a quantum adversary while~\cite{chunggeneral} considered a non-signalling one. The first part of the protocol in these works takes the min-entropy source and extracts blocks of bits, some of them close to uniform with respect to the used devices, by enumerating all possible seeds. The different blocks are then used as inputs to a randomness expansion protocol~\cite{miller2016robust}. This approach leads to a polynomial number of devices in~\cite{chung2014physical} and exponential in~\cite{chunggeneral}. Furthermore, in both works the security parameter is inverse polynomial in the number of bits used from the source, the efficiency of the protocols vanishes, and the amount of tolerated noise is low. 

A min-entropy source is of course more general than the SV-source considered in the current work. Our work cannot be extended as is to the case of a min-entropy source. On the other hand, it is also not clear how to take the work of~\cite{chung2014physical,chunggeneral} and decrease the number of devices -- to get close to uniform inputs for the randomness expansion protocol starting with a single weak source one must enumerate the seeds; each seed should then be used while running the protocol on a different set of devices. The number of devices (and hence also the zero extraction rate) is thus a fundamental part in the proof technique of~\cite{chung2014physical,chunggeneral}. 

\subsubsection*{Source-device-adversary model}
In~\cite{chung2014physical,chunggeneral} the authors model the relation between the source, the adversary, and the device differently than what we do here. In particular, they allow for some quantum side information about the \emph{source}, in contrast to our $\lambda$ which is classical. 
In all other mentioned works the assumptions regarding the relation between the three components are similar to the ones considered here (though not mentioned explicitly in the same way). In~\cite{wojewodka2016amplifying} a different scenario is considered, but the security analysis is not complete and only restricted SV-sources can be amplified.

\paragraph*{Organisation of the paper.}
We start in Section~\ref{sec:pre} with some preliminaries. In particular, the necessary information regarding the MDL inequalities and two-source extractors in the Markov model is given. Section~\ref{sec:single-round} is devoted to proving a relation between the observed violation of an MDL inequality and the knowledge that a quantum adversary can gain about the output of the device. In Section~\ref{sec:RAP} we state our randomness amplification protocol and prove its security. We end in Section~\ref{sec:conclustions} with some open questions.

\section{Preliminaries}\label{sec:pre}

\subsection{Notation}

In the following we will denote by 
\begin{itemize}
\item capital letters classical registers (i.e., random variables) and quantum registers.
\item a subscript register, e.g. $X_i$, a single register with label $i$ and a superscript register, e.g. $X^i$, the sequence of registers with labels up to $i$, i.e., $X^i = X_1...X_i$.
\item the operator $\oplus$ addition modulo 2, sometimes also called the XOR operation.
\item $\mathbb{P}_{\mathcal{A}}$ the set of probability distributions over an alphabet $\mathcal{A}$.
\end{itemize}

\subsection{Quantum mechanics}
We introduce the concepts of quantum mechanics that we use throughout our work.
For a more detailed view on quantum mechanics in quantum information theory we refer to Nielsen and Chuang~\cite{nielsen2010quantum}.

A state of a quantum mechanical system can generally be described by density operators.
\begin{defi}[Density operator]
\label{def:density-operator}
	A density operator $\rho$ on a Hilbert space $\mathcal{H}$ is a normalized positive operator on $\mathcal{H}$, i.e., $\rho \geq 0$ and $\tr{\rho} = 1$.
	A density operator is said to be pure if it has the form $\rho = \ketbra{\psi}{\psi}$, where, using Dirac notation, $\ket{\psi} \in \mathcal{H}$.
\end{defi}
\noindent 
A bipartite quantum state on two Hilbert spaces $\mathcal{H}_{A}$ and $\mathcal{H}_{B}$ is described by a density operator $\rho_{AB}$ on the Hilbert space $\mathcal{H}_{A} \otimes \mathcal{H}_{B}$.
If we want to recover the state on $\mathcal{H}_{A}$ alone we take the partial trace, $\rho_{A} = \tr_{B}(\rho_{AB}) = \sum_{b} (\mathrm{id}_{A} \otimes \bra{b}) \rho_{AB} (\mathrm{id}_{A} \otimes \ket{b})$, where $\{\ket{b}\}_{b}$ is an orthonormal basis (ONB) on $\mathcal{H}_{B}$.

Some special density operators are given in the following.
\begin{enumerate}[(i)]
	\item The density operator is said to be fully mixed if $\rho = \frac{1}{d} \mathrm{id}$, where $d = \mathrm{dim}(\mathcal{H})$.
	\item The density operator $\rho_{XA}$ is said to be a classical-quantum state (cq-state) if $\rho = \sum_{i=1}^{d} p_{i} \ketbra{i}{i} \otimes \rho_{A}^{i}$, where $\{\ket{i}\}_{i}$ is an ONB on a $d$-dimensional Hilbert space and $\sum_{i=1}^{d} p_{i} = 1$ with $p_{i} \geq 0 \, \forall \, i$.
	The notion can be extended to an arbitrary amount of classical registers.
\end{enumerate}

We describe the evolution of a quantum state by completely positive trace preserving (CPTP) maps.
\begin{defi}[CPTP map]
\label{def:CPTPM}
	A linear map $\mathcal{E} \in \mathrm{Hom}(\mathrm{End(\mathcal{H}_{A})}, \mathrm{End(\mathcal{H}_{B})})$ is said to be trace preserving if, for any $\rho \in \mathrm{End(\mathcal{H}_{A})}$,
	\begin{equation*}
		\tr \left(\mathcal{E}(\rho) \right) = \tr (\rho) \,.
	\end{equation*}
	The map $\mathcal{E}$ is said to be completely positive if, for any $\rho_{AR} \in \mathrm{End(\mathcal{H}_{A} \otimes \mathcal{H}_{R})}$ and $\rho_{AR} \geq 0$,
	\begin{equation*}
		(\mathcal{E} \otimes \mathcal{I}_{R}) (\rho_{AR}) \geq 0 \,,
	\end{equation*}
	where $\mathcal{H}_{R}$ is any additional Hilbert space and $\mathcal{I}_{R}$ is the identity map on that Hilbert space.
\end{defi}

When talking about the closeness of quantum states we quantify it by the trace distance which describes how well two states can be distinguished.
\begin{defi}[Trace distance]
\label{def:trace-distance}
	The trace distance between two density operators $\rho$ and $\sigma$ on a Hilbert space $\mathcal{H}$ is defined as
	\begin{equation*}
		\delta(\rho, \sigma) = \frac{1}{2} \| \rho - \sigma \|_{1} = \frac{1}{2} \tr \left( \sqrt{(\rho-\sigma)^{\dagger} (\rho-\sigma)} \right) \,.
	\end{equation*}
\end{defi}

\subsection{Entropies and Markov chains}
\label{sec:entropies-markov}

\textbf{Entropies} We make use of the Shannon entropy for classical random variables~\cite{shannon} and its quantum equivalent, the von Neumann entropy~\cite{vonneumann}.
The conditional Shannon entropy is defined as follows.

\begin{defi}[Shannon entropy]
\label{def:shannon-entropy}
	Let $X,Y$ be discrete random variables over the alphabets $\mathcal{X}, \mathcal{Y}$ distributed according to the probability distribution $P_{XY}$.
	Then the conditional Shannon entropy is defined as
	\begin{equation*}
		H(X|Y) = - \sum_{x \in \mathcal{X}, y \in \mathcal{Y}} P_{XY}(x,y) \log_{2}P_{X|Y=y}(x) \,.
	\end{equation*}
\end{defi}

Its quantum equivalent, the von Neumann entropy, is defined for a quantum state $\rho_{AE}$.

\begin{defi}[von Neumann entropy]
\label{def:vonNeumann-entropy}
	Let $\mathcal{H}_{A}$ and $\mathcal{H}_{E}$ be two Hilbert spaces and $\rho_{AE}$ a quantum state on $\mathcal{H}_{A} \otimes \mathcal{H}_{E}$.
	Then the von Neumann entropy is defined as 
	\begin{equation*}
		H(AE)_{\rho_{AE}} = - \tr{\left( \rho_{AE} \log \rho_{AE} \right)} \,.
	\end{equation*}
	Furthermore, the conditional von Neumann entropy is defined as
	\begin{equation*}
		H(A|E)_{\rho_{AE}} = H(AE)_{\rho_{AE}} - H(E)_{\rho_{AE}} \,.
	\end{equation*}
\end{defi}

Furthermore we employ the (smooth) min-entropy, both in the classical and in the quantum case.
The (smooth) min-entropy, was introduced by Renner~\cite{RennerPhD}, for a classical quantum state.

\begin{defi}[Min-entropy]
\label{def:min-entropies}
	Let $\mathcal{H}_{A}$ and $\mathcal{H}_{E}$ be two Hilbert spaces and $\rho_{AE} = \sum_{a} p_{a} \ketbra{a}{a} \otimes \rho_{E}^{a}$ a classical quantum state on $\mathcal{H}_{A} \otimes \mathcal{H}_{E}$. Then the conditional min-entropy is defined as
	\begin{equation*}
		H_{\mathrm{min}}(A|E) = -\log p_{\mathrm{guess}}(A|E) \,,
	\end{equation*}
	where $p_{\mathrm{guess}}(A|E)$ is the maximal probability of guessing $A$ given the quantum system $E$
	\begin{equation*}
		p_{\mathrm{guess}}(A|E) = \max_{\{M_{E}^{a}\}_{a}} \left| \sum_{a} p_{a} \tr{\left( M_{E}^{a}\rho_{E}^{a} \right)} \right| \,.
	\end{equation*}
	The maximization ranges over all sets of POVMs $\{M_{E}^{a}\}_{a}$ on $E$.
\end{defi}

The smooth min-entropy is a smoothed version of the min-entropy, meaning it is the maximum of the min-entropy in an $\varepsilon$-neighbourhood around the probability distribution or quantum state.

\begin{defi}[Smooth min-entropy]
\label{def:smooth-min-entropy}
	Let $\mathcal{H}_{A}$ and $\mathcal{H}_{E}$ be two Hilbert spaces and $\rho_{AE} = \sum_{a} p_{a} \ketbra{a}{a} \otimes \rho_{E}^{a}$ a classical quantum state on $\mathcal{H}_{A} \otimes \mathcal{H}_{E}$. Then the conditional smooth min-entropy is defined as
	\begin{equation*}
		H_{\mathrm{min}}^{\varepsilon}(A|E)_{\rho_{AE}} = \max_{\sigma_{AE} \in \mathcal{B}^{\varepsilon}(\rho_{AE})} H_{\mathrm{min}}(A|E)_{\sigma_{AE}} \,,
	\end{equation*}
	where $\mathcal{B}^{\varepsilon}(\rho_{AE})$ is the set of sub-normalised states that are at most $\varepsilon$ away from $\rho_{AE}$ in terms of purified distance (see~\cite{tomamichel2010entropyduality}).
\end{defi}
When the quantum state is clear from the context we drop the subscript of the entropies and simply write $H(A|E)$ instead of $H(A|E)_{\rho_{AE}}$.

The mutual information $I(X:Y|Z)$ quantifies the common information of $X$ and $Y$, given $Z$ and can be described as a function of the entropies of the parts.

\begin{defi}[Mutual information]
\label{def:mutual-info}
	Let $X,Y$ and $Z$ be random variables.
	Then the Shannon mutual information is defined as
	\begin{equation*}
		I(X:Y|Z) = H(X|Z) - H(X|YZ) \,.
	\end{equation*}
	
	In the quantum case, let $\rho_{XYZ}$ be a quantum state.
	Then the quantum mutual information s defined as
	\begin{equation*}
		I(X:Y|Z) = H(X|Z) - H(X|YZ) \,.
	\end{equation*}
\end{defi}

\begin{defi}[Markov chain]
	A set of random variables $X,Y,X$, or a tripartite quantum state $\rho_{XYZ}$, is said to form a (quantum) Markov chain, denoted by $X \leftrightarrow Y \leftrightarrow Z$, if the conditional mutual information $I(X:Z|Y)$ vanishes.
\end{defi}

\subsection{Weak sources of randomness}
\label{sec:weak-sources}

We consider two classes of weak random sources, an SV sources and a min-entropy source.
The SV source was first introduced by Santha and Vazirani~\cite{SV}.
Formally an SV source is defined as follows.

\begin{defi}[$\mu$-SV source, \cite{SV}]
\label{def:SV}
	Let $S$ be any source producing a sequence of binary random variables $X_i$ that can depend on some side information $\lambda$.
	Then, for any $\mu \in ( 0, \frac{1}{2} )$, $S$ is called an $\mu$-SV source if the random variables $X_{i}$ are distributed according to some probability distribution $P_{X_i|X^{i-1}, \lambda}$ that depends on $\lambda$ and satisfies
	\begin{equation}
		\frac{1}{2} - \mu \leq P_{X_i|X^{i-1}, \lambda}(0|x^{i-1}) \leq \frac{1}{2} + \mu \quad \forall i,x^{i-1} \, .
	\end{equation} 
\end{defi}
\noindent We see that an SV source produces bits that are all, to some extent, random, even given the previous bits and some possible side information.

An MDL source produces two bits at a time and bounds the probability of each outcome in a similar way as the SV source.

\begin{defi}[$\mu$-MDL source, \cite{MDL}]
\label{def:MDL-source}
	Let $S$ be any source producing binary random variables $X_{i}$ and $Y_{i}$ that can depend on some side information $\lambda$.
	Then, for any $\mu = \{\mu_\mathrm{min}, \mu_\mathrm{max}\} \in ( 0, \frac{1}{4} ) \times ( \frac{1}{4}, 1 )$, $S$ is called a $\mu-\mathrm{MDL}$ source if the outputs are distributed according to some probability distribution $P_{X_{i}Y_{i}|X^{i-1}Y^{i-1}, \lambda}$ 
	that depends on $\lambda$ and satisfies
	\begin{equation}
		\mu_\mathrm{min} \leq P_{X_{i}Y_{i}|X^{i-1}Y^{i-1}, \lambda}(x_{i}y_{i}|x^{i-1}y^{i-1}) \leq \mu_\mathrm{max} \quad \forall x^{i},y^{i} \, . \label{eq:MDL-source}
	\end{equation}	
\end{defi}

In our work we us the notation of MDL sources. These are directly related to the SV sources as shown below.

\begin{lma}
	For all $0 \leq \mu \leq \nicefrac{1}{2}$ a $\mu$-SV source is a $\left\{ \left( \frac{1}{2} - \mu \right)^2, \left( \frac{1}{2} + \mu \right)^2 \right\}$-MDL source. \label{lma:SV-MDL}
\end{lma}
\begin{proof}
	Employing the definition of conditional probabilities $P_{X_i|X^{i-1} , \lambda} = \frac{P_{X^i , \lambda}}{P_{X^{i-1} , \lambda}}$ and $P_{X_{i+1}|X^{i} , \lambda} = \frac{P_{X^{i+1} , \lambda}}{P_{X^{i} , \lambda}}$ we find $P_{X_{i+1} X_{i} | X^{i-1} , \lambda} = \frac{P_{X^{i+1} , \lambda}}{P_{X^{i-1} , \lambda}} = P_{X_i|X^{i-1} , \lambda} P_{X_{i+1}|X^{i} , \lambda}$. From that it follows immediately that the constraints for two consecutive outputs of the SV source are
	\begin{equation}
		\left( \frac{1}{2} - \mu \right)^2 \leq P_{X_{i+1}X_i|X^{i-1}, \lambda}(x_{i+1}x_i|x^{i-1}) \leq \left( \frac{1}{2} + \mu \right)^2 \quad \forall i,x^{i+1} \, .
	\end{equation}
	Choosing $\mu_\mathrm{min} = \left( \frac{1}{2} - \mu \right)^2$ and $\mu_\mathrm{max} = \left( \frac{1}{2} + \mu \right)^2$ this satisfies Definition~\ref{def:MDL-source} of an MDL source.
\end{proof}

Finally a min-entropy source is a source that produces a bit string that has a min-entropy which is lower bounded by some constant.

\begin{defi}[$k$-min-entropy source, \cite{min-entropy-source}]
\label{def:min-entropy-source}
	Let $S$ be any source producing a sequence of binary random variables $X_i$ that can depend on some side information $\lambda$.
	Furthermore let $n$ be the arbitrary length of that sequence.
	Then $S$ is said to be a $k$-min-entropy source if the min-entropy of the bit string conditioned on the side information is lower bounded by $k$, i.e., $H_{\mathrm{min}}(X^{n}|\lambda) \geq k$.
\end{defi}

It is worthwhile noticing that any SV source can also be used as a min-entropy source.
The reversed implication, however, is not true, since in an SV source each new bit must contain a minimal amount of randomness.
In this sense the output of the SV source has more structure.

With regards to randomness amplification it has been shown by Santha and Vazirani~\cite{SV} that, classically, a single SV-source, private or public, cannot be amplified.
If one has access to two or more independent sources of which at least one is private, one can extract randomness from them using a randomness extractor.
However, if all the sources are public this is still not possible.

\subsection{Non-local games and Bell inequalities}
\label{sec:nl-games}

\paragraph*{Non-local games.} During a non-local game two players, Alice and Bob, are given questions by a verifier and have to return answers.
Both the questions and answers can be described simply as numbers.
The questions, $x$ and $y$, are taken from alphabets (we restrict ourselves to binary alphabets) $\mathcal{X}$ and $\mathcal{Y}$, and distributed according to some probability distribution $P_{XY}$.
Similarly, the answers, $a$ and $b$, can be chosen from (binary) alphabets $\mathcal{A}$ and $\mathcal{B}$.
Alice and Bob win a round of the game if the questions and answers satisfy a previously defined condition.
Formally we can think of a function $w: \mathcal{X} \times \mathcal{Y} \times \mathcal{A} \times \mathcal{B} \rightarrow \mathcal{W}$, where $\mathcal{W}$ is the set describing the outcome of the game.

In order to win the game with the highest probability Alice and Bob can, before the game starts, choose a strategy.
After the game starts they are no longer allowed to communicate.
The rules of the game are that the players are not allowed to communicate, one player does not know the other player's question, and that the players cannot influence the questions they are asked.

In terms of strategy we distinguish two classes, the first one being classical local hidden variable (LHV)/ shared randomness strategies.
In an LHV strategy, Alice and Bob share some common information $\lambda$ and, according to the common information, choose their answers deterministically.
The second class of strategies are quantum strategies.
Using a quantum strategy, Alice and Bob can share a multipartite quantum state.
They can then use the questions to choose measurements that are done on the quantum state.
The results of these measurements can then be used to produce answers for the questions.

It can be shown that quantum strategies are strictly more powerful than LHV strategies.
Namely, for some non-local games, there exist quantum strategies that achieve a winning probability that is higher than any LHV strategy can achieve.
We call the probability distributions $P_{ABXY}$, of the questions and answers, that characterise these strategies non-local statistics.

Using this fact that strategies producing non-local statistics are more powerful than LHV strategies, we can certify quantumness using non-local games.
We can do this by analysing the winning probability of the strategy in the game.
If the winning probability is higher than the threshold for any LHV strategy we can conclude that Alice and Bob must have used a quantum strategy.

\paragraph*{Bell inequalities.} An equivalent description of non-local games are Bell inequalities.
In this scenario we consider Bell experiments; i.e., experiments where we have two devices that take inputs (the questions) and produce outputs (the answers).
The probability distributions over the inputs and outputs can, similar to the case of non-local games, be divided into LHV statistics and quantum statistics.
However, the winning probability is replaced by a Bell parameter, a general function of the probability distribution, $f(P_{ABXY})$.
The Bell inequality is then a constraint on the Bell parameter that is satisfied by all LHV statistics.
A Bell inequality could for example look as follows
\begin{equation*}
	f(P_{ABXY}) \leq f_{\mathrm{LHV}} \,.
\end{equation*}

In the Bell experiments we consider some hidden side information $\lambda$.
The assumptions that we make about the setting are that firstly, given the inputs and the side information, the outputs do not depend on each other.
Secondly, we assume that , given $x$ and $\lambda$, $a$ does not depend on $y$, and, given $y$ and $\lambda$, $b$ does not depend on $x$.
Finally we require that the questions be independent of the side information.
Given these assumptions we can, similar as with non-local games, certify quantumness by calculating the Bell parameter and comparing it to the local threshold.
If the Bell parameter exceeds the local threshold we know that the statistics must be non-local.
Statistics that are not non-local are called local.
The set of local statistics is called the local polytope, $\mathcal{L}$.
We can think of the facets of the local polytope as the Bell inequalities.
If one Bell inequality is violated by $P_{AB|XY}$ the statistics lie outside of $\mathcal{L}$ and are thus non-local.
The local polytope with its facets is schematically depicted in Figure~\ref{fig:local_polytope}.

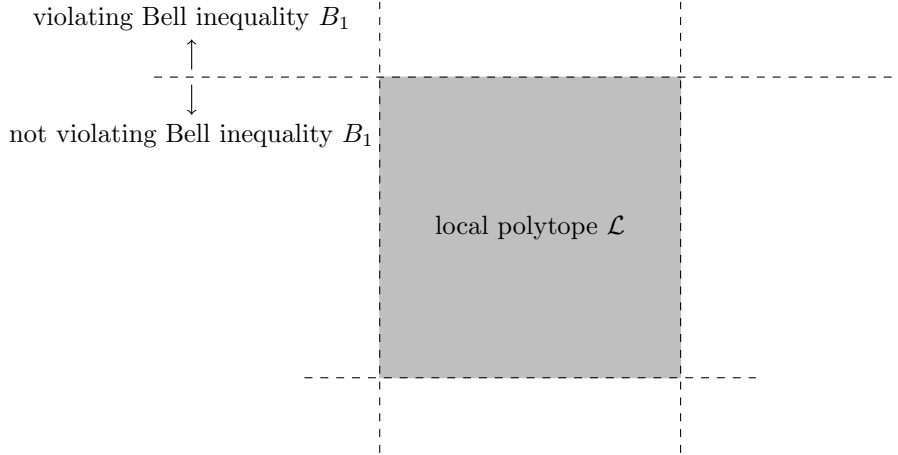
\begin{figure}
	\centering
	\begin{tikzpicture}
		\filldraw[lightgray] (0,0) rectangle (4,4);
		\node at (2,2) {local polytope $\mathcal{L}$};
		\draw[dashed] (-3,4) -- (7,4);
		\draw[dashed] (4,-1) -- (4,5);
		\draw[dashed] (-1,0) -- (5,0);
		\draw[dashed] (0,-1) -- (0,5);
		\draw[->] (-2.5,4.1) -- (-2.5,4.5);
		\node[above] at (-2.5,4.5) {violating Bell inequality $B_1$};
		\draw[->] (-2.5,3.9) -- (-2.5,3.5);
		\node[below] at (-2.5,3.5) {not violating Bell inequality $B_1$};
	\end{tikzpicture}
	\caption{Schematic description of a local polytope with boundaries (facets)  described by Bell inequalities. For example the top horizontal dashed line is determined by the Bell inequality $B_1$. A violation of this Bell inequality means that the probability distribution lies above the dashed line.}
\label{fig:local_polytope}
\end{figure}

\subsubsection{The CHSH game}

As an example of a non-local game we consider the CHSH game.
The winning function for the game is
\begin{align*}
	w: \mathcal{X} \times \mathcal{Y} \times \mathcal{A} \times \mathcal{B} &\rightarrow \{0,1\} \\
	(x,y,a,b) &\mapsto \begin{cases} 1 &\text{if } a\oplus b = x \wedge y \\ 0 &\text{otherwise} \end{cases} \,,
\end{align*}
meaning the game is won, if and only if the questions and answers satisfy $a\oplus b = x \wedge y$.
It can be shown that, if the questions are uniformly distributed, no classical strategy can achieve a winning probability higher than $p_{\mathrm{win}} = \frac{3}{4}$.
However, if Alice and Bob share a maximally entangled state and do measurements according to the questions and use the outputs as answers, they can achieve a winning probability $p_{\mathrm{win}} = \frac{2 + \sqrt{2}}{4}$.

The CHSH game is the game corresponding to the CHSH inequality~\cite{CHSH},
\begin{equation}
	\beta \equiv \sum_{a,b,x,y \in \{0,1\}} (-1)^{a+b+xy} P_{AB|XY}(ab|xy) \leq 2 \,. \label{eq:CHSH}
\end{equation}
An equivalent version of the CHSH inequality (while enforcing non-signalling condition) is
\begin{equation}
	\alpha \equiv P_{AB|XY}(00|00) - \big( P_{AB|XY}(01|01) + P_{AB|XY}(10|10) + P_{AB|XY}(00|11) \big) \leq 0 \label{eq:Eberhard} \,.
\end{equation}
This inequality was first introduced by Eberhard~\cite{EberhCHSH}.
Within quantum mechanics we can have non-local values $\beta \in [2,2 \sqrt{2}]$ and $\alpha \in [0,\frac{\sqrt{2}-1}{2}]$.
Thus, given the affine relation between the two values we find the relation
\begin{equation}
	\beta = 4 \alpha + 2 \quad \Leftrightarrow \quad \alpha = \frac{\beta}{4} - \frac{1}{2} \,. \label{eq:relation_alpha_beta}
\end{equation}

\subsection{Measurement dependent locality}
\label{sec:MDL}

In standard non-local games we usually assume that the questions are uniformly distributed and cannot be influenced by anyone.
This assumption is called measurement independence.
P\"utz \textit{et.\@ al}~\cite{MDL} weakened the assumption of measurement independence to an assumption of limited measurement dependence, where Eve can influence the distribution of the questions to some extent, and studied Bell inequalities in this scenario.
A schematic drawing of the setting in this scenario is shown in Figure~\ref{fig:intro_mdl}.
The inputs can now depend on some hidden information and need not be uniform anymore.
The way Eve can influence the distribution of the questions is described by an MDL source (Definition~\ref{def:MDL-source}).
The main result of their work is that we can verify the usage of quantum strategies for any amount of measurement dependence, as long as $\mu_{\mathrm{min}} > 0$.

In order to verify quantum strategies with an MDL source, we need a new Bell inequality, an MDL inequality~\cite{MDL}
\begin{equation}
	S_{\mu} \equiv \mu_\mathrm{min} P_{ABXY}(0000) - \mu_\mathrm{max} \big( P_{ABXY}(0101) + P_{ABXY}(1010) + P_{ABXY}(0011) \big) \leq 0 \, . \label{eq:MDL_ineq}
\end{equation}
Using this inequality we verify quantum strategies if $S_{\mu} \geq 0$.
Furthermore we now call statistics that do not violate Equation~\eqref{eq:MDL_ineq} measurement dependent local (MDL).
This MDL inequality translates into a game with winning function
\begin{align*}
	w(a,b,x,y) = \, \begin{cases} \mu_{\mathrm{min}} & \text{if } (a,b,x,y) = (0,0,0,0) \\ 
    -\mu_{\mathrm{max}} & \text{if } (a,b,x,y) \in \{(0,1,0,1), (1,0,1,0), (0,0,1,1)\} \\
	0 & \text{otherwise} \,.\end{cases} 
\label{eq:winning_func}
\end{align*}

\subsection{Untrusted device}
\label{sec:untrusted-devices}
In our randomness amplification protocol we use two separated untrusted devices to play a non-local game.
Untrusted in this context means that we assume that the adversary produces the devices and can produce them (almost) anyway she wants.
However, we enforce the condition that we can use the device to play a two-player non-local game with binary inputs and outputs; i.e., upon receipt of a binary input, the devices produce a binary output.
This condition can be easily checked during the execution of the protocol.
If the devices do not produce outputs or produce outputs that are not binary we can simply abort the protocol.

Moreover, we assume that quantum mechanics is complete. 
Thus we can model the inner workings of the device as doing measurements on an unknown quantum state.
The measurements can depend on the inputs and the outputs can depend on the outcome of the measurement.
If the devices are used sequentially in a number of rounds like in our protocol, the measurements can be different in each round.
In addition the new quantum state on which the measurements are done can depend on previous rounds.

In a device-independent adversarial scenario we play the non-local game to verify the quantumness of the inner workings of the devices.
Hence we can think of the adversary implementing a strategy, i.e., a specific set of states and measurements, in the device such that she gains a maximal amount of knowledge of the outputs.
This strategy also includes her attempt to trick us into thinking that the devices produce non-local statistics whereas they are not.
Since the adversary is also assumed to be the manufacturer of the devices she can build a third device that contains a purification of the quantum states in the two other devices.

\subsection{Quantum-proof randomness extractors in the Markov model}
\label{sec:extractors}

A (classical) two-source extractor is defined as follows.
\begin{defi}[Two-source extractor, \cite{raz05extractors}]\label{def:two-source}
	A function $\mathrm{Ext}:\{0,1\}^{n_1}\times\{0,1\}^{n_2}\rightarrow\{0,1\}^m$ is called a $(k_1,k_2,\varepsilon)$ two-source extractor if for any two independent sources $X_1,X_2$ with $H_{\text{min}}\left(X_1\right)\geq k_1$ and $H_{\text{min}}\left(X_2\right)\geq k_2$, we have 
	\[
		\frac{1}{2}\| \rho_{\mathrm{Ext}(X_1,X_2)} - \rho_{U_m}  \| \leq \varepsilon \;,
	\]
	where $\rho_{U_m}$ is the fully mixed state on a system of dimension $2^m$. $\mathrm{Ext}$ is said to be \emph{strong in the $i$'th input} if
  	\[
		\frac{1}{2} \|  \rho_{\mathrm{Ext}(X_1,X_2)X_i} - \rho_{U_m} \otimes \rho_{X_i} \| \leq \varepsilon \;. 
	\]
	 If $\mathrm{Ext}$ is not strong in any of its inputs it is said to be weak. 
\end{defi}

In our work we use extractors that work in the presence of quantum side information described by the Markov model introduced in~\cite{Extractors}.
In the Markov model we assume that the two sources of a two-source extractor together with the side information $C$ form a Markov chain: $I(X_{1} : X_{2} | C) = 0$ (where $X_1$ and $X_2$ are classical registers, while $C$ can hold a quantum state).
This can be interpreted as the requirement that, given the side information, the two sources are independent.
Formally the quantum Markov model and a quantum-proof two-source extractor in the Markov model are defined as follows.

\begin{defi}[Quantum Markov model, \cite{Extractors}]
\label{def:quantum_markov_mutual_info}
  	A ccq-state \(\rho_{X_1X_2C}\) is said to belong to the Markov model if $X_1 \leftrightarrow C \leftrightarrow X_2$ is a Markov chain (i.e., $I(X_1:X_2|C) = 0$).
\end{defi}

\begin{defi}[Strong quantum-proof two-source extractor in the Markov model, \cite{Extractors}]
\label{def:quantum-two-source.mc}
	A function $\mathrm{Ext}: \{0,1\}^{n} \times \{0,1\}^{d} \to \{0,1\}^m$ is a $(k_1,k_2,\varepsilon)$ quantum-proof two-source extractor in the Markov model, strong in the second source, if for all sources $X_1,X_2$, and quantum side information $C$, where $X_1 \leftrightarrow C \leftrightarrow X_2$ form a Markov chain, and with min-entropy $H_{\text{min}}\left(X_1|C\right)\geq k_1$ and $H_{\text{min}}\left(X_2|C\right)\geq k_2$, we have
	\begin{equation*}
	 	\frac{1}{2}\| \rho_{\mathrm{Ext}(X_1,X_2)X_2C} - \rho_{U_m} \otimes \rho_{X_2C} \| \leq \varepsilon \;. 
	\end{equation*}
	where $\rho_{\mathrm{Ext}(X_1,X_2)C} = \mathrm{Ext} \otimes \mathcal{I}_C \rho_{X_1X_2C}$ and $\rho_{U_m}$ is the fully mixed state on a system of dimension $2^m$. 
\end{defi}

The main result of~\cite{Extractors} is that any (classical) two-source extractor is also quantum-proof in the Markov model:
\begin{lma}\label{lem:quantum_markov_two_source}
	Any $(k_1,k_2,\varepsilon)$-[strong] two-source extractor is a $\left(k_1 + \log \frac{1}{\varepsilon}, k_2 + \log\frac{1}{\varepsilon},\sqrt{3\varepsilon \cdot 2^{(m-2)}}\right)$-[strong] quantum-proof two-source extractor in the Markov model, where $m$ is the output length of the extractor.
\end{lma}

In this work we use such an extractor, but for a source with a lower bound on the smooth min-entropy rather than the min-entropy itself. The effect of this on the parameters of the extractor was also investigated in~\cite{Extractors}. We use the following form of the statement:
\begin{lma}
\label{lma:smooth-entropy-bound}
	Let $\mathrm{Ext} : \{0,1\}^{n} \times  \{0,1\}^{d} \to \{0,1\}^m$ be a $(k_1,k_2,\varepsilon)$ quantum-proof two-source extractor in the Markov model, strong in the source $X_i$. Then for any Markov state $\rho_{X_1X_2C}$ with $H_{\mathrm{min}}^{\varepsilon_s}(X_1|C)_{\rho} \geq k_1+\log(1/\varepsilon)+1$ and $H_{\mathrm{min}}(X_2|C)_{\rho} \geq k_2+\log(1/\varepsilon)+1$,
	\[
		\frac{1}{2} \Vert \rho_{\mathrm{Ext}(X_1,X_2)X_iC} - \rho_{U_m} \otimes \rho_{X_iC} \Vert \leq 6  \left(\varepsilon_s + \varepsilon\right) \;.
	\]
\end{lma}

\subsection{The entropy accumulation theorem}
\label{sec:EAT}

The entropy accumulation theorem (EAT), introduced in~\cite{EAT}, gives a straightforward way of bounding the smooth min-entropy of a system consisting of $n$ random variables that possibly depend on each other. 
For our work the simplified versions of the definitions and theorems of~\cite{EAT}, as presented in~\cite{RotemEAT}, suffice.
In the following we introduce the definitions and theorems which are crucial to working with the EAT.

\begin{defi}[EAT channels]
\label{def:EATchannel}
	EAT channels $\mathcal{N}_i: R_{i-1} \rightarrow R_iA_iB_iI_iC_i$, for $i \in [n]$, are CPTP maps such that for all $i \in [n]$:
	\begin{enumerate}
	\item $A_i,B_i,I_i$ and $C_i$ are finite-dimensional classical systems (RV). $A_i$ and $B_i$ are of dimension $d_{A_i}$ and $d_{B_i}$ respectively. $R_i$ are arbitrary quantum registers.
	
	\item For any input state $\sigma_{R_{i-1}R'}$, where $R'$ is a register isomorphic to $R_{i-1}$, the output state $\sigma_{R_iA_iB_iI_iC_iR'} = \left( \mathcal{N}_i \otimes \mathcal{I}_{R'} \right) \left( \sigma_{R_{i-1}R'} \right)$ has the property that the classical value $C_i$ can be measured from the marginal $\sigma_{A_iB_iI_i}$ without changing the state.
	
	\item For any initial state $\rho^0_{R_0E}$, the final state $\rho_{A^nB^nI^nC^nE} = \left( \textrm{Tr}_{R_n} \circ \mathcal{N}_n \circ \dots \circ \mathcal{N}_1 \right) \otimes \mathcal{I}_E \rho^0_{R_0E}$ satisfies the Markov chain condition $A^{i-1}B^{i-1} \leftrightarrow I^{i-1}E \leftrightarrow I_i$ for each $i \in [n]$.
	\end{enumerate}
\end{defi}

\begin{defi}[Min-tradeoff function]
\label{def:min_tradeoff}
	Let $\mathcal{N}_1,\ldots,\mathcal{N}_N$ be a family of EAT channels. Let $\mathcal{C}$ denote the common alphabet of $C_1,\ldots,C_n$. 	
	A function $f_{\min}$ from $\mathbb{P}_{\mathcal{C}}$ to the real numbers is called a \emph{min-tradeoff function} for $\{\mathcal{N}_i\}$ if it satisfies
	\[
		f_{\min}(p) \leq \inf_{\sigma_{R_{i-1}R'}:\mathcal{N}_i(\sigma)_{C_i}=p} H\left( A_i B_i | I_i R' \right)_{\mathcal{N}_i(\sigma)} \;
	\]
	for all $i\in [n]$, where the infimum is taken over all input states of $\mathcal{N}_i$ for which the marginal on $C_i$ of the output state is the probability distribution $p$, and the infimum over the empty set is defined as plus infinity.
\end{defi}

\begin{defi}
\label{def:freq}
	Let $C^{n}$ be a set of random variables over the alphabet $\mathcal{C}$. Then $\mathrm{freq}_{C^{n}}$ defines the probability distribution over $\mathcal{C}$ defined by $\mathrm{freq}_{C^{n}}(x) = \frac{|\{ i \in \{1, \dots, n\} : C_{i} = x \}|}{n}$.
\end{defi}

\begin{defi}[Infinity norm]
\label{def:inf-norm}
	Let $f: \Omega \rightarrow \mathbb{R}$ be a function over some set $\Omega \subset \mathbb{R}^{m}$. Then the infinity norm of the gradient of $f$ is defined as
	\begin{equation*}
		\|  \nabla f \|_\infty = \sup \left\{ \frac{\partial}{\partial x_{i}} f(\textbf{x}) : \textbf{x} \in \Omega, \, i \in \{1,\dots,m\} \right\} \,.
	\end{equation*}
\end{defi}

\begin{thm}[EAT,~\cite{EAT}]
\label{thm:EAT}
	Let $\mathcal{N}_i:R_{i-1}\rightarrow R_i A_i B_i I_i C_i$ for $i\in [n]$ be EAT channels as in Definition~\ref{def:EATchannel}, $\rho_{A^nB^nI^nC^nE} = \left( \tr_{R_n}\circ \mathcal{N}_n \circ \dots \circ \mathcal{N}_1\right) \otimes \mathcal{I}_E \; \rho_{R_0E}$ be the final state, 
	$\Omega$ an event  defined over $\mathcal{C}^n$, $p_\Omega$ the probability of $\Omega$ in $\rho$, 
	and $\rho_{|\Omega}$ the final state conditioned on $\Omega$. Let $\varepsilon_{\text{s}} \in (0,1)$.
	 
	For $f_{\min}$ a min-tradeoff function for $\{\mathcal{N}_i\}$, as in Definition~\ref{def:min_tradeoff}, and any $t\in \mathbb{R}$ such that $f_{\min}\left( \mathrm{freq}_{C^n} \right) \geq t$ for any $C^n\in\mathcal{C}^n$ for which $\Pr\left[C^n\right]_{\rho_{|\Omega}}> 0$,
	\[
		H_{\min}^{\varepsilon_{\text{s}}} \left( A^nB^n|I^nE \right)_{\rho_{|\Omega}} > n t - v\sqrt{n} \;,
	\]
	where $v = 2\left(\log(1+2 d_{A_iB_i} ) + \lceil \|  \nabla f_{\min} \|_\infty \rceil \right)\sqrt{1-2\log (\varepsilon_{\text{s}} \cdot p_\Omega)}$ and $d_{A_iB_i}$ denotes the dimension of $A_iB_i$.
\end{thm}

To gain some intuition regarding the EAT we now give a short explanation of how it is used below. The concrete and formal details are given in the following sections. Our EAT-channels are chosen to be the channels describing the actions in each step of the protocol (both of the honest parties and the uncharacterised quantum device). The event $\Omega$ is the event of not aborting the protocol. $\rho_{|\Omega}$ is hence the state in the end of the protocol conditioned on not aborting. The goal is then to lower-bound the conditional smooth min-entropy of this state and this is exactly what Theorem~\ref{thm:EAT} gives us. The first order term in the given bound is $n t$ where $n$ is the number of rounds of the protocol and $t$ is the minimal amount of entropy accumulated in each step, quantified using the min-tradeoff function. In Section~\ref{sec:single-round} we make the relevant analysis to find the value of $t$.

\section{Secret randomness from a single round}
\label{sec:single-round}

In this section we quantify the randomness of the outputs of an MDL experiment.
With that achieved we can carry on in Section~\ref{sec:RAP} to quantify the randomness of the outputs in a sequence of MDL experiments.
Hence quantifying the randomness in a single MDL experiment is crucial in our process of producing an arbitrary amount of randomness.

In our single MDL experiment we consider a device consisting of two separated components, such that one can enforce a situation in which the ``non-signalling conditions'' between the components hold. (i.e., the two components cannot signal, or communicate, with one another).
During the execution of the experiment the two operators of the device, Alice and Bob, draw inputs, $X$ and $Y$, from the MDL source.
They then feed the inputs to their component and record the output, $A$ and $B$, that it generates.
As noted in Section~\ref{sec:untrusted-devices}, we consider a third party that can hold a purification, $E$, of the quantum state in the device.
We thus want to quantify the randomness of $A$ and $B$ given $X$, $Y$, and $E$.
An algorithmic description of the MDL experiment is given in Protocol~\ref{alg:MDL-experiment}.
In Step~\ref{step:MDL-exp-symm} we use a uniform and independent random bit $F$ to symmetrise the outputs.
Of course, in the context of randomness amplification we cannot do this.
Nevertheless, we use this just as a step in the proof and later argue that the symmetrisation step can be dropped in practice. 

Formally we choose to quantify the randomness by the von Neumann entropy, $H(AB|XYE)$.
The remaining part of this section is dedicated to proving the following bound on this entropy. 
\begin{lma}
\label{lma:holevo_bound}
	Consider the MDL experiment described in Protocol~\ref{alg:MDL-experiment} where both the inputs and the outputs are binary.
	Then, for a state and a set of measurements (i.e., strategy of the adversary) yielding a violation $S_\mu > 0$ of Inequality~\eqref{eq:MDL_ineq}, the bound
	\begin{equation}
		H(AB | XYE) \geq 1 - h \left( \frac{1}{2} + \frac{1}{\mu^*} \sqrt{ S_\mu (S_\mu + \mu^*)} \right) \label{eq:single_round_bound}
	\end{equation}
	on the von Neumann entropy of the outputs holds, where $\mu^{*} = \mu_{\mathrm{min}} \cdot \mu_{\mathrm{max}}$.
\end{lma}
\noindent We prove the lemma by employing the bound on the Holevo quantity (introduced later) that Pironio \textit{et.\@ al} derived in~\cite{Pironio} for the CHSH game.
We adapt the bound to the MDL game with biased inputs.

\begin{algorithm}[t]
	\floatname{algorithm}{Protocol}
	\raggedright
	\caption{Execution of an MDL experiment}
	\label{alg:MDL-experiment}
	\begin{algorithmic}[1]
		\STATEx \textbf{Arguments:} 
		\STATEx\hspace{\algorithmicindent} $M(\mu)$ -- $\mu$-MDL source
		\STATEx\hspace{\algorithmicindent} $D$ -- untrusted device of two components
	
		\STATEx
	
		\STATE Alice and Bob choose inputs from the MDL source with parameters $\mu$.
		\STATE Alice and Bob use $D$ with $X,Y$ and record their outputs as $A$, $B$.
		\STATE [optional] Alice and Bob choose a uniform and independent binary random variable $F$ and update their outputs as $\tilde{A} = A \oplus F$ and $\tilde{B} = B \oplus F$. \label{step:MDL-exp-symm}
	\end{algorithmic}
\end{algorithm}

To prove  Lemma~\ref{lma:holevo_bound} we first express the entropy in terms of the Holevo quantity, $\chi (\tilde{A} : F E| X = x)$, similarly to what was done as in~\cite{RotemEAT}. The expression is given in the following lemma. 
The proof is given in Appendix~\ref{apx:proofs}.

\begin{lma}
\label{lma:rewrite}
	In an MDL experiment with binary inputs and outputs, as described in Protocol~\ref{alg:MDL-experiment}, with two devices, between which the non-signaling condition holds, the entropy of the outputs can be lower bounded as
	\begin{equation}
		H(\tilde{A} \tilde{B}|XYFE) \geq \sum_x \mathrm{Pr}[X = x] \cdot \big( 1 - \chi (\tilde{A} : F E| X = x) \big ) \,,
	\end{equation}
	where  $\chi (\tilde{A} : F E| X = x) = H(FE|X=x) - H(FE|\tilde{A},X=x)$ is the Holevo quantity.
\end{lma}

We now proceed to prove Lemma~\ref{lma:holevo_bound}.
\begin{proof}[Proof of Lemma~\ref{lma:holevo_bound}]
	In the proof of our claim we first prove an upper bound on the Holevo quantity of the symmetrized outputs as a function of the MDL violation.
	Once the upper bound on the Holevo quantity is derived we make use of Lemma~\ref{lma:rewrite} and derive the lower bound on the von Neumann entropy of the symmetried outputs.
	Finally we argue why the entropy bound for the symmetrized outputs is also an entropy bound for the unsymmetrized outputs.
	
	In order to upper bound the Holevo quantity we start with the bound that was derived in~\cite[Equation (11)]{Pironio} for the standard CHSH scenario.
	Together with the relation in Equation~\eqref{eq:relation_alpha_beta} we find
	\begin{equation}
	\chi (\tilde{A} : E F| X = x) \leq h \left( \frac{1}{2} + \frac{1}{2} \sqrt{\frac{\beta^2}{4} - 1} \right) = h \left( \frac{1}{2} + \sqrt{\alpha (\alpha + 1)} \right) \,, \label{eq:holevo-bound-CHSH}
	\end{equation}
	where $\beta$ is the violation of the CHSH inequality and $\alpha$ is the violation of the Eberhard inequality.
	
	Continuing we relate this bound to our scenario where the inputs for the Bell measurements are not uniform and not independent.
	To that end we consider two processes producing different states.
	In the first process we consider an MDL source that produces biased bits that might be correlated with some side information $\lambda$.
	In the second process we consider an input source that produces uniform and independent bits.
	We want to relate the Holevo quantity of the outputs in both processes.
	
	In both processes we can describe the generation of the outputs as doing a measurement on an unknown quantum state
	\begin{equation}
		\rho_{Q_AQ_BE, \lambda} \,, \label{eq:unknown-state-pre}
	\end{equation}
	where $Q_{A}$ and $Q_{B}$ are the quantum registers in Alice's and Bob's device respectively and $E$ is the quantum side information that the adversary holds.
	The state can also depend on the classical side information that Eve has.
	The specific measurements can depend on the inputs, $X$ and $Y$, and the classical side information $\lambda$.
	We also include the uniform and independent random variable $F$ in the measurements.
	This variable then determines whether the outputs are being flipped or not, as described in Step~\ref{step:MDL-exp-symm}.
	More precisely we describe the measurement, implemented with a strategy $\lambda$ for specific $x,y,f$, as a CPTP map evolving the unknown quantum state
	\begin{equation}
	\begin{aligned}
		\mathcal{E}_{xyf, \lambda} : Q_{A}Q_{B}E &\rightarrow Q_{A}Q_{B}\tilde{A} \tilde{B} E \label{eq:MDL-CPTPM} \\
		\rho_{Q_AQ_BE, \lambda} &\mapsto \mathcal{E}_{xyf, \lambda}(\rho_{Q_AQ_BE, \lambda}) \,.
	\end{aligned}
	\end{equation}
	Note that the two parts of the device and the adversary are spatially separated and thus the CPTP map factors into three parts,
	\begin{align*}
		\mathcal{E}_{xyf, \lambda} = \mathcal{E}^{A}_{xf, \lambda} &\otimes \mathcal{E}^{B}_{yf, \lambda} \otimes \mathcal{I}_{E} \\
		\mathcal{E}^{A}_{xf, \lambda}: \; Q_{A} &\rightarrow Q_{A} \tilde{A} \\
		\mathcal{E}^{B}_{yf, \lambda}: \; Q_{B} &\rightarrow Q_{B} \tilde{B} \,,
	\end{align*}
	where $\mathcal{I}_{E}$ is the identity map on the adversary's quantum register.
	
	\textbf{Process 1} is associated to the measurements in our MDL scenario.
		First we choose inputs $X, Y \in \{0,1\}$ according to a distribution $P_{XY|\lambda}$ satisfying Definition~\ref{def:MDL-source}.
		Furthermore we also choose an independent and uniform random variable $F \in \{0,1\}$ for the symmetrisation of the outputs.
		For a specific strategy $\lambda$ of the adversary the post measurement state is
		\begin{equation*}
			\rho^1_{Q_{A}Q_{B} \tilde{A} \tilde{B} XYFE, \lambda} = \sum_{x, y} P_{XY|\lambda}(x,y) \sum_{f} \frac{1}{2} \, (\mathcal{E}^{A}_{xf, \lambda} \otimes \mathcal{E}^{B}_{yf, \lambda} \otimes \mathcal{I}_{E})(\rho_{Q_AQ_BE, \lambda}) \otimes |xyf\rangle \langle xyf| \,, \nonumber
		\end{equation*}
		where $\{|x \rangle \}_{x}, \{|y \rangle \}_{y},$ and $\{|f \rangle \}_{f}$ each form an orthonormal basis of a two dimensional Hilbert space. After tracing out the systems $Q_{A}, Q_{B}, B$ and $Y$, which are irrelevant for the calculation of $\chi$, we are left with
		\begin{equation*}
			\rho^1_{\tilde{A} XFE, \lambda} = \sum_{x} P_{X|\lambda}(x) \left( \sum_{f} \frac{1}{2} \, (\tr{}_{Q_{A}} \circ \mathcal{E}^{A}_{xf, \lambda} \otimes \mathcal{I}_{E})(\rho_{Q_AE, \lambda}) \otimes |f\rangle \langle f|  \right) \otimes |x\rangle \langle x| \,. \label{eq:rho_1}
		\end{equation*}
		We denote
		\begin{equation*}
			\rho^1_{\tilde{A} FE|X=x, \lambda} = \left( \sum_{f} \frac{1}{2} \, (\tr{}_{Q_{A}} \circ \mathcal{E}^{A}_{xf, \lambda} \otimes \mathcal{I}_{E})(\rho_{Q_AE, \lambda}) \otimes |f\rangle \langle f|  \right) \,.
		\end{equation*}
	
	\textbf{Process 2} is associated to the standard CHSH scenario.
		First we choose the inputs $X, Y \in \{0,1\}$ independent of everything else and uniformly at random.
		Then we choose an independent and uniform random variable $F \in \{0,1\}$ to symmetrise the outputs.
		Similar to Process~1, for a specific strategy $\lambda$ of the adversary, the post measurement state is
		\begin{equation*}
			\rho^2_{Q_{A}Q_{B} \tilde{A} \tilde{B} XYFE, \lambda} = \sum_{x, y} \frac{1}{4} \sum_{f} \frac{1}{2} \, (\mathcal{E}^{A}_{xf, \lambda} \otimes \mathcal{E}^{B}_{yf, \lambda} \otimes \mathcal{I}_{E})(\rho_{Q_AQ_BE}) \otimes |xyf\rangle \langle xyf| \,. \nonumber
		\end{equation*}
		After tracing out the systems $Q_{A}, Q_{B}, B$ and $Y$ we are left with
		\begin{equation*}
			\rho^2_{\tilde{A} XFE, \lambda} = \sum_{x} \frac{1}{2} \left( \sum_{f} \frac{1}{2} \, (\tr{}_{Q_{A}} \circ \mathcal{E}^{A}_{xf, \lambda} \otimes \mathcal{I}_{E})(\rho_{Q_AE}) \otimes |f\rangle \langle f| \right) \otimes |x\rangle \langle x| \,. \label{eq:rho_2}
		\end{equation*}
		We denote
		\begin{equation*}
			\rho^2_{\tilde{A}FE|X=x, \lambda} = \left( \sum_{f} \frac{1}{2} \, (\tr{}_{Q_{A}} \circ \mathcal{E}^{A}_{xf, \lambda} \otimes \mathcal{I}_{E})(\rho_{Q_AE}) \otimes |f\rangle \langle f| \right) \,.
		\end{equation*}
	
	We observe that the states $\rho^1_{\tilde{A} FE|X=x, \lambda}$ and $\rho^2_{\tilde{A} FE|X=x, \lambda}$ are equal and consequently we find
	\begin{align*}
		\chi (\tilde{A} : E F| X = x)_{\rho^1_{\tilde{A} FE|X=x, \lambda}} = \chi (\tilde{A} : E F| X = x)_{\rho^2_{\tilde{A} FE|X=x, \lambda}}\,.
	\end{align*}
	This concludes our prove that the Holevo quantity is the same in Process~1, with biased inputs, and Process~2, with uniform inputs.
	
	In the next step we want to express the bound on the Holevo quantity as a function of the violation $S_\mu$ of our MDL inequality. Starting with Equation~\eqref{eq:holevo-bound-CHSH} we know that 
	\begin{align*}
		\chi (\tilde{A} : E F| X = x)_{\rho^1_{\tilde{A} FE|X=x, \lambda}} &= \chi (\tilde{A} : E F| X = x)_{\rho^2_{\tilde{A}FE|X=x, \lambda}} \\
		&\leq h \left( \frac{1}{2} + \sqrt{\alpha (\alpha + 1)} \right) \,.
	\end{align*}
	Now we can relate $S_\mu$ to a \emph{minimal} Bell violation $\alpha$ that would have been observed with the given state and measurements.
	For the relation between the two violations we find
	\begin{align*}
		S_\mu &= \mu_ \mathrm{min} P_{\tilde{A} \tilde{B}|XY}(00|00) \cdot \underbrace{P_{XY}(00)}_{\leq \mu_ \mathrm{max}} - \\ 
		& \qquad \mu_ \mathrm{max} \Big( P_{\tilde{A} \tilde{B}|XY}(01|01) \cdot \underbrace{P_{XY}(01)}_{\geq \mu_ \mathrm{min}} + P_{\tilde{A} \tilde{B}|XY}(10|10) \cdot \underbrace{P_{XY}(10)}_{\geq \mu_ \mathrm{min}} + P_{\tilde{A} \tilde{B}|XY}(00|11) \cdot \underbrace{P_{XY}(11)}_{\geq \mu_ \mathrm{min}} \Big) \\
		& \leq \mu^* \left( P_{\tilde{A} \tilde{B}|XY}(00|00) - \big( P_{\tilde{A} \tilde{B}|XY}(01|01) + P_{\tilde{A} \tilde{B}|XY}(10|10) + P_{\tilde{A} \tilde{B}|XY}(00|11) \big) \right) \\
		& = \mu^* \cdot \alpha
	\end{align*}
	and hence
	\begin{equation}
		\alpha \geq \frac{1}{\mu^*} S_\mu \,. \label{eq:relation_alpha_smu}
	\end{equation}
	
	We find the final bound on the Holevo quantity by plugging this relation into Equation~\ref{eq:holevo-bound-CHSH},
	\begin{align*}
		\chi (\tilde{A} : E F| X = x) &\leq h \left( \frac{1}{2} + \sqrt{ \alpha (\alpha + 1)} \right) \\
		&\leq h \left( \frac{1}{2} + \frac{1}{\mu^*} \sqrt{ S_\mu (S_\mu + \mu^*)} \right) \,,
	\end{align*}
	where the last inequality holds because $h(x)$ is monotonically decreasing for $x \in [\frac{1}{2},1]$.
	A bound on the entropy can be found by employing Lemma~\ref{lma:rewrite},
	\begin{equation*}
		H(\tilde{A} \tilde{B}| XYFE) \geq 1 - h \left( \frac{1}{2} + \frac{1}{\mu^*} \sqrt{ S_\mu (S_\mu + \mu^*)} \right) \,.
	\end{equation*}
	
	We conclude the proof by showing that the bound on the entropy of the symmetrized outputs is the same as the bound on the entropy of the unsymmetrized outputs.
	Namely we have
	\begin{align*}
		H(\tilde{A} \tilde{B}| XYFE) &\leq H(AB| XYFE) \\
		&= H(AB|XYE) \,,
	\end{align*}
	where the first step follows because, for fixed $F$, the symmetrisation step is a deterministic operation, and the second step follows because $F$ is independent of everything else.
\end{proof}

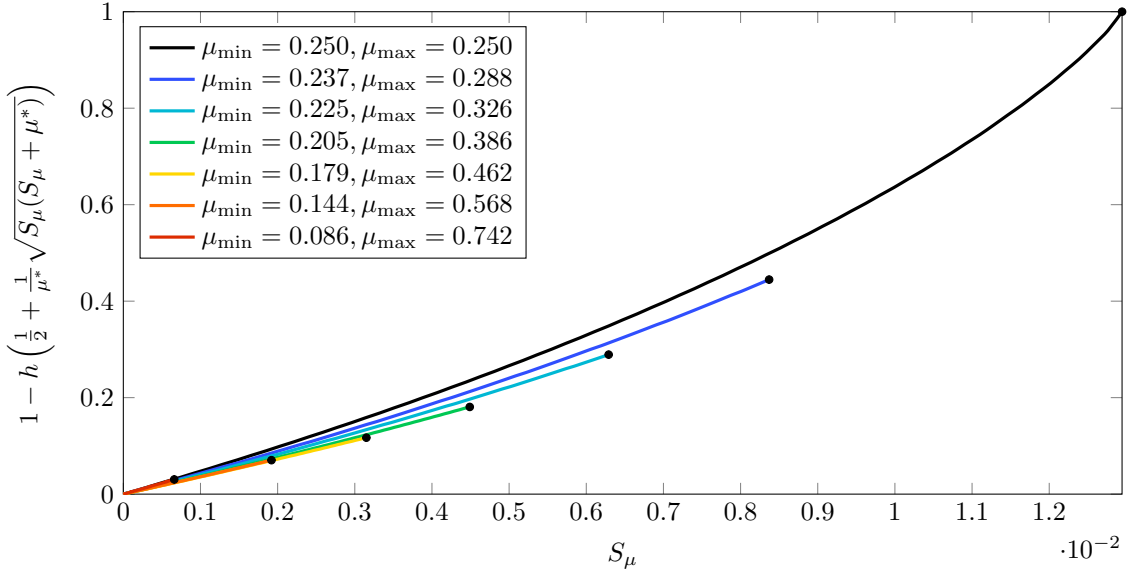
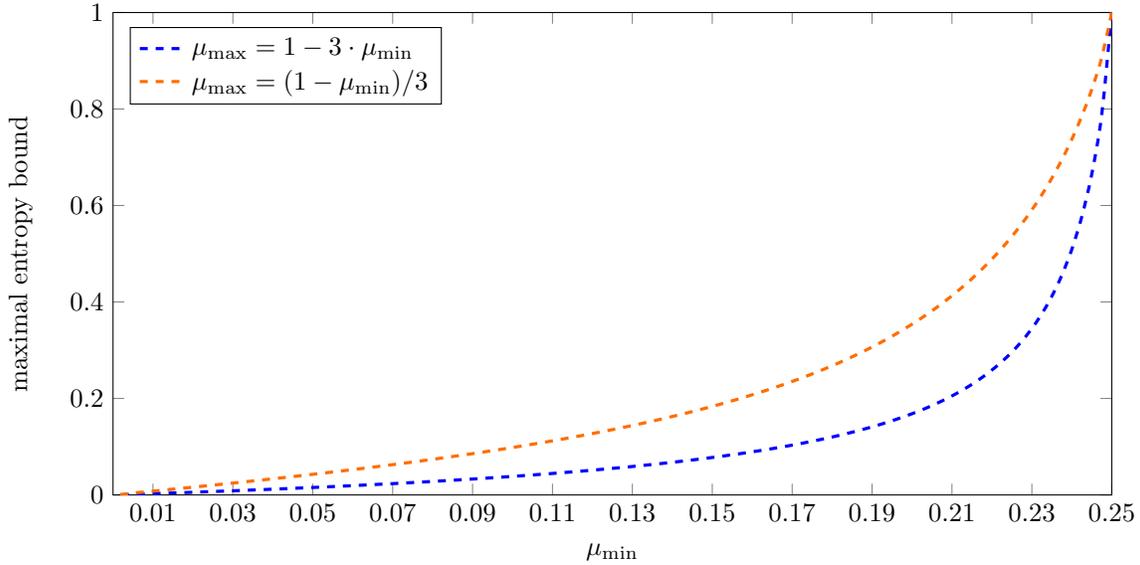
\begin{figure}
\centering
	\subfloat[Entropy bound in an MDL experiment as function of $S_\mu$, as given in Equation~\eqref{eq:single_round_bound}.]{
	\begin{tikzpicture}
		\begin{axis}[
			height=8cm,
			width=.9\textwidth,
			xlabel=$S_{\mu}$,
			ylabel=$1 - h \left( \frac{1}{2} + \frac{1}{\mu^*} \sqrt{ S_\mu (S_\mu + \mu^*)} \right)$,
			xmin=0.0,
			xmax=0.012944,
			ymax=1,
			ymin=0,
			legend style={at={(0.21,0.97)},anchor=north,legend cell align=left,font=\normalsize} 
			]
		
		
			\addplot[black,very thick,smooth] coordinates {
			(0.0, 0.0) (0.00019, 0.00858) (0.00037, 0.01725) (0.00055, 0.026) (0.00074, 0.03484) (0.00092, 0.04376) (0.00111, 0.05278) (0.00129, 0.06189) (0.00148, 0.07109) (0.00166, 0.08039) (0.00185, 0.08978) (0.00203, 0.09926) (0.00222, 0.10885) (0.0024, 0.11854) (0.00259, 0.12832) (0.00277, 0.13821) (0.00296, 0.14821) (0.00314, 0.15831) (0.00333, 0.16852) (0.00351, 0.17885) (0.0037, 0.18928) (0.00388, 0.19983) (0.00407, 0.2105) (0.00425, 0.22128) (0.00444, 0.23219) (0.00462, 0.24322) (0.00481, 0.25438) (0.00499, 0.26567) (0.00518, 0.27709) (0.00536, 0.28865) (0.00555, 0.30035) (0.00573, 0.31219) (0.00592, 0.32417) (0.0061, 0.33631) (0.00629, 0.3486) (0.00647, 0.36105) (0.00666, 0.37365) (0.00684, 0.38643) (0.00703, 0.39938) (0.00721, 0.4125) (0.0074, 0.42581) (0.00758, 0.43931) (0.00777, 0.453) (0.00795, 0.4669) (0.00814, 0.481) (0.00832, 0.49532) (0.00851, 0.50987) (0.00869, 0.52466) (0.00888, 0.53969) (0.00906, 0.55498) (0.00925, 0.57055) (0.00943, 0.58639) (0.00962, 0.60253) (0.0098, 0.61899) (0.00999, 0.63579) (0.01017, 0.65294) (0.01036, 0.67047) (0.01054, 0.68841) (0.01073, 0.70679) (0.01091, 0.72565) (0.0111, 0.74502) (0.01128, 0.76498) (0.01146, 0.78557) (0.01165, 0.80689) (0.01183, 0.82904) (0.01202, 0.85215) (0.0122, 0.87645) (0.01239, 0.90222) (0.01257, 0.92998) (0.01276, 0.96085) (0.012944, 1)
			};
			\addlegendentry{$\mu_{\mathrm{min}} = 0.250, \mu_{\mathrm{max}} = 0.250$}
			
			\addplot[plot1,very thick,smooth] coordinates {
			(0.0, 0.0) (0.00012, 0.00514) (0.00024, 0.01031) (0.00036, 0.01551) (0.00049, 0.02074) (0.00061, 0.026) (0.00073, 0.03129) (0.00085, 0.03661) (0.00097, 0.04197) (0.00109, 0.04736) (0.00121, 0.05278) (0.00133, 0.05823) (0.00146, 0.06372) (0.00158, 0.06924) (0.0017, 0.07479) (0.00182, 0.08038) (0.00194, 0.086) (0.00206, 0.09166) (0.00218, 0.09735) (0.0023, 0.10308) (0.00243, 0.10884) (0.00255, 0.11464) (0.00267, 0.12047) (0.00279, 0.12634) (0.00291, 0.13225) (0.00303, 0.1382) (0.00315, 0.14418) (0.00328, 0.1502) (0.0034, 0.15627) (0.00352, 0.16237) (0.00364, 0.1685) (0.00376, 0.17468) (0.00388, 0.1809) (0.004, 0.18716) (0.00412, 0.19346) (0.00425, 0.19981) (0.00437, 0.20619) (0.00449, 0.21262) (0.00461, 0.21909) (0.00473, 0.22561) (0.00485, 0.23216) (0.00497, 0.23877) (0.00509, 0.24542) (0.00522, 0.25211) (0.00534, 0.25885) (0.00546, 0.26564) (0.00558, 0.27248) (0.0057, 0.27936) (0.00582, 0.2863) (0.00594, 0.29328) (0.00606, 0.30031) (0.00619, 0.3074) (0.00631, 0.31454) (0.00643, 0.32173) (0.00655, 0.32897) (0.00667, 0.33627) (0.00679, 0.34362) (0.00691, 0.35103) (0.00704, 0.3585) (0.00716, 0.36602) (0.00728, 0.37361) (0.0074, 0.38125) (0.00752, 0.38896) (0.00764, 0.39672) (0.00776, 0.40456) (0.00788, 0.41245) (0.00801, 0.42041) (0.00813, 0.42844) (0.00825, 0.43654) (0.00837, 0.4447)
			};
			\addlegendentry{$\mu_\mathrm{min} = 0.237, \mu_\mathrm{max} = 0.288$}
			
			\addplot[plot2,very thick,smooth] coordinates {
			(0.0, 0.0) (9e-05, 0.0036) (0.00018, 0.00722) (0.00027, 0.01085) (0.00036, 0.0145) (0.00046, 0.01816) (0.00055, 0.02184) (0.00064, 0.02553) (0.00073, 0.02924) (0.00082, 0.03296) (0.00091, 0.0367) (0.001, 0.04045) (0.00109, 0.04423) (0.00119, 0.04801) (0.00128, 0.05181) (0.00137, 0.05563) (0.00146, 0.05947) (0.00155, 0.06332) (0.00164, 0.06719) (0.00173, 0.07107) (0.00182, 0.07497) (0.00191, 0.07889) (0.00201, 0.08283) (0.0021, 0.08678) (0.00219, 0.09075) (0.00228, 0.09473) (0.00237, 0.09874) (0.00246, 0.10276) (0.00255, 0.1068) (0.00264, 0.11085) (0.00274, 0.11492) (0.00283, 0.11902) (0.00292, 0.12313) (0.00301, 0.12725) (0.0031, 0.1314) (0.00319, 0.13556) (0.00328, 0.13975) (0.00337, 0.14395) (0.00347, 0.14817) (0.00356, 0.15241) (0.00365, 0.15667) (0.00374, 0.16094) (0.00383, 0.16524) (0.00392, 0.16956) (0.00401, 0.17389) (0.0041, 0.17825) (0.00419, 0.18263) (0.00429, 0.18702) (0.00438, 0.19144) (0.00447, 0.19587) (0.00456, 0.20033) (0.00465, 0.20481) (0.00474, 0.20931) (0.00483, 0.21383) (0.00492, 0.21837) (0.00502, 0.22293) (0.00511, 0.22752) (0.0052, 0.23212) (0.00529, 0.23675) (0.00538, 0.2414) (0.00547, 0.24608) (0.00556, 0.25077) (0.00565, 0.25549) (0.00574, 0.26023) (0.00584, 0.265) (0.00593, 0.26979) (0.00602, 0.2746) (0.00611, 0.27943) (0.0062, 0.28429) (0.00629, 0.28918)
			};
			\addlegendentry{$\mu_\mathrm{min} = 0.225, \mu_\mathrm{max} = 0.326$}

			\addplot[plot3,very thick,smooth] coordinates {
			(0.0, 0.0) (7e-05, 0.00238) (0.00013, 0.00476) (0.0002, 0.00715) (0.00026, 0.00955) (0.00033, 0.01195) (0.00039, 0.01436) (0.00046, 0.01678) (0.00052, 0.0192) (0.00059, 0.02163) (0.00065, 0.02407) (0.00072, 0.02651) (0.00078, 0.02896) (0.00085, 0.03141) (0.00091, 0.03388) (0.00098, 0.03635) (0.00104, 0.03882) (0.00111, 0.04131) (0.00117, 0.0438) (0.00124, 0.0463) (0.0013, 0.0488) (0.00137, 0.05131) (0.00143, 0.05383) (0.0015, 0.05636) (0.00156, 0.05889) (0.00163, 0.06143) (0.00169, 0.06398) (0.00176, 0.06653) (0.00182, 0.06909) (0.00189, 0.07166) (0.00195, 0.07424) (0.00202, 0.07682) (0.00208, 0.07941) (0.00215, 0.08201) (0.00221, 0.08462) (0.00228, 0.08723) (0.00234, 0.08985) (0.00241, 0.09248) (0.00247, 0.09511) (0.00254, 0.09776) (0.0026, 0.10041) (0.00267, 0.10307) (0.00273, 0.10573) (0.0028, 0.10841) (0.00286, 0.11109) (0.00293, 0.11378) (0.00299, 0.11647) (0.00306, 0.11918) (0.00312, 0.12189) (0.00319, 0.12461) (0.00325, 0.12734) (0.00332, 0.13007) (0.00338, 0.13282) (0.00345, 0.13557) (0.00351, 0.13833) (0.00358, 0.1411) (0.00364, 0.14388) (0.00371, 0.14666) (0.00377, 0.14945) (0.00384, 0.15226) (0.0039, 0.15506) (0.00397, 0.15788) (0.00403, 0.16071) (0.0041, 0.16354) (0.00416, 0.16639) (0.00423, 0.16924) (0.00429, 0.1721) (0.00436, 0.17497) (0.00442, 0.17785) (0.00449, 0.18073) 
			};
			\addlegendentry{$\mu_\mathrm{min} = 0.205, \mu_\mathrm{max} = 0.386$}
			
			\addplot[plot4,very thick,smooth] coordinates {
			(0.0, 0.0) (5e-05, 0.00159) (9e-05, 0.00319) (0.00014, 0.00478) (0.00018, 0.00638) (0.00023, 0.00798) (0.00027, 0.00959) (0.00032, 0.0112) (0.00037, 0.01281) (0.00041, 0.01442) (0.00046, 0.01604) (0.0005, 0.01766) (0.00055, 0.01928) (0.00059, 0.02091) (0.00064, 0.02254) (0.00068, 0.02417) (0.00073, 0.0258) (0.00078, 0.02744) (0.00082, 0.02908) (0.00087, 0.03073) (0.00091, 0.03237) (0.00096, 0.03402) (0.001, 0.03568) (0.00105, 0.03733) (0.0011, 0.03899) (0.00114, 0.04065) (0.00119, 0.04232) (0.00123, 0.04399) (0.00128, 0.04566) (0.00132, 0.04733) (0.00137, 0.04901) (0.00141, 0.05069) (0.00146, 0.05238) (0.00151, 0.05407) (0.00155, 0.05576) (0.0016, 0.05745) (0.00164, 0.05915) (0.00169, 0.06085) (0.00173, 0.06255) (0.00178, 0.06426) (0.00183, 0.06597) (0.00187, 0.06768) (0.00192, 0.0694) (0.00196, 0.07112) (0.00201, 0.07284) (0.00205, 0.07457) (0.0021, 0.0763) (0.00214, 0.07803) (0.00219, 0.07976) (0.00224, 0.0815) (0.00228, 0.08325) (0.00233, 0.08499) (0.00237, 0.08674) (0.00242, 0.08849) (0.00246, 0.09025) (0.00251, 0.09201) (0.00256, 0.09377) (0.0026, 0.09554) (0.00265, 0.09731) (0.00269, 0.09908) (0.00274, 0.10086) (0.00278, 0.10263) (0.00283, 0.10442) (0.00287, 0.1062) (0.00292, 0.10799) (0.00297, 0.10979) (0.00301, 0.11158) (0.00306, 0.11338) (0.0031, 0.11519) (0.00315, 0.117)
			};
			\addlegendentry{$\mu_\mathrm{min} = 0.179, \mu_\mathrm{max} = 0.462$}
			
			\addplot[plot5,very thick,smooth] coordinates {
			(0.0, 0.0) (3e-05, 0.00098) (6e-05, 0.00197) (8e-05, 0.00295) (0.00011, 0.00393) (0.00014, 0.00492) (0.00017, 0.00591) (0.00019, 0.00689) (0.00022, 0.00788) (0.00025, 0.00887) (0.00028, 0.00986) (0.00031, 0.01085) (0.00033, 0.01185) (0.00036, 0.01284) (0.00039, 0.01384) (0.00042, 0.01483) (0.00044, 0.01583) (0.00047, 0.01683) (0.0005, 0.01783) (0.00053, 0.01883) (0.00056, 0.01983) (0.00058, 0.02084) (0.00061, 0.02184) (0.00064, 0.02285) (0.00067, 0.02385) (0.00069, 0.02486) (0.00072, 0.02587) (0.00075, 0.02688) (0.00078, 0.02789) (0.00081, 0.02891) (0.00083, 0.02992) (0.00086, 0.03094) (0.00089, 0.03195) (0.00092, 0.03297) (0.00094, 0.03399) (0.00097, 0.03501) (0.001, 0.03603) (0.00103, 0.03705) (0.00106, 0.03807) (0.00108, 0.0391) (0.00111, 0.04012) (0.00114, 0.04115) (0.00117, 0.04218) (0.00119, 0.04321) (0.00122, 0.04424) (0.00125, 0.04527) (0.00128, 0.0463) (0.00131, 0.04733) (0.00133, 0.04837) (0.00136, 0.04941) (0.00139, 0.05044) (0.00142, 0.05148) (0.00144, 0.05252) (0.00147, 0.05356) (0.0015, 0.0546) (0.00153, 0.05565) (0.00156, 0.05669) (0.00158, 0.05774) (0.00161, 0.05879) (0.00164, 0.05983) (0.00167, 0.06088) (0.0017, 0.06193) (0.00172, 0.06299) (0.00175, 0.06404) (0.00178, 0.06509) (0.00181, 0.06615) (0.00183, 0.06721) (0.00186, 0.06826) (0.00189, 0.06932) (0.00192, 0.07038)
			};
			\addlegendentry{$\mu_\mathrm{min} = 0.144, \mu_\mathrm{max} = 0.568$}
			
			\addplot[plot6,very thick,smooth] coordinates {
			(0.0, 0.0) (1e-05, 0.00044) (2e-05, 0.00087) (3e-05, 0.00131) (4e-05, 0.00174) (5e-05, 0.00217) (6e-05, 0.00261) (7e-05, 0.00304) (8e-05, 0.00348) (9e-05, 0.00391) (0.0001, 0.00435) (0.00011, 0.00478) (0.00011, 0.00522) (0.00012, 0.00566) (0.00013, 0.00609) (0.00014, 0.00653) (0.00015, 0.00697) (0.00016, 0.0074) (0.00017, 0.00784) (0.00018, 0.00828) (0.00019, 0.00872) (0.0002, 0.00915) (0.00021, 0.00959) (0.00022, 0.01003) (0.00023, 0.01047) (0.00024, 0.01091) (0.00025, 0.01135) (0.00026, 0.01178) (0.00027, 0.01222) (0.00028, 0.01266) (0.00029, 0.0131) (0.0003, 0.01354) (0.00031, 0.01398) (0.00032, 0.01442) (0.00033, 0.01486) (0.00034, 0.01531) (0.00034, 0.01575) (0.00035, 0.01619) (0.00036, 0.01663) (0.00037, 0.01707) (0.00038, 0.01751) (0.00039, 0.01796) (0.0004, 0.0184) (0.00041, 0.01884) (0.00042, 0.01928) (0.00043, 0.01973) (0.00044, 0.02017) (0.00045, 0.02061) (0.00046, 0.02106) (0.00047, 0.0215) (0.00048, 0.02195) (0.00049, 0.02239) (0.0005, 0.02284) (0.00051, 0.02328) (0.00052, 0.02373) (0.00053, 0.02417) (0.00054, 0.02462) (0.00055, 0.02506) (0.00056, 0.02551) (0.00056, 0.02595) (0.00057, 0.0264) (0.00058, 0.02685) (0.00059, 0.0273) (0.0006, 0.02774) (0.00061, 0.02819) (0.00062, 0.02864) (0.00063, 0.02909) (0.00064, 0.02953) (0.00065, 0.02998) (0.00066, 0.03043) 
			};
			\addlegendentry{$\mu_\mathrm{min} = 0.086, \mu_\mathrm{max} = 0.742$}
			
			\addplot[color=black, only marks, mark=\MARKFORM, mark size=\MARKSZ] coordinates { (0.012944, 1) };
			\addplot[color=black, only marks, mark=\MARKFORM, mark size=\MARKSZ] coordinates { (0.00837, 0.4447) };
			\addplot[color=black, only marks, mark=\MARKFORM, mark size=\MARKSZ] coordinates { (0.00629, 0.28918) };
			\addplot[color=black, only marks, mark=\MARKFORM, mark size=\MARKSZ] coordinates { (0.00449, 0.18073) };
			\addplot[color=black, only marks, mark=\MARKFORM, mark size=\MARKSZ] coordinates { (0.00315, 0.117) };
			\addplot[color=black, only marks, mark=\MARKFORM, mark size=\MARKSZ] coordinates { (0.00192, 0.07038) };
			\addplot[color=black, only marks, mark=\MARKFORM, mark size=\MARKSZ] coordinates { (0.00066, 0.03043) };
		\end{axis}  
	\end{tikzpicture}
	\label{fig:entropy_single_round_S}
	}
	\\
	\subfloat[Lower bound on maximal achievable entropy as function of MDL source parameters, $\mu$.]{
	\begin{tikzpicture}
		\begin{axis}[
			height=8cm,
			width=.9\textwidth,
			xlabel=$\mu_{\mathrm{min}}$,
			ylabel=maximal entropy bound,
			xmin=0,
			xmax=0.25,
			ymax=1,
			ymin=0,
			xtick={0.01, 0.03, 0.05, 0.07, 0.09, 0.11 , 0.13, 0.15, 0.17, 0.19, 0.21, 0.23, 0.25},
			xticklabels={0.01, 0.03, 0.05, 0.07, 0.09, 0.11 , 0.13, 0.15, 0.17, 0.19, 0.21, 0.23, 0.25},
			ytick={0,0.2,0.4,0.6,0.8,1},
			legend style={at={(0.173,0.97)},anchor=north,legend cell align=left,font=\normalsize} 
			]
		
		
			\addplot[blue,very thick,smooth,dashed] coordinates {
			(0.25, 1.0) (0.24747, 0.78663) (0.24495, 0.663) (0.24242, 0.57426) (0.2399, 0.50654) (0.23737, 0.45294) (0.23485, 0.40938) (0.23232, 0.37324) (0.2298, 0.34275) (0.22727, 0.31667) (0.22475, 0.29409) (0.22222, 0.27433) (0.2197, 0.25689) (0.21717, 0.24137) (0.21465, 0.22746) (0.21212, 0.21491) (0.2096, 0.20352) (0.20707, 0.19314) (0.20455, 0.18363) (0.20202, 0.17487) (0.19949, 0.16679) (0.19697, 0.15929) (0.19444, 0.15232) (0.19192, 0.14582) (0.18939, 0.13973) (0.18687, 0.13402) (0.18434, 0.12865) (0.18182, 0.12359) (0.17929, 0.11881) (0.17677, 0.11428) (0.17424, 0.10999) (0.17172, 0.10592) (0.16919, 0.10204) (0.16667, 0.09835) (0.16414, 0.09482) (0.16162, 0.09145) (0.15909, 0.08823) (0.15657, 0.08514) (0.15404, 0.08218) (0.15152, 0.07934) (0.14899, 0.07661) (0.14646, 0.07398) (0.14394, 0.07145) (0.14141, 0.069) (0.13889, 0.06665) (0.13636, 0.06438) (0.13384, 0.06218) (0.13131, 0.06006) (0.12879, 0.058) (0.12626, 0.05601) (0.12374, 0.05409) (0.12121, 0.05222) (0.11869, 0.0504) (0.11616, 0.04864) (0.11364, 0.04693) (0.11111, 0.04527) (0.10859, 0.04365) (0.10606, 0.04208) (0.10354, 0.04055) (0.10101, 0.03906) (0.09848, 0.03761) (0.09596, 0.0362) (0.09343, 0.03482) (0.09091, 0.03347) (0.08838, 0.03216) (0.08586, 0.03088) (0.08333, 0.02963) (0.08081, 0.02841) (0.07828, 0.02721) (0.07576, 0.02604) (0.07323, 0.0249) (0.07071, 0.02379) (0.06818, 0.0227) (0.06566, 0.02163) (0.06313, 0.02058) (0.06061, 0.01956) (0.05808, 0.01855) (0.05556, 0.01757) (0.05303, 0.01661) (0.05051, 0.01566) (0.04798, 0.01473) (0.04545, 0.01383) (0.04293, 0.01294) (0.0404, 0.01206) (0.03788, 0.0112) (0.03535, 0.01036) (0.03283, 0.00953) (0.0303, 0.00872) (0.02778, 0.00792) (0.02525, 0.00714) (0.02273, 0.00637) (0.0202, 0.00561) (0.01768, 0.00487) (0.01515, 0.00414) (0.01263, 0.00342) (0.0101, 0.00271) (0.00758, 0.00202) (0.00505, 0.00133) (0.00253, 0.00023) (0, 0)
			};
			\addlegendentry{$\mu_\mathrm{max} = 1 - 3 \cdot \mu_\mathrm{min}$}
			
			\addplot[plot5,very thick,smooth,dashed] coordinates {
			(0.25, 1.0) (0.24747, 0.90603) (0.24495, 0.83923) (0.24242, 0.78355) (0.2399, 0.7353) (0.23737, 0.69265) (0.23485, 0.65445) (0.23232, 0.61994) (0.2298, 0.58853) (0.22727, 0.55978) (0.22475, 0.53335) (0.22222, 0.50895) (0.2197, 0.48634) (0.21717, 0.46532) (0.21465, 0.44573) (0.21212, 0.42743) (0.2096, 0.41028) (0.20707, 0.39418) (0.20455, 0.37903) (0.20202, 0.36475) (0.19949, 0.35126) (0.19697, 0.33851) (0.19444, 0.32642) (0.19192, 0.31494) (0.18939, 0.30404) (0.18687, 0.29366) (0.18434, 0.28377) (0.18182, 0.27433) (0.17929, 0.26531) (0.17677, 0.25668) (0.17424, 0.24842) (0.17172, 0.2405) (0.16919, 0.23289) (0.16667, 0.22558) (0.16414, 0.21856) (0.16162, 0.21179) (0.15909, 0.20527) (0.15657, 0.19898) (0.15404, 0.19292) (0.15152, 0.18705) (0.14899, 0.18138) (0.14646, 0.1759) (0.14394, 0.17058) (0.14141, 0.16543) (0.13889, 0.16044) (0.13636, 0.15559) (0.13384, 0.15089) (0.13131, 0.14631) (0.12879, 0.14186) (0.12626, 0.13753) (0.12374, 0.13332) (0.12121, 0.12921) (0.11869, 0.12521) (0.11616, 0.1213) (0.11364, 0.11749) (0.11111, 0.11377) (0.10859, 0.11013) (0.10606, 0.10657) (0.10354, 0.10309) (0.10101, 0.09969) (0.09848, 0.09635) (0.09596, 0.09309) (0.09343, 0.08989) (0.09091, 0.08675) (0.08838, 0.08367) (0.08586, 0.08065) (0.08333, 0.07769) (0.08081, 0.07477) (0.07828, 0.07191) (0.07576, 0.0691) (0.07323, 0.06633) (0.07071, 0.0636) (0.06818, 0.06092) (0.06566, 0.05829) (0.06313, 0.05569) (0.06061, 0.05312) (0.05808, 0.0506) (0.05556, 0.04811) (0.05303, 0.04565) (0.05051, 0.04323) (0.04798, 0.04084) (0.04545, 0.03848) (0.04293, 0.03614) (0.0404, 0.03384) (0.03788, 0.03156) (0.03535, 0.02931) (0.03283, 0.02708) (0.0303, 0.02488) (0.02778, 0.0227) (0.02525, 0.02054) (0.02273, 0.0184) (0.0202, 0.01628) (0.01768, 0.01419) (0.01515, 0.01211) (0.01263, 0.01005) (0.0101, 0.00801) (0.00758, 0.00598) (0.00505, 0.00397) (0.00253, 0.00142) (0, 0)
			};
			\addlegendentry{$\mu_\mathrm{max} = (1 - \mu_\mathrm{min})/3$}
		\end{axis}  
	\end{tikzpicture}
	\label{fig:entropy_single_round_mu}
	}
	\caption{In Figure~\ref{fig:entropy_single_round_S} we plot the entropy bound given in Equation~\eqref{eq:single_round_bound} as a function of the MDL violation $S_\mu$ for several choices of $\mu$. 
	The black dots indicate a lower bound on the maximal possible $S_{\mu}$ within quantum mechanics.
	In Figure~\ref{fig:entropy_single_round_mu} we show a lower bound on the maximal achievable entropy as a function of $\mu_\mathrm{min}$ (recall that $\mu_{\mathrm{min}} = \nicefrac{1}{4}$ corresponds to a uniform source).
	The exact way we lower bound the maximal achievable entropy is explained in Appendix~\ref{apx:max-viol}.
	The plot of the maximal achievable entropy shows one curve where we fixed $\mu_{\mathrm{min}}$ and chose $\mu_{\mathrm{max}} = 1 - 3\cdot \mu_{\mathrm{min}}$, and one plot where we chose $\mu_{\mathrm{max}} = \nicefrac{(1 - \mu_{\mathrm{min}})}{3}$. 
	These two choices correspond to the two extreme cases for fixed $\mu_{\mathrm{min}}$. 
	A source with $\mu_{\mathrm{max}} = 1 - 3\cdot \mu_{\mathrm{min}}$ is, for our purpose, the worst case since it produces outputs such that one output pair is considerably favoured  while all other pairs appear with low probability $\mu_{\mathrm{min}}$.
	A source where $\mu_{\mathrm{max}} = \nicefrac{(1 - \mu_{\mathrm{min}})}{3}$ is, for our purpose, the best case since only one output pair appears with low probability while all the other pairs appear with equally (high) probability.
	}
\label{fig:single_round_bound}
\end{figure}

A plot of the bound from Lemma~\ref{lma:holevo_bound} is shown in Figure~\ref{fig:single_round_bound}.
Once it shows the entropy bound as a function of the MDL violation $S_\mu$ for different $\mu$, and once a lower bound on the maximal achievable entropy as a function of $\mu_ \mathrm{min}$.
The reason that we can only plot a lower bound on the maximal entropy is due to the dependence of the maximal entropy on the specific source.
The exact reasoning and how we obtained the lower bound is explained in Appendix~\ref{apx:max-viol}.

In the plots we clearly see that the entropy of the outputs increases with increasing MDL violation.
Furthermore we can observe that the maximal achievable entropy bound decreases with increasing source bias.
Intuitively this makes sense since we expect to get a lower amount of randomness in the outputs if we start with less random inputs.

We also see in the above plot that the entropy is non-zero once there is a violation of the relevant inequality. As will be clear from the next Section, this implies that, asymptotically, we can tolerate maximal amount of noise --- as long as there is a violation of the MDL inequality some randomness can be extracted.

\section{Randomness amplification protocol}
\label{sec:RAP}

In the following sections we first introduce the setting of our randomness amplification protocol and explicitly state the assumptions that we are taking.
After that we introduce the protocol and proceed to prove the completeness and the soundness of the protocol.

\subsection{Setting and assumptions}
\label{sec:assumptions}

\begin{figure}
	\begin{center}
	\begin{tikzpicture} 
		\pgfmathsetmacro{\WIDTH}{1.5}
		\node (l) at (0,0) {$\lambda$};
		\node (XY) at (1.5*\WIDTH, -\WIDTH) {$X^{n}Y^{n}$};
		\node (Z) at (3*\WIDTH, -\WIDTH) {$Z^{d}$};
		\node (E) at (\WIDTH, \WIDTH) {$E$};
		\node (AB) at (3*\WIDTH, \WIDTH) {$A^{n}B^{n}$};
		\node (Ext) at (4*\WIDTH, 0) {$\mathrm{Ext}$};
		\node (K) at (5*\WIDTH, 0) {$K^{m}$};
		\draw[->] (l) edge (XY) (XY) edge (Z);
		\draw[->] (l) -- (E) (E) edge (AB);
		\draw[->] (XY) -- (AB);
		\draw[->] (AB) -- (Ext);
		\draw[->] (Z) -- (Ext);
		\draw[->] (Ext) -- (K);
		\draw[->] (.5*\WIDTH, -2*\WIDTH) -- node [below] {time} (4.5*\WIDTH, -2*\WIDTH);
	\end{tikzpicture}
	\end{center}
	\caption{
		Schematic drawing of the setting in the RAP.
		The adversary learns $\lambda$, the previous outputs of the source and possibly any other knowledge that might help her predict the outputs of the source.
		The source then produces $X^{n}Y^{n}$, a string of binary random variables that might depend on $\lambda$.
		The string $X^{n}Y^{n}$ together with a device, that Eve produced, is used by Alice and Bob to produce the outputs $A^{n}B^{n}$.
		During that process Eve can keep a purification of the quantum state in the device, $E$, that helps her predict $A^{n}B^{n}$.
		Finally, the string $Z^{d}$ is produced and the final output $K^{m}$ is produced from $A^{n}B^{n}$ and $Z^{d}$ using the extractor.
	}
\label{fig:RAP-setting}
\end{figure}
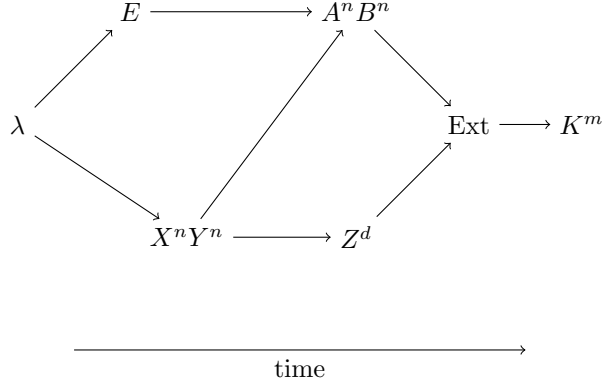

We consider a setting where we have an MDL source and an untrusted device with at least two components.
All the components in our setting are spatially separated and possibly manufactured by an adversary.
Since the source and the components of the device are separated the non-signalling condition holds pairwise between them.
Both the device and the source can be correlated with some classical side information $\lambda$ that the adversary holds.
Furthermore the adversary can have access to a quantum state $\rho_{E}$ that can be correlated with the device and the source.

During the protocol the source produces the inputs $X^{n}Y^{n}$ for the device which then, upon receiving the inputs, produces the outputs $A^{n}B^{n}$.
After the device produced the outputs the source produces another string of binary random variables $Z^{d}$.
The extractor then produces the final output $K^{m}$ using $A^{n}B^{n}$ and $Z^{d}$ as inputs.
The whole setting is schematically depicted in Figure~\ref{fig:RAP-setting}.

We summarise the general assumptions of the analysis of our protocol in the following list.
\begin{enumerate}[1.]
	\item Quantum mechanics is correct.
	\item The adversary is limited by quantum mechanics and without loss of generality we can assume that the adversary only holds a purification of Alice and Bob's initial quantum state.
	\item The untrusted device has at least two separated components.
\end{enumerate}
Moreover, we state the assumptions that are related to our specific setting.
\begin{enumerate}[1.]
\setcounter{enumi}{3}
	\item The adversary only has classical side information, $\lambda$, about the source of randomness. \label{item:guess-1}
	\item The source of randomness is a public $\mu$-MDL source. \label{item:guess-2}
\end{enumerate}
Assumptions~\ref{item:guess-1} and~\ref{item:guess-2} imply that the guessing probability (Definition~\ref{def:min-entropies}) for the outputs of the source is bounded as follows $\mu_{\mathrm{min}} \leq p_{\mathrm{guess}} (X_{i}Y_{i}|X^{i-1}Y^{i-1} E, \lambda) \leq \mu_{\mathrm{max}} \; \forall \lambda, i$.

\begin{enumerate}[1.]
\setcounter{enumi}{5}
	\item While the device produces outputs, it holds that 
	\begin{equation*}
		I(A^{i-1}B^{i-1}:X_{i}Y_{i}|X^{i-1}Y^{i-1} E, \lambda) = 0
	\end{equation*}
	and, after the device is done, it holds that
	\begin{equation*}
	I(Z^{d}:A^{n}B^{n}|X^{n}Y^{n} E, \lambda) = 0 \,.
	\end{equation*}
	\label{item:independence}
\end{enumerate}

The first two assumptions amount to the assumption that quantum mechanics is correct and complete.
Since no experimental evidence has been found that this is not the case, these are reasonable.
Furthermore, the fact that the device consists of at least two components can easily be verified by inspecting it before executing the protocol and the non-signalling condition can be reliably enforced by shielding the two parts of the device.
Without having any restrictions on the source we could not do anything. One therefore must use some assumptions about the source. Here we use Assumptions~\ref{item:guess-1} and~\ref{item:guess-2} which are necessary for our proof technique.
Assumption~\ref{item:independence} can be understood as assuming that, given the adversary's side information, the device and the source are independent, which again can be seen as the restriction that the adversary does not have access to the device and the source after they were produced.
This condition can easily be enforced by securing the devices from being tampered with.

\subsection{Security definition}\label{sec:secur_def}

We define the security via the secrecy of its outputs, similar as was done for the security definition of the DIQKD protocol in~\cite{RotemEAT}.
\begin{defi}[Secrecy]
\label{def:RAP-secrecy}
	A randomness amplification protocol is said to be $\varepsilon_{\mathrm{RA}}$-secret, when implemented using a device $D$, if for an output of length $m$,
	\begin{equation*}
		\left( 1 - \mathrm{Pr}[\text{abort}] \right) \left\Vert \rho_{K^{m} \Sigma} - \rho_{U^{m}} \otimes \rho_{\Sigma} \right\Vert \leq \varepsilon_{RA} \,,
	\end{equation*}
	where $K^{m}$ is the output of the RAP, $\Sigma$ is the adversary's side information that can be correlated with $D$, and $U^{m}$ is a uniform random variable of $m$ bits.
\end{defi}

The protocol is thus said to be secure if either the protocol aborts with high probability or the outputs are close to uniform.

\subsection{The protocol}\label{sec:ra_protocol}

The protocol is given in Protocol~\ref{alg:RAP}.

\begin{algorithm}[t]
	\floatname{algorithm}{Protocol}
	\raggedright
	\caption{Randomness Amplification Protocol}
	\label{alg:RAP}
	\begin{algorithmic}[1]
		\STATEx \textbf{Arguments:} 
		\STATEx\hspace{\algorithmicindent} $M(\mu)$ -- $\mu$-MDL source
		\STATEx\hspace{\algorithmicindent} $D$ -- untrusted device of (at least) two components
		\STATEx\hspace{\algorithmicindent} $n \in \mathbb{N}_+$ -- number of rounds
		\STATEx\hspace{\algorithmicindent} $S_{\mathrm{exp}}$ -- lower bound on the expected violation of the MDL inequality for an honest (possibly noisy) implementation
		\STATEx\hspace{\algorithmicindent} $\delta_{\mathrm{est}} \in \left ( 0,S_\mathrm{exp} \right )$ -- width of statistical confidence interval for the estimation test 
		\STATEx\hspace{\algorithmicindent} $\mathrm{Ext}:\{0,1\}^{2n}\times\{0,1\}^d\rightarrow\{0,1\}^m$ -- $(k_{1}, k_{2}, \varepsilon_{\mathrm{ext}})$ quantum-proof randomness extractor in the Markov model which is strong in the second input.
	
		\STATEx
	
		\STATEx \textbf{Entropy Accumulation:}
		\STATE For every round $i \in \{1,\dots,n\}$ do steps \ref{step:first}-\ref{step:last}.
		\STATE\hspace{\algorithmicindent} Alice and Bob choose inputs $X_i$ and $Y_i$ from the MDL source. \label{step:first}
		\STATE\hspace{\algorithmicindent} Alice and Bob use $D$ with $X_i,Y_i$ and record their outputs as $A_i,B_i$.
		\STATE\hspace{\algorithmicindent} Alice and Bob set $C_i = w(A_i,B_i,X_i,Y_i)$ for $w$ as defined in Equation~\eqref{eq:winning_func}. \label{step:last}
		\STATE Alice and Bob abort the protocol if $\bar{C} \equiv \frac{1}{n} \sum_j C_j < (S_\mathrm{exp} - \delta_{\mathrm{est}})$. \label{step:abortion}
		\STATEx \textbf{Randomness Extraction:}
		\STATE Draw a bit string $Z^{d}$ from $M(\mu)$.
		\STATE Use $\mathrm{Ext}$ to create $K^{m} = \mathrm{Ext}(A^{n}B^{n}, Z^{d})$.
	\end{algorithmic}
\end{algorithm}

Our proposed RAP consists of two parts.
In the first part we accumulate entropy.
For that matter we perform a series of MDL experiments, similar to the one described in Section~\ref{sec:single-round}.
In these MDL experiments we draw inputs from an MDL source and feed them to a device that produces outputs.
Ideally these outputs will be generated by doing measurements on a quantum state such that an MDL inequality is violated.
In the second part we draw another string from the MDL source and use this string, as well as the output from the entropy accumulation part, as inputs for the extractor.
The extractor then produces the final output.

During the entropy accumulation part of the protocol the variable $C_i$ is set in each round.
This variable is set to help evaluate whether the protocol should abort or not.
In each round $C_i$ is set according to the winning function
\begin{align}
	w(A_i,B_i,X_i,Y_i) = \, \begin{cases} \mu_{\mathrm{min}} & \text{if } (A_i,B_i,X_i,Y_i) = (0,0,0,0) \\ 
	-\mu_{\mathrm{max}} & \text{if } (A_i,B_i,X_i,Y_i) \in \{(0,1,0,1), (1,0,1,0), (0,0,1,1)\} \\
	0 & \text{otherwise} \;. \end{cases} 
\label{eq:winning_func}
\end{align}

After the $n$ rounds of of the entropy accumulation routine we decide whether to abort or not by comparing $\bar{C} \equiv \frac{1}{n} \sum_j C_j$ with $(S_\mathrm{exp} - \delta_{\mathrm{est}})$.
Note that we almost always abort for sources where $S^{*}_{\mu} \leq \delta_{\mathrm{est}}$ ($S^{*}_{\mu}$ is the maximal MDL value in quantum mechanics, see Appendix~\ref{apx:max-viol}).
Thus we cannot amplify randomness for sources for which $S^{*}_{\mu} \leq \delta_{\mathrm{est}}$.
However, we need a positive $\delta_{\mathrm{est}}$ in order to get a low probability for aborting in an honest implementation (see Section~\ref{sec:completeness}).
To remedy this problem we can decrease $\delta_{\mathrm{est}}$ at the cost of increasing $n$.
Hence, it is possible to have a reasonable probability of aborting in an honest implementation and still be able to amplify arbitrary SV sources.

We state the following theorem that quantifies the quality of the protocol's output.
It is a formal version of Theorem~\ref{thm:informal} which was given in the introduction.
The proof of the theorem is given in the end of the section as it combines our separate proofs of soundness and completeness.

\begin{thm}\label{thm:formal}
	Given any public SV-source with bias $\mu \in (0,0.5)$ and any two-component device $D$ that fulfils the assumptions described in Section~\ref{sec:assumptions}, let $n$ be the number of rounds in Protocol~\ref{alg:RAP}, $\varepsilon_{\mathrm{s}}, \varepsilon_{\mathrm{EA}} \in (0,1)$, $S_\mathrm{exp} \leq S_{\mu}^{*}$, $\delta_{\mathrm{est}} \in \left ( 0,S_\mathrm{exp} \right )$ and $m,\varepsilon_{\mathrm{ext}} $ the parameters of the $(k_{1}, k_{2}, \varepsilon_{\mathrm{ext}})$-extractor used in Protocol~\ref{alg:RAP}, with $k_1,k_2$ fulfilling Equation~\eqref{eq:ext_k_values}.
	Then:
	\begin{enumerate}
		\item (Secrecy) Protocol~\ref{alg:RAP} produces a string $K^{m}$ of length $m$ such that:
		\[
			\left(1-\Pr[\text{abort}]\right) \left\Vert \rho_{K^{m} \Sigma} - \rho_{U^{m}} \otimes \rho_{\Sigma} \right\Vert \leq 12 \left(\varepsilon_{\mathrm{s}} +\varepsilon_{\mathrm{ext}} \right) + \varepsilon_{\mathrm{EA}}\;,
		\]
		where $\Sigma=EX^nY^nZ^d\lambda$ is  the adversary's side information.
		 \label{part:soundness}
		\item (Completeness) There exists an honest implementation of the device such that Protocol~\ref{alg:RAP} aborts with probability $\varepsilon_{\mathrm{c}} \leq \exp \left( - \frac{2 n \delta_{\mathrm{est}}^2}{\left (\mu_\mathrm{min}+\mu_\mathrm{max} \right )^2} \right)$ when using this device. \label{part:completeness}
	\end{enumerate} 
\end{thm}

\subsection{Completeness}
\label{sec:completeness}

In order for our RAP to be useful we do not only need a protocol that, in theory, produces uniform outputs but also one that can be implemented.
We call this criterion the completeness of the protocol.

\begin{lma}[Completeness]
\label{lma:completeness}
	Let $M(\mu)$ be any $\mu$-MDL source and let $S_\mathrm{exp} \leq S_{\mu}^{*}$, where $S_{\mu}^{*}$ is the maximal possible value for $S_\mu$ in quantum theory for the given MDL source.
	Then Protocol~\ref{alg:RAP} is complete with completeness parameter $\varepsilon_{\mathrm{c}} \leq \exp \left( - \frac{2 n \delta_{\mathrm{est}}^2}{\left (\mu_\mathrm{min}+\mu_\mathrm{max} \right )^2} \right)$;
	i.e., the probability to abort in an honest implementation is upper bounded by $\varepsilon_{\mathrm{c}}$.
\end{lma}

\begin{proof}
	We want to show that there exists a device such that the protocol aborts with probability less than $\varepsilon_{\mathrm{c}}$.
	If we implement our device to perform $n$ independent, identical measurements on the product state $\rho_{Q_AQ_B}^{\otimes n}$, where $\rho_{Q_AQ_B}$ together with the chosen measurements achieves an MDL value $S_{\mathrm{exp}}$ of the MDL inequality, the expectation value of $\bar{C} = \frac{1}{n} \sum_j C_j$ is given by $\mathbb{E}[\bar{C}] = S_{\mathrm{exp}}$.
	We can then use Hoeffding's inequality to get an upper bound on the probability that the protocol aborts.
	We have
	\begin{align*}
		\mathrm{Pr}\left[ \mathrm{aborting} \right] &= \mathrm{Pr}\left[ \bar{C} < (S_{\mathrm{exp}} - \delta_{\mathrm{est}}) \right] \\
		&= \mathrm{Pr}\left[ S_{\mathrm{exp}} - \bar{C} > \delta_{\mathrm{est}} \right] \\
		&\leq \exp \left( - \frac{2 n \delta_{\mathrm{est}}^2}{\left (\mu_\mathrm{min}+\mu_\mathrm{max} \right )^2} \right) \, . \qedhere
	\end{align*}
\end{proof}

In Lemma~\ref{lma:completeness} we showed that, as long as $S_{\mathrm{exp}}$ is less than the maximal quantum value of $S_{\mu}$, there  is an honest implementation of the protocol such that it aborts with low probability.
In order for our protocol to be useful in reality it is important to notice that, as shown in~\cite{MDL}, for all $\mu_{\mathrm{min}} > 0$ the maximal quantum value of $S_{\mu}$ is greater than zero.
Moreover, in Appendix~\ref{apx:max-viol} we explain how to obtain a state that achieves a violation of the MDL inequality.
Thus for all MDL sources with $\mu_{\mathrm{min}} \neq 0$ the entropy bound that we derive later on is non-trivial.

\subsection{Soundeness}
\label{sec:soundness}

In the previous part we showed that our proposed protocol is complete.
Besides that we also want that the protocol does what it is supposed to do, i.e., if it does not abort the outputs should be secret with high probability.
This property, which is sometimes called soundness, is quantified in Definition~\ref{def:RAP-secrecy}.

In the following we prove that Protocol~\ref{alg:RAP} is secret and determine the secrecy parameter $\varepsilon_{\mathrm{RA}}$.
In a first step we derive a lower bound on the smooth min-entropy of the MDL experiments' outputs.
In the second step we show that in our protocol we can make use of the quantum-proof randomness extractors introduced in Section~\ref{sec:extractors}, and then proceed to determine the exact value of $\varepsilon_{\mathrm{RA}}$.

\subsubsection{Lower-bounding the smooth min-entropy}
\label{sec:entropy-bound}

In the first part of the RAP we have the entropy accumulation routine where we perform a number of MDL experiments with our device.
The goal now is to lower bound the smooth min-entropy of the outputs of these experiments.
To achieve this we employ the EAT (introduced in Section~\ref{sec:EAT}) together with the entropy bound for a single MDL experiment that was derived in Section~\ref{sec:single-round}.

In order to apply the entropy accumulation theorem we need a protocol that evolves the states using EAT channels.
In our proposed protocol we have in each round two quantum registers $Q_{\mathrm{A},i}$ and $Q_{\mathrm{B},i}$ holding the quantum state of either of the device's two parts.
Furthermore we have the classical registers $X_i,Y_i$ for the inputs, $A_i,B_i$ for the outputs of the device, and $C_i$ evaluating the outcome of the MDL experiment.
Comparing our registers to Definition~\ref{def:EATchannel}, we can identify $R_i = Q_{\mathrm{A},i}Q_{\mathrm{B},i}$ and $I_i = X_iY_i$, and denote the channels evolving the states in our protocol as 
\begin{align}
	\begin{split}
	\mathcal{N}_i: Q_{\mathrm{A},i-1}Q_{\mathrm{B},i-1} &\rightarrow Q_{\mathrm{A},i}Q_{\mathrm{B},i}A_iB_iX_iY_iC_i \\
	\rho_{Q_{\mathrm{A},i-1}Q_{\mathrm{B},i-1}} &\mapsto \rho_{Q_{\mathrm{A},i}Q_{\mathrm{B},i}A_iB_iX_iY_iC_i} \,.
	\end{split} \label{eq:EAT_channel_protocol_1}
\end{align}

The state after the $n$ rounds of the entropy accumulation part, just before step~\ref{step:abortion} is denoted by
\begin{equation}
	\rho_{A^{n} B^{n} X^{n} Y^{n} C^{n} E} = \left( \textrm{Tr}_{Q_{A,{n}}Q_{B,{n}}} \circ \mathcal{N}_{n} \circ \dots \circ \mathcal{N}_1 \right) \otimes \mathcal{I}_E \rho^0_{Q_{A}Q_{B}E} \label{eq:final_state}
\end{equation}
In step~\ref{step:abortion} Alice and Bob decide whether to abort the protocol or not. We denote by $\Omega$ the event of not aborting,
\begin{equation}
	\Omega = \Big \{\bar{C} \geq (S_{\mathrm{exp}} - \delta_{\mathrm{est}}) \Big \} \, . \label{eq:event_not_abort}
\end{equation}
Combined we denote by $\rho_{A^{n} B^{n} X^{n} Y^{n} C^{n} E|\Omega}$, or short $\rho_{|\Omega}$, the state after the protocol conditioned on not aborting the protocol.

We need to prove that these channels are indeed EAT channels.
\begin{lma}
\label{lma:EAT_channels}
	The channels $\mathcal{N}_i$ that evolve the unknown quantum state of Protocol~\ref{alg:RAP}, are EAT channels, i.e., they satisfy Definition~\ref{def:EATchannel}.
\end{lma}
\begin{proof}
	\begin{enumerate}[1.]
		\item Condition 1. is satisfied because $A_i,B_i,I_i$ represent the (classical, discrete) inputs and outputs of the device that is employed, and $Q_{\mathrm{A},i}Q_{\mathrm{B},i}$ are quantum registers.
		\item Condition 2. is satisfied because $A_i,B_i,I_i$ are classical registers and $C_{i}$ is a classical function of those registers.
		\item As is stated in Section~\ref{sec:assumptions} it holds that $I(A^{i-1}B^{i-1} : I_i | I^{i-1} Em \lambda) = 0$.
		Thus, the Markov chain condition is satisfied.\qedhere
	\end{enumerate}
\end{proof}

Now that we have the necessary preconditions, we can prove a bound on the smooth min-entropy of $A^{n}B^{n}$ given the inputs and the side information.
More precisely, in Theorem~\ref{thm:main}, we lower bound $H^{\varepsilon_\mathrm{s}}_\mathrm{min} (A^{n}B^{n}|X^{n}Y^{n}F^{n}E)_{\rho_{|\Omega}}$ for any $\varepsilon_{\mathrm{s}} \in \{0,1\}$.
In our proof we combine \cite[Lemma~9 and Theorem~10]{RotemEAT} and adapt the proofs to our setting.

The bound is described with the help of the following functions, where $p \in \mathbb{P}_{\mathcal{C}}$:
\begin{align}
	S_\mu(p) &= \mu_\mathrm{min} \cdot p(1) - \mu_\mathrm{max} \cdot p(-1) \,, \label{fct:Smu_of_p} \\
	g_\mu(p) &= \begin{cases} 1 - h \left( \frac{1}{2} + \frac{1}{\mu^*} \sqrt{ S_\mu(p) (S_\mu(p) + \mu^*)} \right) &\text{for } \frac{S_\mu(p)}{\mu^*} \in [0,\frac{\sqrt{2} - 1}{2}) \\ 1 &\text{for } \frac{S_\mu(p)}{\mu^*} \in [\frac{\sqrt{2} - 1}{2}, 1] \end{cases} \,, \\
	a(s_t) &= \D{}{S_{\mu}(p)}g_\mu(p) \Big|_{S_{\mu}(p) = s_t} \quad \text{and} \quad b(s_t) = g(s_t) - a(s_t)\cdot s_t \,. \\
	f_ \mathrm{min}(p, s_t) &= \begin{cases} g_\mu(p) &\text{for } S_{\mu}(p) \leq s_t \\ a(s_t) \cdot S_{\mu}(p) + b(s_t) &\text{for } S_{\mu}(p) > s_t \end{cases} \,, \label{fct:min_tradeoff} \\
	\zeta_{\mu}(s_t, \varepsilon_ \mathrm{s}, \varepsilon_{\mathrm{EA}}) &= 2(\log(9) + a(s_t) \cdot \mu_{\mathrm{max}}) \sqrt{1-2 \log(\varepsilon_ \mathrm{s} \varepsilon_ \mathrm{EA})} \,, \\
	\eta_ \mathrm{opt}(\varepsilon_ \mathrm{s}, \varepsilon_ \mathrm{EA}, S_{\mathrm{exp}} - \delta_{\mathrm{est}}, n, \mu) &= \max_{0 < s_t < \mu^* \cdot \frac{\sqrt{2}-1}{2}} \left[ f_ \mathrm{min}(S_ \mathrm{exp} - \delta_ \mathrm{est}, s_t) - \frac{1}{\sqrt{n}} \zeta_{\mu}(s_t, \varepsilon_ \mathrm{s}, \varepsilon_ \mathrm{EA}) \right] \,. \label{fct:eta_opt}
\end{align}

\begin{thm}[Main]
\label{thm:main}
	Let D be any device, $\rho$ the state (as defined in Equation~\eqref{eq:final_state}) generated using Protocol~\ref{alg:RAP}, $\Omega$ (as defined in Equation~\eqref{eq:event_not_abort}) the event that the protocol does not abort, and $\rho_{|\Omega}$ the state conditioned on not aborting.
	Then, for any $\varepsilon_\mathrm{EA}, \varepsilon_\mathrm{s} \in (0,1)$, either the protocol aborts with probability greater than $1-\varepsilon_\mathrm{EA}$ or
	\begin{equation}
		H^{\varepsilon_\mathrm{s}}_\mathrm{min} (A^{n} B^{n} | X^{n} Y^{n} E)_{\rho_{|\Omega}} > n \cdot \eta_\mathrm{opt} (\varepsilon_\mathrm{s}, \varepsilon_\mathrm{EA}, S_{\mathrm{exp}} - \delta_{\mathrm{est}}, n, \mu) \,
	\end{equation}
	where $\eta_\mathrm{opt}$ is defined in Equation~\eqref{fct:eta_opt}.
\end{thm}

The entropy bound from Theorem~\ref{thm:main} is plotted in Figure~\ref{fig:entropy_rate}.

\begin{proof}[Proof of Theorem~\ref{thm:main}]
	We begin the proof by devising a min-tradeoff function for the EAT channels.
	We then proceed to lower bound the smooth min-entropy by employing the EAT with the given min-tradeoff function.
	\begin{claim} \label{claim:min-tradeoff}
		Let $\{ \mathcal{N}_i \}$ be the set of EAT channels implemented in Protocol~\ref{alg:RAP}. Then, for any $0 < s_t < \mu^* \cdot \frac{\sqrt{2}-1}{2}$, where $\mu^{*} = \mu_{\mathrm{min}} \cdot \mu_{\mathrm{\mathrm{ax}}}$, the function~\eqref{fct:min_tradeoff} is a min-tradeoff function for the set $\{ \mathcal{N}_i \}$.
	\end{claim}
	\begin{proof}[Proof of Claim~\ref{claim:min-tradeoff}]
		Note that, due to Assumption~\ref{item:guess-2}, each $\mathcal{N}_{i}$ describes a single MDL experiment (as described in Chapter~\ref{sec:single-round}).
		Thus, employing the bound from Equation~\eqref{eq:single_round_bound}, it follows directly that		\begin{equation}
			\inf_{\sigma_{R_{i-1}R'}:\mathcal{N}_i(\sigma)_{C_i}=p} H\left( A_i B_i | X_i Y_i R' \right)_{\mathcal{N}_i(\sigma)} \geq 1 - h \left( \frac{1}{2} + \frac{1}{\mu^*} \sqrt{ S_\mu(p) (S_\mu(p) + \mu^*)} \right) \;. \label{eq:one_box_entropy_final}
		\end{equation}
		
		Let $\mathcal{C} = \{ -\mu_{\mathrm{max}}, 0, \mu_{\mathrm{min}} \}$ and define the function $g_\mu$ on $\mathbb{P}_{\mathcal{C}}$ as
		\begin{equation}
			g_\mu(p) = \begin{cases} 1 - h \left( \frac{1}{2} + \frac{1}{\mu^*} \sqrt{ S_\mu(p) (S_\mu(p) + \mu^*)} \right) &\text{for } \frac{S_\mu(p)}{\mu^*} \in [0,\frac{\sqrt{2} - 1}{2}) \\ 1 &\text{for } \frac{S_\mu(p)}{\mu^*} \in [\frac{\sqrt{2} - 1}{2}, 1] \end{cases} \,.
		\end{equation}
		Then any function $f_ \mathrm{min}(p)$, that is differentiable and satisfies $f_ \mathrm{min}(p) \leq g_\mu(p)$ for all $p$, is a min-tradeoff function for the set $\{ \mathcal{N}_i \}$.
		Unfortunately, as $\frac{S_\mu(p)}{\mu^*}$ approaches $\frac{\sqrt{2}-1}{2}$, the derivative of $g_\mu$ diverges.
		Since the bound that we derive later depends on the derivative, we want to avoid this.
		Therefore we linearize $g_\mu$ starting at some point $p_t$ with $\frac{S_\mu(p_t)}{\mu^*} < \frac{\sqrt{2}-1}{2}$ and thus avoid the problem of a diverging derivative.
		
		Consider the change of variables
		\begin{equation}
		\begin{aligned}
			s &= S_\mu(p) \\
			t &= \mu_ \mathrm{max} \cdot p(1) + \mu_ \mathrm{min} \cdot p(-1) \,.
		\end{aligned}
		\end{equation}
		In this orthogonal coordinate system we clearly see that $g_\mu(s,t)$ is independent of $t$ and only changes with $s$.
		Thus we can restrict our attention in analysing $g_\mu(s,t)$ to a slice where $t$ is constant.
		The divergence of the derivative now happens as $\frac{s}{\mu^*}$ approaches $\frac{\sqrt{2}-1}{2}$.
		Hence we linearize $g_\mu$ at some point $s_t < \mu^* \cdot \frac{\sqrt{2}-1}{2}$.
		
		For the linearization we define
		\begin{equation}
			a(s_t) = \D{}{s}g_\mu(s) \Big|_{s_t} \quad \text{and} \quad b(s_t) = g(s_t) - a(s_t)\cdot s_t \,.
		\end{equation}
		Given these constants we can define the function $f_ \mathrm{min}$ as
		\begin{equation}
			f_ \mathrm{min}(s,s_t) = \begin{cases} g_\mu(s) &\text{for } s \leq s_t \\ a(s_t) \cdot s + b(s_t) &\text{for } s > s_t \end{cases} \,.
		\end{equation}
		Note that this is technically not yet a min-tradeoff function, since it is a function taking arguments in $\mathbb{R}$ instead of $\mathbb{P}_{\mathcal{C}}$.
		However, expressing the new variables as a function of $p$ we can get the final min-tradeoff function as
		\begin{equation}
			f_ \mathrm{min}(p,s_t) = \begin{cases} g_\mu(p) &\text{for } s(p) \leq s_t \\ a(s_t) \cdot s(p) + b(s_t) &\text{for } s(p) > s_t \end{cases} \,.
		\end{equation}
		
		Note also that this is a min-tradeoff function for all $0 < s_t < \mu^* \cdot \frac{\sqrt{2}-1}{2}$.
		Hence, when we derive the entropy bound for Protocol~\ref{alg:RAP} we can optimize over the parameter $s_t$ to get the best possible bound.
	\end{proof}	
	
	Now that we have a min-tradeoff function we can continue to derive a lower bound on the smooth min-entropy.
	As stated in Lemma~\ref{lma:EAT_channels} the channels in the protocol are EAT channels and we can employ the EAT.
	Furthermore we realize that the event $\Omega$ of the protocol not aborting implies that the estimation for the MDL violation is at least $S_{\mathrm{exp}} - \delta_{\mathrm{est}}$, i.e., 
	\[
	S_\mu \left(\mathrm{freq}_{C^{n}}\right) \geq S_{\mathrm{exp}} - \delta_{\mathrm{est}} 
	\]
	for any $C^{n}$ for which $\Pr\left[C^{n}\right]_{\rho_{|\Omega}}> 0$.
	Thus, employing the EAT, we find that either the protocol aborts with probability $1-\Pr(\Omega)\geq 1-\varepsilon_{\mathrm{EA}}$ or, the lower bound
	\begin{equation}
		H^{\varepsilon_{\text{s}}}_{\min} \left( A^{n} B^{n} | X^{n} Y^{n} E \right)_{\rho_{|\Omega}} > {n} f_ \mathrm{min}(S_ \mathrm{exp} - \delta_ \mathrm{est}, s_t) - \sqrt{{n}} \zeta_{\mu}(s_t, \varepsilon_ \mathrm{s}, \varepsilon_ \mathrm{EA}) \,,
	\end{equation}
	holds.
	Here we introduced $\zeta_{\mu}(s_t, \varepsilon_ \mathrm{s}, \varepsilon_{\mathrm{EA}}) = 2(\log(9) + a(s_t) \cdot \mu_{ \mathrm{max}}) \sqrt{1-2 \log(\varepsilon_ \mathrm{s} \varepsilon_ \mathrm{EA})}$,
	where we used that $\| \nabla f_{\mathrm{min}} \|_{\infty} = a(s_{t}) \cdot \mu_{\mathrm{max}}$, due to the linearization in the direction of the steepest slope.
	Additionally, in the description of the lower bound we used the argument $S_\mathrm{exp}-\delta_\mathrm{est}$ as shorthand to denote any probability distribution with $S_{\mu}(p) = S_ \mathrm{exp} - \delta_ \mathrm{est}$.
	We can use this abbreviated notation because for all $p$ with fixed $S_{\mu}(p)$, the value of the min-tradeoff function is the same.
	Furthermore, the fact that $f_{\mathrm{min}}(p)$ is constant as long as $S_{\mu}(p)$ is constant is also the reason that, in the EAT, we can set $t = f_ \mathrm{min}(S_ \mathrm{exp} - \delta_ \mathrm{est}, s_t)$ in our lower bound.
	
	Since $s_t$ is chosen arbitrarily, we can optimize over it.
	For the final entropy bound define
	\begin{equation*}
		\eta_ \mathrm{opt}(\varepsilon_ \mathrm{s}, \varepsilon_ \mathrm{EA}, S_{\mathrm{exp}} - \delta_{\mathrm{est}}, n, \mu) = \max_{0 < s_t < \mu^* \cdot \frac{\sqrt{2}-1}{2}} \left[ f_ \mathrm{min}(S_ \mathrm{exp} - \delta_ \mathrm{est}, s_t) - \frac{1}{\sqrt{{n}}} \zeta_{\mu}(s_t, \varepsilon_ \mathrm{s}, \varepsilon_ \mathrm{EA}) \right] \,.
	\end{equation*}
	Thus, the entropy bound reduces to
	\begin{equation*}
		H^{\varepsilon_{\text{s}}}_{\min} \left( A^{n} B^{n} | X^{n} Y^{n} E \right)_{\rho_{|\Omega}} > {n} \cdot \eta_ \mathrm{opt} (\varepsilon_ \mathrm{s}, \varepsilon_ \mathrm{EA}, S_{\mathrm{exp}} - \delta_{\mathrm{est}}, n, \mu) \,. \qedhere
	\end{equation*}
\end{proof}

As stated before, in Figure~\ref{fig:entropy_rate} the entropy rate, $\eta_{\mathrm{opt}}$, is plotted for several different parameters of the RAP.
In Figure~\ref{fig:entropy_rate_n} the asymptotic rates are equal to the single round entropy bounds of Figure~\ref{fig:entropy_single_round_S} with corresponding $\mu$.
Furthermore, we observe that, as expected, the entropy rate decreases for a decreasing number of rounds, $n$.
If the number of rounds decreases below a certain threshold, we do not get a non-trivial (positive) entropy bound anymore.

In Figure~\ref{fig:entropy_rate_mu} we see that, as was the case for a single MDL experiment, the maximal entropy (rate) decreases as the bias of the source increases.
Moreover, we observe that, as $n$ decreases, the minimal MDL violation to achieve a non-trivial entropy bound increases.
Therefore, we can compensate imperfections in the implementation (which lead to a decreasing MDL violation) with an increasing number of rounds.

\begin{figure}
\centering
	\subfloat[Entropy rate, $\eta_{\mathrm{opt}}$, as function of $S_{\mathrm{exp}}$ for different $n$ and, in each plot, fixed $\mu$.]{
	\begin{tikzpicture}
		\begin{axis}[
			title=$\mu_{\mathrm{min}} \stackrel{}{=} 0.250 \text{, } \mu_{\mathrm{max}} \stackrel{}{=} 0.250$,
			xlabel=$S_{\mathrm{exp}}$,
			ylabel=$\eta_{\mathrm{opt}}$,
			xmin=0,
			xmax=0.012944,
			ymax=1,
			ymin=0,
			ytick={0,0.2,0.4,0.6,0.8,1},
			legend style={at={(0.33, 0.97)},anchor=north,legend cell align=left,font=\FONTSZ} 
			]
		
		
			\addplot[plot1,very thick,smooth] coordinates {
			(0.0001, -0.00902) (0.00029, -0.00043) (0.00047, 0.00823) (0.00066, 0.01696) (0.00084, 0.02579) (0.00103, 0.0347) (0.00122, 0.04371) (0.0014, 0.05281) (0.00159, 0.062) (0.00178, 0.07128) (0.00196, 0.08066) (0.00215, 0.09014) (0.00233, 0.09971) (0.00252, 0.10939) (0.00271, 0.11916) (0.00289, 0.12904) (0.00308, 0.13903) (0.00326, 0.14912) (0.00345, 0.15932) (0.00364, 0.16963) (0.00382, 0.18006) (0.00401, 0.1906) (0.0042, 0.20125) (0.00438, 0.21203) (0.00457, 0.22293) (0.00475, 0.23395) (0.00494, 0.2451) (0.00513, 0.25638) (0.00531, 0.26779) (0.0055, 0.27934) (0.00568, 0.29102) (0.00587, 0.30285) (0.00606, 0.31483) (0.00624, 0.32695) (0.00643, 0.33923) (0.00662, 0.35166) (0.0068, 0.36425) (0.00699, 0.37702) (0.00717, 0.38995) (0.00736, 0.40306) (0.00755, 0.41635) (0.00773, 0.42983) (0.00792, 0.44351) (0.0081, 0.45738) (0.00829, 0.47146) (0.00848, 0.48577) (0.00866, 0.50029) (0.00885, 0.51505) (0.00904, 0.53006) (0.00922, 0.54531) (0.00941, 0.56084) (0.00959, 0.57665) (0.00978, 0.59275) (0.00997, 0.60917) (0.01015, 0.62591) (0.01034, 0.64301) (0.01052, 0.66048) (0.01071, 0.67834) (0.0109, 0.69664) (0.01108, 0.71541) (0.01127, 0.73468) (0.01145, 0.75452) (0.01164, 0.77497) (0.01183, 0.79611) (0.01201, 0.81805) (0.0122, 0.84092) (0.01239, 0.86488) (0.01257, 0.8902) (0.01276, 0.9173) (0.01294, 0.947)
			};
			
			\addplot[plot2,very thick,smooth] coordinates {
			(0.0001, -0.01276) (0.00029, -0.00417) (0.00047, 0.00446) (0.00066, 0.01317) (0.00084, 0.02196) (0.00103, 0.03085) (0.00122, 0.03982) (0.0014, 0.04889) (0.00159, 0.05805) (0.00178, 0.0673) (0.00196, 0.07664) (0.00215, 0.08609) (0.00233, 0.09563) (0.00252, 0.10527) (0.00271, 0.11501) (0.00289, 0.12486) (0.00308, 0.13481) (0.00326, 0.14486) (0.00345, 0.15502) (0.00364, 0.1653) (0.00382, 0.17568) (0.00401, 0.18618) (0.0042, 0.1968) (0.00438, 0.20754) (0.00457, 0.21839) (0.00475, 0.22937) (0.00494, 0.24048) (0.00513, 0.25171) (0.00531, 0.26308) (0.0055, 0.27458) (0.00568, 0.28622) (0.00587, 0.298) (0.00606, 0.30992) (0.00624, 0.32199) (0.00643, 0.33421) (0.00662, 0.34659) (0.0068, 0.35913) (0.00699, 0.37184) (0.00717, 0.38471) (0.00736, 0.39776) (0.00755, 0.41099) (0.00773, 0.42441) (0.00792, 0.43801) (0.0081, 0.45182) (0.00829, 0.46583) (0.00848, 0.48006) (0.00866, 0.4945) (0.00885, 0.50918) (0.00904, 0.5241) (0.00922, 0.53927) (0.00941, 0.55471) (0.00959, 0.57042) (0.00978, 0.58642) (0.00997, 0.60273) (0.01015, 0.61936) (0.01034, 0.63633) (0.01052, 0.65367) (0.01071, 0.6714) (0.0109, 0.68955) (0.01108, 0.70815) (0.01127, 0.72724) (0.01145, 0.74688) (0.01164, 0.76711) (0.01183, 0.78801) (0.01201, 0.80967) (0.0122, 0.8322) (0.01239, 0.85576) (0.01257, 0.88058) (0.01276, 0.90701) (0.01294, 0.93569)
			};
			
			\addplot[plot3,very thick,smooth] coordinates {
			(0.0001, -0.02854) (0.00029, -0.01994) (0.00047, -0.01134) (0.00066, -0.00274) (0.00084, 0.00593) (0.00103, 0.01469) (0.00122, 0.02353) (0.0014, 0.03247) (0.00159, 0.04149) (0.00178, 0.05061) (0.00196, 0.05982) (0.00215, 0.06912) (0.00233, 0.07853) (0.00252, 0.08802) (0.00271, 0.09762) (0.00289, 0.10732) (0.00308, 0.11712) (0.00326, 0.12702) (0.00345, 0.13703) (0.00364, 0.14714) (0.00382, 0.15737) (0.00401, 0.16771) (0.0042, 0.17815) (0.00438, 0.18872) (0.00457, 0.1994) (0.00475, 0.21021) (0.00494, 0.22113) (0.00513, 0.23218) (0.00531, 0.24336) (0.0055, 0.25467) (0.00568, 0.26611) (0.00587, 0.27768) (0.00606, 0.2894) (0.00624, 0.30126) (0.00643, 0.31327) (0.00662, 0.32542) (0.0068, 0.33773) (0.00699, 0.3502) (0.00717, 0.36283) (0.00736, 0.37563) (0.00755, 0.3886) (0.00773, 0.40175) (0.00792, 0.41509) (0.0081, 0.42861) (0.00829, 0.44233) (0.00848, 0.45625) (0.00866, 0.47038) (0.00885, 0.48473) (0.00904, 0.4993) (0.00922, 0.51411) (0.00941, 0.52917) (0.00959, 0.54449) (0.00978, 0.56008) (0.00997, 0.57595) (0.01015, 0.59212) (0.01034, 0.60861) (0.01052, 0.62542) (0.01071, 0.6426) (0.0109, 0.66015) (0.01108, 0.67811) (0.01127, 0.69651) (0.01145, 0.71539) (0.01164, 0.73478) (0.01183, 0.75474) (0.01201, 0.77534) (0.0122, 0.79665) (0.01239, 0.81878) (0.01257, 0.84186) (0.01276, 0.86609) (0.01294, 0.89176)
			};
			
			\addplot[plot4,very thick,smooth] coordinates {
			(0.001, -0.04036) (0.00117, -0.03236) (0.00135, -0.02437) (0.00152, -0.01637) (0.00169, -0.00838) (0.00187, -0.00035) (0.00204, 0.00776) (0.00221, 0.01594) (0.00238, 0.0242) (0.00256, 0.03253) (0.00273, 0.04095) (0.0029, 0.04944) (0.00308, 0.05802) (0.00325, 0.06667) (0.00342, 0.07541) (0.0036, 0.08424) (0.00377, 0.09315) (0.00394, 0.10214) (0.00412, 0.11123) (0.00429, 0.1204) (0.00446, 0.12966) (0.00464, 0.13902) (0.00481, 0.14847) (0.00498, 0.15801) (0.00515, 0.16765) (0.00533, 0.17739) (0.0055, 0.18723) (0.00567, 0.19717) (0.00585, 0.20722) (0.00602, 0.21737) (0.00619, 0.22763) (0.00637, 0.23799) (0.00654, 0.24847) (0.00671, 0.25907) (0.00689, 0.26978) (0.00706, 0.28061) (0.00723, 0.29157) (0.0074, 0.30264) (0.00758, 0.31385) (0.00775, 0.32518) (0.00792, 0.33666) (0.0081, 0.34826) (0.00827, 0.36001) (0.00844, 0.3719) (0.00862, 0.38395) (0.00879, 0.39614) (0.00896, 0.40849) (0.00914, 0.42101) (0.00931, 0.43369) (0.00948, 0.44654) (0.00966, 0.45958) (0.00983, 0.47279) (0.01, 0.4862) (0.01017, 0.49981) (0.01035, 0.51362) (0.01052, 0.52765) (0.01069, 0.5419) (0.01087, 0.55639) (0.01104, 0.57112) (0.01121, 0.5861) (0.01139, 0.60136) (0.01156, 0.6169) (0.01173, 0.63274) (0.01191, 0.6489) (0.01208, 0.6654) (0.01225, 0.68226) (0.01242, 0.69951) (0.0126, 0.71718) (0.01277, 0.73531) (0.01294, 0.75395)
			};
			
			\addplot[plot5,very thick,smooth] coordinates {
			(0.001, -0.09024) (0.00117, -0.08224) (0.00135, -0.07425) (0.00152, -0.06626) (0.00169, -0.05826) (0.00187, -0.05027) (0.00204, -0.04227) (0.00221, -0.03428) (0.00238, -0.02629) (0.00256, -0.01829) (0.00273, -0.01025) (0.0029, -0.00214) (0.00308, 0.00605) (0.00325, 0.01432) (0.00342, 0.02266) (0.0036, 0.03109) (0.00377, 0.03959) (0.00394, 0.04817) (0.00412, 0.05684) (0.00429, 0.06559) (0.00446, 0.07442) (0.00464, 0.08334) (0.00481, 0.09234) (0.00498, 0.10143) (0.00515, 0.11062) (0.00533, 0.11989) (0.0055, 0.12925) (0.00567, 0.13871) (0.00585, 0.14827) (0.00602, 0.15792) (0.00619, 0.16767) (0.00637, 0.17752) (0.00654, 0.18747) (0.00671, 0.19752) (0.00689, 0.20769) (0.00706, 0.21796) (0.00723, 0.22834) (0.0074, 0.23883) (0.00758, 0.24944) (0.00775, 0.26016) (0.00792, 0.271) (0.0081, 0.28197) (0.00827, 0.29306) (0.00844, 0.30428) (0.00862, 0.31563) (0.00879, 0.32711) (0.00896, 0.33874) (0.00914, 0.3505) (0.00931, 0.36241) (0.00948, 0.37446) (0.00966, 0.38668) (0.00983, 0.39904) (0.01, 0.41158) (0.01017, 0.42428) (0.01035, 0.43715) (0.01052, 0.4502) (0.01069, 0.46344) (0.01087, 0.47687) (0.01104, 0.49049) (0.01121, 0.50433) (0.01139, 0.51838) (0.01156, 0.53266) (0.01173, 0.54717) (0.01191, 0.56192) (0.01208, 0.57694) (0.01225, 0.59222) (0.01242, 0.60779) (0.0126, 0.62367) (0.01277, 0.63986) (0.01294, 0.65639) 
			};
			
			\addplot[plot6,very thick,smooth] coordinates {
			(0.001, -0.12762) (0.00117, -0.11962) (0.00135, -0.11163) (0.00152, -0.10363) (0.00169, -0.09564) (0.00187, -0.08765) (0.00204, -0.07965) (0.00221, -0.07166) (0.00238, -0.06366) (0.00256, -0.05567) (0.00273, -0.04768) (0.0029, -0.03968) (0.00308, -0.03169) (0.00325, -0.02369) (0.00342, -0.01562) (0.0036, -0.00748) (0.00377, 0.00073) (0.00394, 0.00902) (0.00412, 0.01739) (0.00429, 0.02584) (0.00446, 0.03437) (0.00464, 0.04298) (0.00481, 0.05168) (0.00498, 0.06045) (0.00515, 0.06931) (0.00533, 0.07826) (0.0055, 0.08729) (0.00567, 0.09642) (0.00585, 0.10563) (0.00602, 0.11493) (0.00619, 0.12433) (0.00637, 0.13382) (0.00654, 0.1434) (0.00671, 0.15309) (0.00689, 0.16287) (0.00706, 0.17275) (0.00723, 0.18274) (0.0074, 0.19283) (0.00758, 0.20303) (0.00775, 0.21333) (0.00792, 0.22375) (0.0081, 0.23428) (0.00827, 0.24493) (0.00844, 0.25569) (0.00862, 0.26657) (0.00879, 0.27758) (0.00896, 0.28871) (0.00914, 0.29998) (0.00931, 0.31137) (0.00948, 0.3229) (0.00966, 0.33457) (0.00983, 0.34638) (0.01, 0.35834) (0.01017, 0.37044) (0.01035, 0.38271) (0.01052, 0.39513) (0.01069, 0.40772) (0.01087, 0.42047) (0.01104, 0.4334) (0.01121, 0.44652) (0.01139, 0.45982) (0.01156, 0.47331) (0.01173, 0.48701) (0.01191, 0.50091) (0.01208, 0.51504) (0.01225, 0.52939) (0.01242, 0.54398) (0.0126, 0.55882) (0.01277, 0.57392) (0.01294, 0.5893) 
			};

			\addplot[black,very thick,smooth,dashed] coordinates {
			(0.0, -1e-05) (0.00019, 0.0087) (0.00038, 0.01749) (0.00056, 0.02637) (0.00075, 0.03534) (0.00094, 0.0444) (0.00113, 0.05356) (0.00131, 0.06281) (0.0015, 0.07215) (0.00169, 0.08159) (0.00188, 0.09113) (0.00206, 0.10077) (0.00225, 0.11051) (0.00244, 0.12036) (0.00263, 0.13031) (0.00281, 0.14036) (0.003, 0.15053) (0.00319, 0.1608) (0.00338, 0.17119) (0.00356, 0.18169) (0.00375, 0.19231) (0.00394, 0.20305) (0.00413, 0.21391) (0.00431, 0.22489) (0.0045, 0.236) (0.00469, 0.24724) (0.00488, 0.25861) (0.00507, 0.27011) (0.00525, 0.28175) (0.00544, 0.29354) (0.00563, 0.30546) (0.00582, 0.31754) (0.006, 0.32977) (0.00619, 0.34215) (0.00638, 0.3547) (0.00657, 0.3674) (0.00675, 0.38028) (0.00694, 0.39334) (0.00713, 0.40657) (0.00732, 0.41999) (0.0075, 0.4336) (0.00769, 0.4474) (0.00788, 0.46142) (0.00807, 0.47565) (0.00825, 0.49009) (0.00844, 0.50477) (0.00863, 0.51969) (0.00882, 0.53486) (0.009, 0.55029) (0.00919, 0.56599) (0.00938, 0.58198) (0.00957, 0.59828) (0.00976, 0.61489) (0.00994, 0.63185) (0.01013, 0.64916) (0.01032, 0.66686) (0.01051, 0.68498) (0.01069, 0.70354) (0.01088, 0.72259) (0.01107, 0.74216) (0.01126, 0.76232) (0.01144, 0.78313) (0.01163, 0.80467) (0.01182, 0.82705) (0.01201, 0.85042) (0.01219, 0.87498) (0.01238, 0.90103) (0.01257, 0.92912) (0.01276, 0.96034) (0.01294, 0.99976) 
			};
		\end{axis} 
	\end{tikzpicture}
	~
	\begin{tikzpicture}
		\begin{axis}[
			title=$\mu_{\mathrm{min}} \stackrel{}{=} 0.210 \text{, } \mu_{\mathrm{max}} \stackrel{}{=} 0.371$,
			xlabel=$S_{\mathrm{exp}}$,
			ylabel=$\eta_{\mathrm{opt}}$,
			xmin=0,
			xmax=0.00462,
			ymax=1,
			ymin=0,
			ytick={0,0.2,0.4,0.6,0.8,1},
			legend style={at={(0.34, 0.97)},anchor=north,legend cell align=left,font=\FONTSZ} 
			]
		
		
			\addplot[plot1,very thick,smooth] coordinates {
			(0.0001, -0.01039) (0.00017, -0.0078) (0.00024, -0.00521) (0.00031, -0.00262) (0.00038, -3e-05) (0.00045, 0.00256) (0.00052, 0.00517) (0.00059, 0.00778) (0.00066, 0.0104) (0.00073, 0.01303) (0.0008, 0.01567) (0.00087, 0.01831) (0.00094, 0.02096) (0.00101, 0.02362) (0.00108, 0.02629) (0.00115, 0.02897) (0.00122, 0.03165) (0.00129, 0.03434) (0.00136, 0.03704) (0.00143, 0.03975) (0.0015, 0.04247) (0.00157, 0.0452) (0.00163, 0.04793) (0.0017, 0.05067) (0.00177, 0.05342) (0.00184, 0.05618) (0.00191, 0.05895) (0.00198, 0.06173) (0.00205, 0.06451) (0.00212, 0.0673) (0.00219, 0.0701) (0.00226, 0.07291) (0.00233, 0.07573) (0.0024, 0.07856) (0.00247, 0.0814) (0.00254, 0.08424) (0.00261, 0.0871) (0.00268, 0.08996) (0.00275, 0.09283) (0.00282, 0.09571) (0.00289, 0.0986) (0.00296, 0.1015) (0.00303, 0.10441) (0.0031, 0.10733) (0.00317, 0.11025) (0.00324, 0.11319) (0.00331, 0.11613) (0.00338, 0.11909) (0.00345, 0.12205) (0.00352, 0.12502) (0.00359, 0.12801) (0.00366, 0.131) (0.00373, 0.134) (0.0038, 0.13701) (0.00387, 0.14003) (0.00394, 0.14306) (0.00401, 0.1461) (0.00408, 0.14915) (0.00415, 0.15221) (0.00422, 0.15528) (0.00429, 0.15836) (0.00436, 0.16145) (0.00443, 0.16455) (0.0045, 0.16766) (0.00457, 0.17078) (0.00464, 0.17391) (0.0047, 0.17706) (0.00477, 0.18021) (0.00484, 0.18337) (0.00491, 0.18654) 
			};
			\addlegendentry{$n = 10^9, \delta_\mathrm{est} = 10^{-4}$}
			
			\addplot[plot2,very thick,smooth] coordinates {
			(0.0001, -0.01469) (0.00017, -0.0121) (0.00024, -0.00951) (0.00031, -0.00693) (0.00038, -0.00434) (0.00045, -0.00175) (0.00052, 0.00084) (0.00059, 0.00345) (0.00066, 0.00606) (0.00073, 0.00867) (0.0008, 0.0113) (0.00087, 0.01393) (0.00094, 0.01657) (0.00101, 0.01922) (0.00108, 0.02188) (0.00115, 0.02455) (0.00122, 0.02722) (0.00129, 0.0299) (0.00136, 0.03259) (0.00143, 0.03529) (0.0015, 0.038) (0.00157, 0.04071) (0.00163, 0.04343) (0.0017, 0.04616) (0.00177, 0.0489) (0.00184, 0.05165) (0.00191, 0.05441) (0.00198, 0.05717) (0.00205, 0.05994) (0.00212, 0.06273) (0.00219, 0.06552) (0.00226, 0.06831) (0.00233, 0.07112) (0.0024, 0.07394) (0.00247, 0.07676) (0.00254, 0.07959) (0.00261, 0.08244) (0.00268, 0.08529) (0.00275, 0.08815) (0.00282, 0.09102) (0.00289, 0.09389) (0.00296, 0.09678) (0.00303, 0.09968) (0.0031, 0.10258) (0.00317, 0.1055) (0.00324, 0.10842) (0.00331, 0.11135) (0.00338, 0.11429) (0.00345, 0.11724) (0.00352, 0.1202) (0.00359, 0.12317) (0.00366, 0.12615) (0.00373, 0.12914) (0.0038, 0.13214) (0.00387, 0.13515) (0.00394, 0.13816) (0.00401, 0.14119) (0.00408, 0.14423) (0.00415, 0.14728) (0.00422, 0.15033) (0.00429, 0.1534) (0.00436, 0.15647) (0.00443, 0.15956) (0.0045, 0.16266) (0.00457, 0.16576) (0.00464, 0.16888) (0.0047, 0.17201) (0.00477, 0.17514) (0.00484, 0.17829) (0.00491, 0.18145) 
			};
			\addlegendentry{$n = 5 \cdot 10^8, \delta_\mathrm{est} = 10^{-4}$}
			
			\addplot[plot3,very thick,smooth] coordinates {
			(0.0001, -0.03285) (0.00017, -0.03026) (0.00024, -0.02767) (0.00031, -0.02508) (0.00038, -0.0225) (0.00045, -0.01991) (0.00052, -0.01732) (0.00059, -0.01473) (0.00066, -0.01214) (0.00073, -0.00956) (0.0008, -0.00697) (0.00087, -0.00438) (0.00094, -0.00178) (0.00101, 0.00082) (0.00108, 0.00343) (0.00115, 0.00605) (0.00122, 0.00868) (0.00129, 0.01132) (0.00136, 0.01396) (0.00143, 0.01661) (0.0015, 0.01927) (0.00157, 0.02194) (0.00163, 0.02462) (0.0017, 0.0273) (0.00177, 0.02999) (0.00184, 0.03269) (0.00191, 0.0354) (0.00198, 0.03812) (0.00205, 0.04084) (0.00212, 0.04358) (0.00219, 0.04632) (0.00226, 0.04907) (0.00233, 0.05183) (0.0024, 0.05459) (0.00247, 0.05737) (0.00254, 0.06015) (0.00261, 0.06295) (0.00268, 0.06575) (0.00275, 0.06856) (0.00282, 0.07138) (0.00289, 0.0742) (0.00296, 0.07704) (0.00303, 0.07988) (0.0031, 0.08274) (0.00317, 0.0856) (0.00324, 0.08847) (0.00331, 0.09135) (0.00338, 0.09424) (0.00345, 0.09714) (0.00352, 0.10005) (0.00359, 0.10296) (0.00366, 0.10589) (0.00373, 0.10882) (0.0038, 0.11177) (0.00387, 0.11472) (0.00394, 0.11769) (0.00401, 0.12066) (0.00408, 0.12364) (0.00415, 0.12663) (0.00422, 0.12963) (0.00429, 0.13264) (0.00436, 0.13566) (0.00443, 0.13869) (0.0045, 0.14173) (0.00457, 0.14478) (0.00464, 0.14784) (0.0047, 0.15091) (0.00477, 0.15399) (0.00484, 0.15708) (0.00491, 0.16018)
			};
			\addlegendentry{$n = 10^8, \delta_\mathrm{est} = 10^{-4}$}
			
			\addplot[plot4,very thick,smooth] coordinates {
			(0.001, -0.04645) (0.00106, -0.04435) (0.00111, -0.04225) (0.00117, -0.04014) (0.00123, -0.03804) (0.00128, -0.03593) (0.00134, -0.03383) (0.0014, -0.03173) (0.00145, -0.02962) (0.00151, -0.02752) (0.00157, -0.02541) (0.00162, -0.02331) (0.00168, -0.0212) (0.00174, -0.0191) (0.00179, -0.017) (0.00185, -0.01489) (0.00191, -0.01279) (0.00196, -0.01068) (0.00202, -0.00858) (0.00208, -0.00647) (0.00213, -0.00436) (0.00219, -0.00224) (0.00225, -0.00012) (0.0023, 0.002) (0.00236, 0.00414) (0.00242, 0.00627) (0.00247, 0.00842) (0.00253, 0.01057) (0.00259, 0.01272) (0.00265, 0.01488) (0.0027, 0.01704) (0.00276, 0.01921) (0.00282, 0.02139) (0.00287, 0.02357) (0.00293, 0.02575) (0.00299, 0.02794) (0.00304, 0.03014) (0.0031, 0.03234) (0.00316, 0.03455) (0.00321, 0.03676) (0.00327, 0.03898) (0.00333, 0.0412) (0.00338, 0.04343) (0.00344, 0.04566) (0.0035, 0.0479) (0.00355, 0.05015) (0.00361, 0.0524) (0.00367, 0.05466) (0.00372, 0.05692) (0.00378, 0.05919) (0.00384, 0.06146) (0.00389, 0.06374) (0.00395, 0.06603) (0.00401, 0.06832) (0.00406, 0.07061) (0.00412, 0.07292) (0.00418, 0.07522) (0.00423, 0.07754) (0.00429, 0.07986) (0.00435, 0.08218) (0.0044, 0.08452) (0.00446, 0.08685) (0.00452, 0.0892) (0.00457, 0.09155) (0.00463, 0.0939) (0.00469, 0.09626) (0.00474, 0.09863) (0.0048, 0.101) (0.00486, 0.10338) (0.00491, 0.10577)
			};
			\addlegendentry{$n = 5 \cdot 10^7, \delta_\mathrm{est} = 10^{-3}$}
			
			\addplot[plot5,very thick,smooth] coordinates {
			(0.001, -0.10387) (0.00106, -0.10177) (0.00111, -0.09967) (0.00117, -0.09756) (0.00123, -0.09546) (0.00128, -0.09335) (0.00134, -0.09125) (0.0014, -0.08915) (0.00145, -0.08704) (0.00151, -0.08494) (0.00157, -0.08283) (0.00162, -0.08073) (0.00168, -0.07862) (0.00174, -0.07652) (0.00179, -0.07442) (0.00185, -0.07231) (0.00191, -0.07021) (0.00196, -0.0681) (0.00202, -0.066) (0.00208, -0.0639) (0.00213, -0.06179) (0.00219, -0.05969) (0.00225, -0.05758) (0.0023, -0.05548) (0.00236, -0.05338) (0.00242, -0.05127) (0.00247, -0.04917) (0.00253, -0.04706) (0.00259, -0.04496) (0.00265, -0.04285) (0.0027, -0.04075) (0.00276, -0.03865) (0.00282, -0.03654) (0.00287, -0.03444) (0.00293, -0.03233) (0.00299, -0.03023) (0.00304, -0.02813) (0.0031, -0.02602) (0.00316, -0.02392) (0.00321, -0.02181) (0.00327, -0.01971) (0.00333, -0.0176) (0.00338, -0.01549) (0.00344, -0.01338) (0.0035, -0.01126) (0.00355, -0.00913) (0.00361, -0.007) (0.00367, -0.00486) (0.00372, -0.00272) (0.00378, -0.00057) (0.00384, 0.00158) (0.00389, 0.00374) (0.00395, 0.0059) (0.00401, 0.00807) (0.00406, 0.01024) (0.00412, 0.01242) (0.00418, 0.01461) (0.00423, 0.0168) (0.00429, 0.01899) (0.00435, 0.02119) (0.0044, 0.0234) (0.00446, 0.02561) (0.00452, 0.02783) (0.00457, 0.03005) (0.00463, 0.03228) (0.00469, 0.03451) (0.00474, 0.03675) (0.0048, 0.03899) (0.00486, 0.04124) (0.00491, 0.0435)
			};
			\addlegendentry{$n = 10^7, \delta_\mathrm{est} = 10^{-3}$}
			
			\addplot[plot6,very thick,smooth] coordinates {
			(0.001, -0.1469) (0.00106, -0.1448) (0.00111, -0.14269) (0.00117, -0.14059) (0.00123, -0.13848) (0.00128, -0.13638) (0.00134, -0.13428) (0.0014, -0.13217) (0.00145, -0.13007) (0.00151, -0.12796) (0.00157, -0.12586) (0.00162, -0.12375) (0.00168, -0.12165) (0.00174, -0.11955) (0.00179, -0.11744) (0.00185, -0.11534) (0.00191, -0.11323) (0.00196, -0.11113) (0.00202, -0.10903) (0.00208, -0.10692) (0.00213, -0.10482) (0.00219, -0.10271) (0.00225, -0.10061) (0.0023, -0.09851) (0.00236, -0.0964) (0.00242, -0.0943) (0.00247, -0.09219) (0.00253, -0.09009) (0.00259, -0.08798) (0.00265, -0.08588) (0.0027, -0.08378) (0.00276, -0.08167) (0.00282, -0.07957) (0.00287, -0.07746) (0.00293, -0.07536) (0.00299, -0.07326) (0.00304, -0.07115) (0.0031, -0.06905) (0.00316, -0.06694) (0.00321, -0.06484) (0.00327, -0.06274) (0.00333, -0.06063) (0.00338, -0.05853) (0.00344, -0.05642) (0.0035, -0.05432) (0.00355, -0.05222) (0.00361, -0.05011) (0.00367, -0.04801) (0.00372, -0.0459) (0.00378, -0.0438) (0.00384, -0.04169) (0.00389, -0.03959) (0.00395, -0.03749) (0.00401, -0.03538) (0.00406, -0.03328) (0.00412, -0.03117) (0.00418, -0.02907) (0.00423, -0.02697) (0.00429, -0.02486) (0.00435, -0.02275) (0.0044, -0.02063) (0.00446, -0.0185) (0.00452, -0.01638) (0.00457, -0.01424) (0.00463, -0.0121) (0.00469, -0.00996) (0.00474, -0.00781) (0.0048, -0.00566) (0.00486, -0.0035) (0.00491, -0.00133)
			};
			\addlegendentry{$n = 5 \cdot 10^6, \delta_\mathrm{est} = 10^{-3}$}

			\addplot[black,very thick,smooth,dashed] coordinates {
			(0.0, -1e-05) (7e-05, 0.00263) (0.00014, 0.00529) (0.00021, 0.00795) (0.00028, 0.01062) (0.00036, 0.0133) (0.00043, 0.01598) (0.0005, 0.01868) (0.00057, 0.02138) (0.00064, 0.02409) (0.00071, 0.02681) (0.00078, 0.02954) (0.00085, 0.03227) (0.00093, 0.03502) (0.001, 0.03777) (0.00107, 0.04053) (0.00114, 0.0433) (0.00121, 0.04608) (0.00128, 0.04887) (0.00135, 0.05166) (0.00142, 0.05447) (0.0015, 0.05728) (0.00157, 0.06011) (0.00164, 0.06294) (0.00171, 0.06578) (0.00178, 0.06863) (0.00185, 0.07148) (0.00192, 0.07435) (0.00199, 0.07723) (0.00207, 0.08011) (0.00214, 0.08301) (0.00221, 0.08591) (0.00228, 0.08882) (0.00235, 0.09174) (0.00242, 0.09468) (0.00249, 0.09762) (0.00256, 0.10057) (0.00264, 0.10352) (0.00271, 0.10649) (0.00278, 0.10947) (0.00285, 0.11246) (0.00292, 0.11546) (0.00299, 0.11846) (0.00306, 0.12148) (0.00313, 0.12451) (0.0032, 0.12755) (0.00328, 0.13059) (0.00335, 0.13365) (0.00342, 0.13671) (0.00349, 0.13979) (0.00356, 0.14288) (0.00363, 0.14597) (0.0037, 0.14908) (0.00377, 0.1522) (0.00385, 0.15532) (0.00392, 0.15846) (0.00399, 0.16161) (0.00406, 0.16477) (0.00413, 0.16794) (0.0042, 0.17112) (0.00427, 0.17431) (0.00434, 0.17751) (0.00442, 0.18072) (0.00449, 0.18394) (0.00456, 0.18718) (0.00463, 0.19042) (0.0047, 0.19368) (0.00477, 0.19694) (0.00484, 0.20022) (0.00491, 0.20351)
			};
			\addlegendentry{asymptotic i.i.d rate}
		\end{axis}
	\end{tikzpicture}
	\label{fig:entropy_rate_n}
	}
	\\
	\subfloat[Entropy rate, $\eta_{\mathrm{opt}}$, as function of $S_{\mathrm{exp}}$ for different $\mu$ and, in each plot, fixed $n$.]{
	\begin{tikzpicture}
		\begin{axis}[
			title=$n \stackrel{}{=} 10^{11} \text{, } \delta_\mathrm{est} \stackrel{}{=} 10^{-4}$,
			xlabel=$S_{\mathrm{exp}}$,
			ylabel=$\eta_{\mathrm{opt}}$,
			xmin=0,
			xmax=0.01293,
			ymax=1,
			ymin=0,
			ytick={0,0.2,0.4,0.6,0.8,1},
			legend style={at={(0.29, 0.97)},anchor=north,legend cell align=left,font=\FONTSZ} 
			]
		
		
			\addplot[plot1,very thick,smooth] coordinates {
			(0.0001, -0.0009) (0.00029, 0.00773) (0.00047, 0.01644) (0.00066, 0.02525) (0.00084, 0.03414) (0.00103, 0.04312) (0.00122, 0.05219) (0.0014, 0.06136) (0.00159, 0.07062) (0.00178, 0.07997) (0.00196, 0.08942) (0.00215, 0.09897) (0.00233, 0.10862) (0.00252, 0.11837) (0.00271, 0.12823) (0.00289, 0.13818) (0.00308, 0.14825) (0.00326, 0.15842) (0.00345, 0.1687) (0.00364, 0.1791) (0.00382, 0.18961) (0.00401, 0.20023) (0.0042, 0.21098) (0.00438, 0.22185) (0.00457, 0.23284) (0.00475, 0.24395) (0.00494, 0.2552) (0.00513, 0.26657) (0.00531, 0.27809) (0.0055, 0.28974) (0.00568, 0.30153) (0.00587, 0.31346) (0.00606, 0.32555) (0.00624, 0.33778) (0.00643, 0.35018) (0.00662, 0.36273) (0.0068, 0.37545) (0.00699, 0.38834) (0.00717, 0.4014) (0.00736, 0.41464) (0.00755, 0.42807) (0.00773, 0.4417) (0.00792, 0.45552) (0.0081, 0.46955) (0.00829, 0.48379) (0.00848, 0.49826) (0.00866, 0.51296) (0.00885, 0.5279) (0.00904, 0.54309) (0.00922, 0.55855) (0.00941, 0.57428) (0.00959, 0.59031) (0.00978, 0.60664) (0.00997, 0.6233) (0.01015, 0.6403) (0.01034, 0.65767) (0.01052, 0.67543) (0.01071, 0.69361) (0.0109, 0.71226) (0.01108, 0.73139) (0.01127, 0.75108) (0.01145, 0.77136) (0.01164, 0.79232) (0.01183, 0.81404) (0.01201, 0.83665) (0.0122, 0.86031) (0.01239, 0.88524) (0.01257, 0.91183) (0.01276, 0.94074) (0.01294, 0.97362) 
			};
			
			\addplot[plot2,very thick,smooth] coordinates {
			(0.0001, -0.00104) (0.00017, 0.00155) (0.00024, 0.00415) (0.00031, 0.00675) (0.00038, 0.00936) (0.00045, 0.01198) (0.00052, 0.01461) (0.00059, 0.01725) (0.00066, 0.01989) (0.00073, 0.02254) (0.0008, 0.0252) (0.00087, 0.02787) (0.00094, 0.03055) (0.00101, 0.03323) (0.00108, 0.03592) (0.00115, 0.03862) (0.00122, 0.04133) (0.00129, 0.04405) (0.00136, 0.04677) (0.00143, 0.04951) (0.0015, 0.05225) (0.00157, 0.055) (0.00163, 0.05776) (0.0017, 0.06052) (0.00177, 0.0633) (0.00184, 0.06608) (0.00191, 0.06887) (0.00198, 0.07168) (0.00205, 0.07449) (0.00212, 0.0773) (0.00219, 0.08013) (0.00226, 0.08297) (0.00233, 0.08581) (0.0024, 0.08866) (0.00247, 0.09153) (0.00254, 0.0944) (0.00261, 0.09728) (0.00268, 0.10017) (0.00275, 0.10307) (0.00282, 0.10597) (0.00289, 0.10889) (0.00296, 0.11182) (0.00303, 0.11475) (0.0031, 0.1177) (0.00317, 0.12065) (0.00324, 0.12361) (0.00331, 0.12658) (0.00338, 0.12957) (0.00345, 0.13256) (0.00352, 0.13556) (0.00359, 0.13857) (0.00366, 0.14159) (0.00373, 0.14462) (0.0038, 0.14766) (0.00387, 0.15071) (0.00394, 0.15377) (0.00401, 0.15684) (0.00408, 0.15992) (0.00415, 0.16301) (0.00422, 0.16611) (0.00429, 0.16921) (0.00436, 0.17233) (0.00443, 0.17546) (0.0045, 0.1786) (0.00457, 0.18176) (0.00464, 0.18492) (0.0047, 0.18809) (0.00477, 0.19127) (0.00484, 0.19446) (0.00491, 0.19767) 
			};
			
			\addplot[plot3,very thick,smooth] coordinates {
			(0.0001, -0.00126) (0.00014, 4e-05) (0.00018, 0.00134) (0.00021, 0.00265) (0.00025, 0.00395) (0.00029, 0.00526) (0.00033, 0.00656) (0.00036, 0.00787) (0.0004, 0.00919) (0.00044, 0.0105) (0.00048, 0.01182) (0.00051, 0.01313) (0.00055, 0.01445) (0.00059, 0.01578) (0.00063, 0.0171) (0.00066, 0.01843) (0.0007, 0.01975) (0.00074, 0.02108) (0.00078, 0.02242) (0.00081, 0.02375) (0.00085, 0.02509) (0.00089, 0.02642) (0.00093, 0.02776) (0.00096, 0.02911) (0.001, 0.03045) (0.00104, 0.0318) (0.00108, 0.03314) (0.00111, 0.03449) (0.00115, 0.03585) (0.00119, 0.0372) (0.00123, 0.03856) (0.00126, 0.03992) (0.0013, 0.04128) (0.00134, 0.04264) (0.00138, 0.044) (0.00141, 0.04537) (0.00145, 0.04674) (0.00149, 0.04811) (0.00153, 0.04948) (0.00156, 0.05086) (0.0016, 0.05224) (0.00164, 0.05362) (0.00168, 0.055) (0.00171, 0.05638) (0.00175, 0.05777) (0.00179, 0.05915) (0.00183, 0.06054) (0.00186, 0.06194) (0.0019, 0.06333) (0.00194, 0.06473) (0.00198, 0.06613) (0.00201, 0.06753) (0.00205, 0.06893) (0.00209, 0.07034) (0.00213, 0.07174) (0.00216, 0.07315) (0.0022, 0.07456) (0.00224, 0.07598) (0.00228, 0.07739) (0.00231, 0.07881) (0.00235, 0.08023) (0.00239, 0.08165) (0.00243, 0.08308) (0.00246, 0.08451) (0.0025, 0.08594) (0.00254, 0.08737) (0.00258, 0.0888) (0.00261, 0.09024) (0.00265, 0.09168) (0.00269, 0.09312) 
			};
			
			\addplot[plot4,very thick,smooth] coordinates {
			(0.0001, -0.00162) (0.00012, -0.00092) (0.00014, -0.00021) (0.00016, 0.0005) (0.00018, 0.0012) (0.0002, 0.00191) (0.00021, 0.00262) (0.00023, 0.00333) (0.00025, 0.00404) (0.00027, 0.00475) (0.00029, 0.00546) (0.00031, 0.00617) (0.00033, 0.00688) (0.00035, 0.00759) (0.00037, 0.00831) (0.00039, 0.00902) (0.0004, 0.00974) (0.00042, 0.01045) (0.00044, 0.01117) (0.00046, 0.01188) (0.00048, 0.0126) (0.0005, 0.01332) (0.00052, 0.01403) (0.00054, 0.01475) (0.00056, 0.01547) (0.00058, 0.01619) (0.0006, 0.01691) (0.00061, 0.01763) (0.00063, 0.01835) (0.00065, 0.01907) (0.00067, 0.0198) (0.00069, 0.02052) (0.00071, 0.02124) (0.00073, 0.02197) (0.00075, 0.02269) (0.00077, 0.02342) (0.00079, 0.02414) (0.0008, 0.02487) (0.00082, 0.0256) (0.00084, 0.02633) (0.00086, 0.02705) (0.00088, 0.02778) (0.0009, 0.02851) (0.00092, 0.02924) (0.00094, 0.02997) (0.00096, 0.03071) (0.00098, 0.03144) (0.001, 0.03217) (0.00101, 0.0329) (0.00103, 0.03364) (0.00105, 0.03437) (0.00107, 0.03511) (0.00109, 0.03584) (0.00111, 0.03658) (0.00113, 0.03732) (0.00115, 0.03805) (0.00117, 0.03879) (0.00119, 0.03953) (0.0012, 0.04027) (0.00122, 0.04101) (0.00124, 0.04175) (0.00126, 0.04249) (0.00128, 0.04323) (0.0013, 0.04398) (0.00132, 0.04472) (0.00134, 0.04546) (0.00136, 0.04621) (0.00138, 0.04695) (0.0014, 0.0477) (0.00141, 0.04844) 
			};

			\addplot[plot5,very thick,smooth] coordinates {
			(0.0001, -0.00232) (0.00011, -0.00196) (0.00012, -0.00161) (0.00012, -0.00125) (0.00013, -0.0009) (0.00014, -0.00054) (0.00015, -0.00019) (0.00015, 0.00017) (0.00016, 0.00052) (0.00017, 0.00088) (0.00018, 0.00124) (0.00018, 0.00159) (0.00019, 0.00195) (0.0002, 0.0023) (0.00021, 0.00266) (0.00022, 0.00302) (0.00022, 0.00337) (0.00023, 0.00373) (0.00024, 0.00409) (0.00025, 0.00444) (0.00025, 0.0048) (0.00026, 0.00516) (0.00027, 0.00552) (0.00028, 0.00587) (0.00028, 0.00623) (0.00029, 0.00659) (0.0003, 0.00695) (0.00031, 0.00731) (0.00032, 0.00767) (0.00032, 0.00802) (0.00033, 0.00838) (0.00034, 0.00874) (0.00035, 0.0091) (0.00035, 0.00946) (0.00036, 0.00982) (0.00037, 0.01018) (0.00038, 0.01054) (0.00038, 0.0109) (0.00039, 0.01126) (0.0004, 0.01162) (0.00041, 0.01198) (0.00042, 0.01234) (0.00042, 0.0127) (0.00043, 0.01306) (0.00044, 0.01342) (0.00045, 0.01378) (0.00045, 0.01414) (0.00046, 0.0145) (0.00047, 0.01487) (0.00048, 0.01523) (0.00048, 0.01559) (0.00049, 0.01595) (0.0005, 0.01631) (0.00051, 0.01667) (0.00052, 0.01704) (0.00052, 0.0174) (0.00053, 0.01776) (0.00054, 0.01812) (0.00055, 0.01849) (0.00055, 0.01885) (0.00056, 0.01921) (0.00057, 0.01958) (0.00058, 0.01994) (0.00058, 0.0203) (0.00059, 0.02067) (0.0006, 0.02103) (0.00061, 0.0214) (0.00062, 0.02176) (0.00062, 0.02212) (0.00063, 0.02249) 
			};
			
			\addplot[color=black, only marks, mark=\MARKFORM, mark size=\MARKSZ] coordinates { (0.01294, 0.97362) };
			\addplot[color=black, only marks, mark=\MARKFORM, mark size=\MARKSZ] coordinates { (0.00491, 0.19767) };
			\addplot[color=black, only marks, mark=\MARKFORM, mark size=\MARKSZ] coordinates { (0.00269, 0.09312) };
			\addplot[color=black, only marks, mark=\MARKFORM, mark size=\MARKSZ] coordinates { (0.00141, 0.04844) };
			\addplot[color=black, only marks, mark=\MARKFORM, mark size=\MARKSZ] coordinates { (0.00063, 0.02249) };
		\end{axis}  
	\end{tikzpicture}
	~
	\begin{tikzpicture}
		\begin{axis}[
			title=$n \stackrel{}{=} 5 \cdot 10^8 \text{, } \delta_\mathrm{est} \stackrel{}{=} 10^{-4}$,
			xlabel=$S_{\mathrm{exp}}$,
			ylabel=$\eta_{\mathrm{opt}}$,
			xmin=0,
			xmax=0.01293,
			ymax=1,
			ymin=0,
			ytick={0,0.2,0.4,0.6,0.8,1},
			legend style={at={(0.38, 0.97)},anchor=north,legend cell align=left,font=\FONTSZ} 
			]
		
		
			\addplot[plot1,very thick,smooth] coordinates {
			(0.0001, -0.01276) (0.00029, -0.00417) (0.00047, 0.00446) (0.00066, 0.01317) (0.00084, 0.02196) (0.00103, 0.03085) (0.00122, 0.03982) (0.0014, 0.04889) (0.00159, 0.05805) (0.00178, 0.0673) (0.00196, 0.07664) (0.00215, 0.08609) (0.00233, 0.09563) (0.00252, 0.10527) (0.00271, 0.11501) (0.00289, 0.12486) (0.00308, 0.13481) (0.00326, 0.14486) (0.00345, 0.15502) (0.00364, 0.1653) (0.00382, 0.17568) (0.00401, 0.18618) (0.0042, 0.1968) (0.00438, 0.20754) (0.00457, 0.21839) (0.00475, 0.22937) (0.00494, 0.24048) (0.00513, 0.25171) (0.00531, 0.26308) (0.0055, 0.27458) (0.00568, 0.28622) (0.00587, 0.298) (0.00606, 0.30992) (0.00624, 0.32199) (0.00643, 0.33421) (0.00662, 0.34659) (0.0068, 0.35913) (0.00699, 0.37184) (0.00717, 0.38471) (0.00736, 0.39776) (0.00755, 0.41099) (0.00773, 0.42441) (0.00792, 0.43801) (0.0081, 0.45182) (0.00829, 0.46583) (0.00848, 0.48006) (0.00866, 0.4945) (0.00885, 0.50918) (0.00904, 0.5241) (0.00922, 0.53927) (0.00941, 0.55471) (0.00959, 0.57042) (0.00978, 0.58642) (0.00997, 0.60273) (0.01015, 0.61936) (0.01034, 0.63633) (0.01052, 0.65367) (0.01071, 0.6714) (0.0109, 0.68955) (0.01108, 0.70815) (0.01127, 0.72724) (0.01145, 0.74688) (0.01164, 0.76711) (0.01183, 0.78801) (0.01201, 0.80967) (0.0122, 0.8322) (0.01239, 0.85576) (0.01257, 0.88058) (0.01276, 0.90701) (0.01294, 0.93569) 
			};
			\addlegendentry{$\mu_{\mathrm{min}} = 0.250, \mu_{\mathrm{max}} = 0.250$}
			
			\addplot[plot2,very thick,smooth] coordinates {
			(0.0001, -0.01469) (0.00017, -0.0121) (0.00024, -0.00951) (0.00031, -0.00693) (0.00038, -0.00434) (0.00045, -0.00175) (0.00052, 0.00084) (0.00059, 0.00345) (0.00066, 0.00606) (0.00073, 0.00867) (0.0008, 0.0113) (0.00087, 0.01393) (0.00094, 0.01657) (0.00101, 0.01922) (0.00108, 0.02188) (0.00115, 0.02455) (0.00122, 0.02722) (0.00129, 0.0299) (0.00136, 0.03259) (0.00143, 0.03529) (0.0015, 0.038) (0.00157, 0.04071) (0.00163, 0.04343) (0.0017, 0.04616) (0.00177, 0.0489) (0.00184, 0.05165) (0.00191, 0.05441) (0.00198, 0.05717) (0.00205, 0.05994) (0.00212, 0.06273) (0.00219, 0.06552) (0.00226, 0.06831) (0.00233, 0.07112) (0.0024, 0.07394) (0.00247, 0.07676) (0.00254, 0.07959) (0.00261, 0.08244) (0.00268, 0.08529) (0.00275, 0.08815) (0.00282, 0.09102) (0.00289, 0.09389) (0.00296, 0.09678) (0.00303, 0.09968) (0.0031, 0.10258) (0.00317, 0.1055) (0.00324, 0.10842) (0.00331, 0.11135) (0.00338, 0.11429) (0.00345, 0.11724) (0.00352, 0.1202) (0.00359, 0.12317) (0.00366, 0.12615) (0.00373, 0.12914) (0.0038, 0.13214) (0.00387, 0.13515) (0.00394, 0.13816) (0.00401, 0.14119) (0.00408, 0.14423) (0.00415, 0.14728) (0.00422, 0.15033) (0.00429, 0.1534) (0.00436, 0.15647) (0.00443, 0.15956) (0.0045, 0.16266) (0.00457, 0.16576) (0.00464, 0.16888) (0.0047, 0.17201) (0.00477, 0.17514) (0.00484, 0.17829) (0.00491, 0.18145) 
			};
			\addlegendentry{$\mu_{\mathrm{min}} = 0.210, \mu_{\mathrm{max}} = 0.371$}
			
			\addplot[plot3,very thick,smooth] coordinates {
			(0.0001, -0.01777) (0.00014, -0.01647) (0.00018, -0.01517) (0.00021, -0.01387) (0.00025, -0.01257) (0.00029, -0.01127) (0.00033, -0.00997) (0.00036, -0.00867) (0.0004, -0.00737) (0.00044, -0.00607) (0.00048, -0.00477) (0.00051, -0.00348) (0.00055, -0.00218) (0.00059, -0.00088) (0.00063, 0.00043) (0.00066, 0.00173) (0.0007, 0.00304) (0.00074, 0.00435) (0.00078, 0.00566) (0.00081, 0.00697) (0.00085, 0.00828) (0.00089, 0.0096) (0.00093, 0.01092) (0.00096, 0.01224) (0.001, 0.01356) (0.00104, 0.01489) (0.00108, 0.01621) (0.00111, 0.01754) (0.00115, 0.01887) (0.00119, 0.0202) (0.00123, 0.02154) (0.00126, 0.02288) (0.0013, 0.02421) (0.00134, 0.02556) (0.00138, 0.0269) (0.00141, 0.02824) (0.00145, 0.02959) (0.00149, 0.03094) (0.00153, 0.03229) (0.00156, 0.03364) (0.0016, 0.035) (0.00164, 0.03635) (0.00168, 0.03771) (0.00171, 0.03907) (0.00175, 0.04044) (0.00179, 0.0418) (0.00183, 0.04317) (0.00186, 0.04454) (0.0019, 0.04591) (0.00194, 0.04728) (0.00198, 0.04866) (0.00201, 0.05004) (0.00205, 0.05142) (0.00209, 0.0528) (0.00213, 0.05418) (0.00216, 0.05557) (0.0022, 0.05696) (0.00224, 0.05835) (0.00228, 0.05974) (0.00231, 0.06114) (0.00235, 0.06253) (0.00239, 0.06393) (0.00243, 0.06533) (0.00246, 0.06674) (0.0025, 0.06814) (0.00254, 0.06955) (0.00258, 0.07096) (0.00261, 0.07237) (0.00265, 0.07379) (0.00269, 0.07521) 
			};
			\addlegendentry{$\mu_{\mathrm{min}} = 0.167, \mu_{\mathrm{max}} = 0.500$}
			
			\addplot[plot4,very thick,smooth] coordinates {
			(0.0001, -0.02298) (0.00012, -0.02227) (0.00014, -0.02156) (0.00016, -0.02086) (0.00018, -0.02015) (0.0002, -0.01944) (0.00021, -0.01874) (0.00023, -0.01803) (0.00025, -0.01732) (0.00027, -0.01662) (0.00029, -0.01591) (0.00031, -0.0152) (0.00033, -0.0145) (0.00035, -0.01379) (0.00037, -0.01308) (0.00039, -0.01238) (0.0004, -0.01167) (0.00042, -0.01096) (0.00044, -0.01026) (0.00046, -0.00955) (0.00048, -0.00884) (0.0005, -0.00814) (0.00052, -0.00743) (0.00054, -0.00672) (0.00056, -0.00602) (0.00058, -0.00531) (0.0006, -0.00461) (0.00061, -0.0039) (0.00063, -0.00319) (0.00065, -0.00249) (0.00067, -0.00178) (0.00069, -0.00107) (0.00071, -0.00036) (0.00073, 0.00035) (0.00075, 0.00106) (0.00077, 0.00176) (0.00079, 0.00248) (0.0008, 0.00319) (0.00082, 0.0039) (0.00084, 0.00461) (0.00086, 0.00532) (0.00088, 0.00604) (0.0009, 0.00675) (0.00092, 0.00746) (0.00094, 0.00818) (0.00096, 0.0089) (0.00098, 0.00961) (0.001, 0.01033) (0.00101, 0.01105) (0.00103, 0.01176) (0.00105, 0.01248) (0.00107, 0.0132) (0.00109, 0.01392) (0.00111, 0.01464) (0.00113, 0.01536) (0.00115, 0.01608) (0.00117, 0.01681) (0.00119, 0.01753) (0.0012, 0.01825) (0.00122, 0.01898) (0.00124, 0.0197) (0.00126, 0.02043) (0.00128, 0.02115) (0.0013, 0.02188) (0.00132, 0.02261) (0.00134, 0.02333) (0.00136, 0.02406) (0.00138, 0.02479) (0.0014, 0.02552) (0.00141, 0.02625) 
			};
			\addlegendentry{$\mu_{\mathrm{min}} = 0.124, \mu_{\mathrm{max}} = 0.629$}

			\addplot[plot5,very thick,smooth] coordinates {
			(0.0001, -0.03279) (0.00011, -0.03243) (0.00012, -0.03208) (0.00012, -0.03172) (0.00013, -0.03137) (0.00014, -0.03101) (0.00015, -0.03065) (0.00015, -0.0303) (0.00016, -0.02994) (0.00017, -0.02959) (0.00018, -0.02923) (0.00018, -0.02888) (0.00019, -0.02852) (0.0002, -0.02817) (0.00021, -0.02781) (0.00022, -0.02746) (0.00022, -0.0271) (0.00023, -0.02675) (0.00024, -0.02639) (0.00025, -0.02603) (0.00025, -0.02568) (0.00026, -0.02532) (0.00027, -0.02497) (0.00028, -0.02461) (0.00028, -0.02426) (0.00029, -0.0239) (0.0003, -0.02355) (0.00031, -0.02319) (0.00032, -0.02284) (0.00032, -0.02248) (0.00033, -0.02213) (0.00034, -0.02177) (0.00035, -0.02141) (0.00035, -0.02106) (0.00036, -0.0207) (0.00037, -0.02035) (0.00038, -0.01999) (0.00038, -0.01964) (0.00039, -0.01928) (0.0004, -0.01893) (0.00041, -0.01857) (0.00042, -0.01822) (0.00042, -0.01786) (0.00043, -0.01751) (0.00044, -0.01715) (0.00045, -0.01679) (0.00045, -0.01644) (0.00046, -0.01608) (0.00047, -0.01573) (0.00048, -0.01537) (0.00048, -0.01502) (0.00049, -0.01466) (0.0005, -0.01431) (0.00051, -0.01395) (0.00052, -0.0136) (0.00052, -0.01324) (0.00053, -0.01288) (0.00054, -0.01253) (0.00055, -0.01217) (0.00055, -0.01182) (0.00056, -0.01146) (0.00057, -0.01111) (0.00058, -0.01075) (0.00058, -0.0104) (0.00059, -0.01004) (0.0006, -0.00969) (0.00061, -0.00933) (0.00062, -0.00898) (0.00062, -0.00862) (0.00063, -0.00826) 
			};
			\addlegendentry{$\mu_{\mathrm{min}} = 0.083, \mu_{\mathrm{max}} = 0.750$}
			
			\addplot[color=black, only marks, mark=\MARKFORM, mark size=\MARKSZ] coordinates { (0.01294, 0.93569) };
			\addplot[color=black, only marks, mark=\MARKFORM, mark size=\MARKSZ] coordinates { (0.00491, 0.18145) };
			\addplot[color=black, only marks, mark=\MARKFORM, mark size=\MARKSZ] coordinates { (0.00269, 0.07521) };
			\addplot[color=black, only marks, mark=\MARKFORM, mark size=\MARKSZ] coordinates { (0.00141, 0.02625) };
		\end{axis}  
	\end{tikzpicture}
	\label{fig:entropy_rate_mu}
	}
	\caption{Plots of the entropy rate, $\eta_{\mathrm{opt}}$, as a function of the expected MDL violation, $S_{\mathrm{exp}}$, for different parameters of the EAP.
	In addition to the noted values of $n$, $\mu$, and $\delta_{\mathrm{est}} = 10^{-5}$, we used the parameters $\varepsilon_{\mathrm{s}} = \varepsilon_{\mathrm{EA}} = 10^{-7}$.
	As in Figure~\ref{fig:entropy_single_round_S} the black dots indicate the maximal possible MDL violation for given $\mu$.
	Note that in Figure~\ref{fig:entropy_rate_n} the asymptotic rates are equal to the curves in Figure~\ref{fig:entropy_single_round_S} with corresponding $\mu$.
	Furthermore the curves in Figure~\ref{fig:entropy_rate_mu} converge to the curves in Figure~\ref{fig:entropy_single_round_S}.
	Both these facts are a consequence of the asymptotic equipartition property.
	For convenience we choose $\mu$ such that $\mu_{\mathrm{max}} = 1 - 3\cdot \mu_{\mathrm{min}}$.
	}
\label{fig:entropy_rate}
\end{figure}
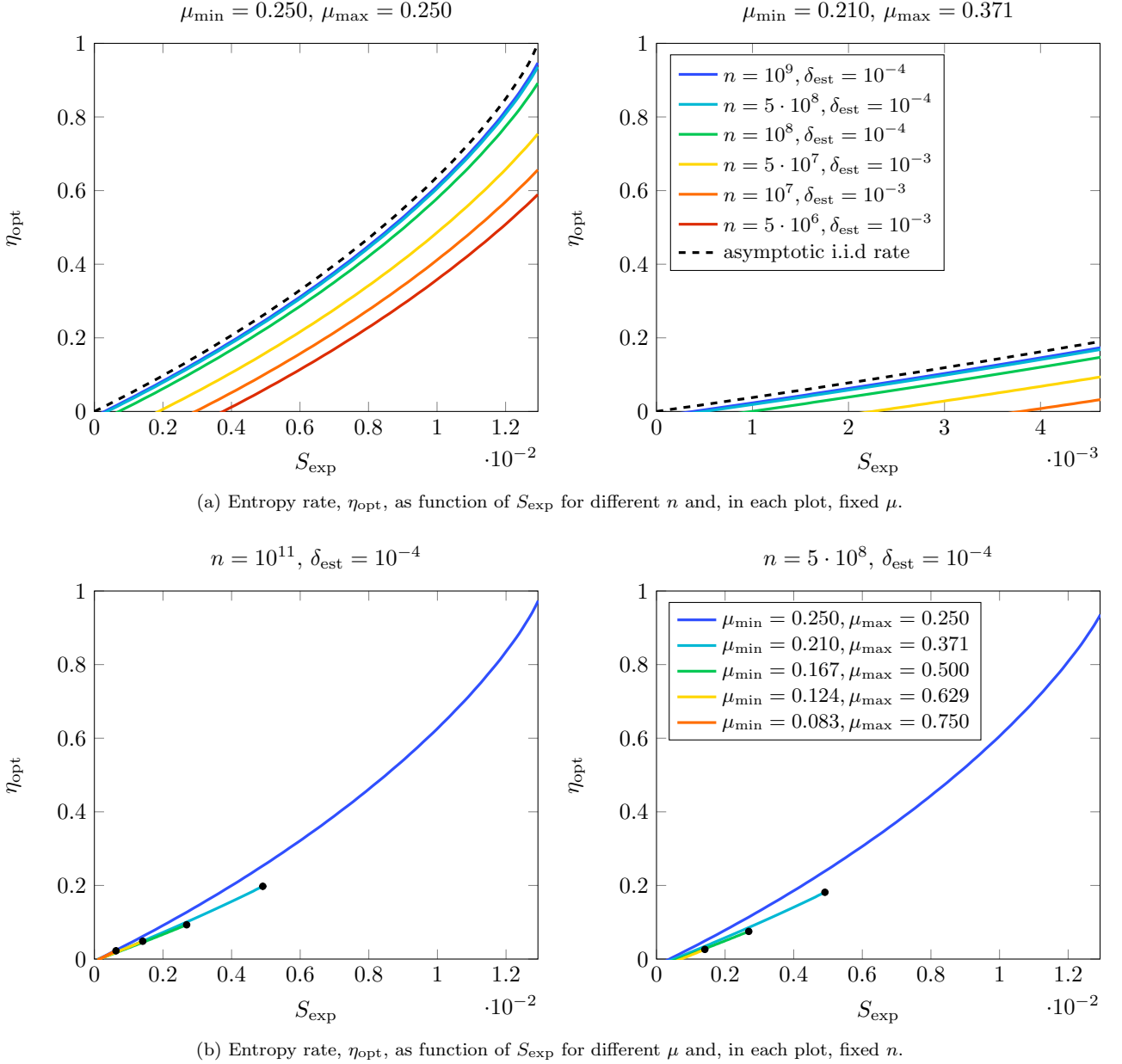

\subsubsection{Applying the extractor}

So far we gave an explicit lower bound on the smooth min-entropy of the entropy accumulation routine's output.
The last part in our RAP, that produces the final bits, is the application of a randomness extractor (cf. Section~\ref{sec:extractors}).
More precisely we are using a quantum-proof two-source randomness extractor in the Markov model.

Using a two-source extractor we need, as the name indicates, two inputs.
The first input that we use is the outputs generated by the entropy accumulation routine, $A^{n}B^{n}$.
As the second input to the extractor we use additional raw bits from the source.\footnote{The first part of the source's output is used as input for the entropy accumulation routine and the second part as second input for the extractor.}
Thus we first use the source to produce the inputs to the MDL experiment and then to draw inputs, $Z^{d}$, for the extractor directly.
The exact setup that we use for that is described in Section~\ref{sec:assumptions}.

Since we are using extractors that work in the Markov model we require that the two inputs to the extractor are independent conditioned on the adversary's side information, i.e., $I(Z^{d} : A^{n} B^{n} | X^{n} Y^{n} E, \lambda) = 0$.
The fact that this is indeed the case in our setting is explicitly stated in Section~\ref{sec:assumptions}.
Hence we can use the extractor to quantify the secrecy of the outputs.

\begin{rmk}
\label{rmk:non-markov}
	When using $I(Z^{d} : A^{n} B^{n} | X^{n} Y^{n} E, \lambda) = 0$ we assume that the adversary has full access to $E$ and $\lambda$.
	However in a realistic setting this might not be the case, thus leading to the Markov condition not being satisfied.
	Nevertheless, as stated in Section~5.2 in~\cite{Extractors}, the deletion of a part of the side information cannot decrease the security of the extractor.
	Consequently, if the adversary is less powerful and does not have access to all of $E$ and $\lambda$, and thus the Markov condition is not satisfied, the extractor is still secure.
\end{rmk}

The quality of the extractor's output depends on the (smooth) min-entropy of the two sources.
Thus, in addition to the mutual information vanishing, we also need to know what the min-entropy of the random variables $Z^{d}$ is.

\begin{lma}
\label{lma:MDL-min-entropy}
	Let $Z^{d}$ be the output of a $\mu$-MDL source. Then, the lower bound
	\begin{equation*}
		H_{\mathrm{min}}(Z^{d}|X^{n}Y^{n} E, \lambda) \geq - \frac{d}{2} \cdot \log(\mu_{\mathrm{max}})
	\end{equation*}
	on the min-entropy holds.
\end{lma}

\begin{proof}
	For a $\mu$-MDL source we require that the guessing probability of the outputs is bounded (recall form Section~\ref{sec:assumptions}),
	\[
		\mu_{\mathrm{min}} \leq p_{\mathrm{guess}} (Z_{2i}Z_{2i+1}|Z^{2i-1} E, \lambda) \leq \mu_{\mathrm{max}} \; \forall \lambda, i \,.
	\]
	
	Thus the maximal probability of any particular string appearing is at most $\mu_{\mathrm{max}}^{\nicefrac{d}{2}}$.
	Finally, since the min-entropy is the negative logarithm of the maximal guessing probability, the lemma follows.
\end{proof}

Using the results from~\cite{Extractors}, Theorem~\ref{thm:main}, and Lemma~\ref{lma:MDL-min-entropy} we can determine how close to uniform the output of our RAP is.

\begin{lma}
\label{lma:RAP-secrecy}
	Let $\mathrm{Ext}:\{0,1\}^{2n}\times\{0,1\}^d\rightarrow\{0,1\}^m$ be a $(k_1,k_2,\varepsilon_{ext})$ be a two-source quantum-proof extractor in the Markov model, strong in the second input, such that
	\begin{equation} \label{eq:ext_k_values}
	\begin{split}
		k_1 &\leq n \cdot \eta_{\mathrm{opt}}(\varepsilon_{\mathrm{EA}}, \varepsilon_{\mathrm{s}}, S_{\mathrm{exp}} - \delta_{\mathrm{est}}, n, \mu) -\log(1/\varepsilon_{ext})-1 \\
		k_2 &\leq -\frac{d}{2} \cdot \log(\mu_{\mathrm{max}}) -\log(1/\varepsilon_{ext})-1 
	\end{split}
	\end{equation}
	 Consider the RAP (Protocol~\ref{alg:RAP}) using $\mathrm{Ext}$ and any $\varepsilon_{\mathrm{EA}}, \varepsilon_{\mathrm{s}} \in (0,1)$. 
	Then, either the protocol aborts with probability greater than $1 - \varepsilon_{\mathrm{EA}}$, or for the output $K^{m}$ together with the whole information the adversary possibly has access to, $\Sigma = Z^{d} X^{n} Y^{n} E \lambda$, it holds that
	\begin{equation*}
		\frac{1}{2} \left\Vert \rho_{K^{m} \Sigma} - \rho_{U_m} \otimes \rho_{\Sigma} \right\Vert \leq 6 \left(\varepsilon_{\mathrm{s}} +  \varepsilon_{\mathrm{ext}} \right) \,.
	\end{equation*}
\end{lma}
\begin{proof}
	Starting with Theorem~\ref{thm:main}, we know that, either the protocol aborts with probability greater than $1 - \varepsilon_{\mathrm{EA}}$, or the smooth min-entropy of the entropy accumulation routine's output is lower bounded by $n \cdot \eta_{\mathrm{opt}}(\varepsilon_{\mathrm{EA}}, \varepsilon_{\mathrm{s}}, S_{\mathrm{exp}} - \delta_{\mathrm{est}}, n, \mu)$.
	For the second input of the extractor, using Lemma~\ref{lma:MDL-min-entropy}, we know that the min-entropy of the string $Z^{d}$ is lower bounded by $\frac{d}{2} \cdot \log(\nicefrac{1}{\mu_{\mathrm{max}}})$.
	Furthermore by assumption the state that is generated in the protocol is a Markov state.
	Thus we can employ Lemma~\ref{lma:smooth-entropy-bound} to get an upper bound on the distance between the output $K^{m}$ and a uniform string, and proof the claim.
\end{proof}

\begin{rmk}
\label{rmk:RAP-secrecy}
	As stated in Lemma~\ref{lem:quantum_markov_two_source}, one can construct two-source quantum-proof extractors in the Markov model from classical ones.  The parameters of the chosen extractor affect the parameters of our protocol directly. In particular, the security parameter (given below) and the efficiency of the protocol (the extraction rate) depend on the extractor. 
	It is important to note that there are explicit extractors with good parameters for our purpose. In particular:
	\begin{enumerate}
		\item If one of the two sources (either the device's outputs $A^nB^n$ or the seed $Z^{d}$ for $d=2n$) has (smooth) min-entropy of more than $n$ one can use the explicit construction of an extractor given in~\cite[Corollary 25]{Extractors} to extract a linear number of bits. Focusing on the the seed, the min-entropy is sufficiently high when $\mu_{\max} \leq 1/2$.\footnote{In terms of an SV-source, the source should be such that, roughly, $0.3\leq p(0) \leq 0.7$; recall Lemma~\ref{lma:SV-MDL}.} 
		\item Otherwise, one can use the explicit construction of an extractor given in~\cite[Corollary 30]{Extractors} to extract a logarithmic number of bits.
		\item To extract a sub-linear number of bits using an explicit extractor one can also consider a simple modification of our protocol, similarly to what was done in~\cite[Theorem~2]{brandao2016realistic} -- given another device with two components one can use the inputs to run the same protocol with the additional device and by this create another source of randomness. Combined with what we had before, we now have three sources of randomness (the outputs of the two devices and the seed) in the Markov model (see~\cite[Definition~7]{Extractors}) . Thus, the three-source extractor given in~\cite[Corollary 28]{Extractors} can be used to extract a sub-linear number of bits.
	\end{enumerate}
\end{rmk}

After putting everything together we can determine the secrecy parameter for our RAP corresponding to the secrecy definition (Definition~\ref{def:RAP-secrecy}).
In the final theorem we state $\varepsilon_{\mathrm{RA}}$ in terms of the RAP's parameters. 

\begin{thm}[Secrecy]
\label{thm:secrecy}
	For any $\varepsilon_{\mathrm{EA}}, \varepsilon_{\mathrm{s}} \in (0,1)$ the RAP (Protocol~\ref{alg:RAP}) with the given parameters is $\varepsilon_{\mathrm{RA}}$-secret (according to Definition~\ref{def:RAP-secrecy}), with $\varepsilon_{\mathrm{RA}} = 12 \left(\varepsilon_{\mathrm{s}} +\varepsilon_{\mathrm{ext}} \right)+ \varepsilon_{\mathrm{EA}}$.
\end{thm}
\begin{proof}
	In the following let $\Sigma = Z^{d} X^{n} Y^{n} E \lambda$ be the whole information the adversary has access to.
	Starting with Lemma~\ref{lma:RAP-secrecy} we can distinguish two cases.
	\begin{cse}
		The protocol aborts with probability greater than $1 - \varepsilon_{\mathrm{EA}}$.
	\end{cse}
	In that case, we find 
	\[
		\left( 1 - \mathrm{Pr}[\text{abort}] \right) \left\Vert \rho_{K^{m} \Sigma} - \rho_{U^{m}} \otimes \rho_{\Sigma} \right\Vert \leq \varepsilon_{\mathrm{EA}} \left\Vert \rho_{K^{m} \Sigma} - \rho_{U^{m}} \otimes \rho_{\Sigma} \right\Vert \\
		\leq \varepsilon_{\mathrm{EA}} \;,
	\]
	since the trace distance is always less than one.
	
	\begin{cse}
		The protocol aborts with probability less than $1 - \varepsilon_{\mathrm{EA}}$ (hence the entropy is sufficiently high).
	\end{cse}
	In that case, using the bound from Lemma~\ref{lma:RAP-secrecy}, we find	
	\[
		\left( 1 - \mathrm{Pr}[\text{abort}] \right) \left\Vert \rho_{K^{m} \Sigma} - \rho_{U^{m}} \otimes \rho_{\Sigma} \right\Vert \leq \left\Vert \rho_{K^{m} \Sigma} - \rho_{U^{m}} \otimes \rho_{\Sigma} \right\Vert \leq 12 \left(\varepsilon_{\mathrm{s}} +\varepsilon_{\mathrm{ext}} \right) \,. \qedhere
	\]
\end{proof}

We can now continue to prove Theorem~\ref{thm:formal}.

\begin{proof}[Proof of Theorem~\ref{thm:formal}]
	Part~\ref{part:soundness} follows directly from the proof of Theorem~\ref{thm:secrecy}.
	Part~\ref{part:completeness} follows directly from Lemma~\ref{lma:completeness}.
\end{proof}

\section{Open questions} \label{sec:conclustions}

We end with some open questions:

\begin{enumerate}[1.]
	\item Is the amount of extractable randomness given in our work tight? There are few things that one can consider when trying to improve the extraction rate:
	\begin{enumerate}[i.]
		\item While the bound given in Lemma~\ref{lma:holevo_bound} is non-trivial for any violation of the MDL inequality, it might not be tight.
		\item We used the MDL inequality derived in~\cite{MDL}. They derived their inequality with the motivation of detecting quantumness for an arbitrary MDL source. Thus it might be possible that there are other MDL inequalities that are more suitable for quantifying randomness.
		\item The final length of the extracted randomness depends on the parameters of the extractor used. Finding quantum-proof multi-source extractors for the Markov model which have good parameters is therefore of interest. This can be achieved by considering better specific (classical) two-source extractors and then applying the technique of~\cite{Extractors}, or by improving over the parameters of~\cite{Extractors} for general constructions.
	\end{enumerate}
	\item Can the analysis be extended such that the adversary is allowed to hold some quantum side information about the source? Currently we only allow the adversary to know $\lambda$ in advance (while $E$ is the quantum side information about the device itself). In Particular, this is a realistic assumptions in scenarios where the device and the producer of the weak source are different parties. Nevertheless, it will be interesting to see whether holding quantum side information about the source before producing the device is beneficial for the adversary and what the consequences for the security of our protocol are.
	\item Is it possible to amplify min-entropy sources while maintaining similar parameters? In particular, can it be done with a constant number of devices? (in contrast to what was done in~\cite{chung2014physical}).
	The technique presented here does not work if the SV (MDL) source is replaced with a min-entropy sources (while it might be possible to extend them to block-sources). Thus, another approach has to be taken.
	\item Similarly, is it possible to amplify randomness against a non-signalling adversary  while maintaining similar parameters? 
	Our RAP works only against an adversary that is bound by quantum mechanics and an extension to the non-signalling case is not possible using the techniques that we employed. In particular, the proofs of both~\cite{EAT} and~\cite{Extractors} use the assumption that everything can be described with the formalism of quantum physics. We remark that, while it might be possible to extend one of these results to the non-signalling case, an extension of both of them will result in a contradiction with~\cite{arnon2012limits}. 
	Previous works that focused on non-signalling adversaries cannot be used to achieve similar statements as we derived in this work. 
	\item What is the effect of using a weak source of randomness in device-independent protocols that assume prefect randomness, e.g., device-independent quantum key distribution protocol or randomness expansion? 
	In such protocols random bits are used not only for choosing the inputs for the devices, but also for choosing the rounds in which a ``test'' is carried out.
	To analyse the effect of replacing perfect randomness with weak randomness one can use our RAP.
	One trivial possibility to include our RAP would be to just use it separately to generate uniform bits, before starting with the other protocols.
	Another option is to use our protocol as the main building block for the test rounds. The test rounds themselves can then be chosen with the SV-source, by using techniques such as enumeration. 
\end{enumerate}

\subsection*{Acknowledgments}
We thank Gilles P{\"u}tz for helpful discussions about the MDL inequalities and for letting us use his code to evaluate numerically the optimal violation of the inequalities within quantum physics.  We also thank Jean-Daniel Bancal, Roger Colbeck, Christopher Portman,  and Thomas Vidick for helpful comments. 
RAF was supported by the Swiss National Science Foundation via the  National Center for Competence in Research, QSIT, and by the Air Force Office of Scientific Research (AFOSR) via grant~FA9550-16-1-0245.

\appendix
\appendixpage

\section{Finding the maximal quantum violation of an MDL inequality}
\label{apx:max-viol}

It is not possible to find the maximal MDL value ($S_{\mu}^{*}$) in quantum mechanics for an MDL source with fixed $\mu$ since this value depends on the specific probability distribution of the source.\footnote{For fixed $\mu$ the probability distribution for the source's outputs is not necessarily fixed.}
However, we can find a lower bound on $S_{\mu}^{*}$ by taking the worst case probability distribution for a fixed $\mu$.
What we get is the value 
\begin{align}
	\tilde{S}_\mu &\equiv \mu_\mathrm{min}^2 P_{AB|XY}(00|00) - \mu_\mathrm{max}^2 \big( P_{AB|XY}(01|01) + P_{AB|XY}(10|10) + P_{AB|XY}(00|11) \big) \,. \label{eq:tilde-S}
\end{align}
The value $\tilde{S}_{\mu}$ is a lower bound on $S_{\mu}$ that is independent of the source as long as $\mu$ is fixed.
Therefore, when we find the maximum of $\tilde{S}_{\mu}$ in quantum mechanics ($\tilde{S}_{\mu}^{*}$) we also get lower bound on $S_{\mu}^{*}$.

\begin{lma}
	For fixed state and measurements $\tilde{S}_\mu$ is a lower bound for $S_\mu$ (as defined in Equation~\ref{eq:MDL_ineq}).
\label{lma:MDL_ineq_cond}
\end{lma} 
\begin{proof}
	First note that with $\mu_\mathrm{min} \leq P_{XY|\Sigma} (xy|\sigma) \leq \mu_\mathrm{max} \; \forall x,y,\sigma$ and $P_{XY} = \sum_\sigma P_\Sigma(\sigma) P_{XY|\Sigma} (xy|\sigma)$ it also holds that 
	\begin{equation}
		\mu_\mathrm{min} \leq P_{XY} (xy) \leq \mu_\mathrm{max} \; \forall x,y \,. \label{eq:mdl_w/o_lambda}
	\end{equation}
	Employing these bounds we find
	\begin{align*}
		S_\mu &= \mu_\mathrm{min} P_{ABXY}(0000) - \mu_\mathrm{max} \big( P_{ABXY}(0101) + P_{ABXY}(1010) + P_{ABXY}(0011) \big) \notag \\
		&= \mu_\mathrm{min} \underbrace{P_{XY}(00)}_{\geq \mu_\mathrm{min}} P_{AB|XY}(00|00) - \\ 
		& \qquad - \mu_\mathrm{max} \big( \underbrace{P_{XY}(01)}_{\leq \mu_\mathrm{max}} P_{AB|XY}(01|01) + \underbrace{P_{XY}(10)}_{\leq \mu_\mathrm{max}} P_{AB|XY}(10|10) + \underbrace{P_{XY}(11)}_{\leq \mu_\mathrm{max}} P_{AB|XY}(00|11) \big) \notag \\
		&\geq \mu_\mathrm{min}^2 P_{AB|XY}(00|00) - \mu_\mathrm{max}^2 \big( P_{AB|XY}(01|01) + P_{AB|XY}(10|10) + P_{AB|XY}(00|11) \big) = \tilde{S}_\mu \nonumber
	\end{align*}
\end{proof}

We found $\tilde{S}_{\mu}^{*}$ by maximising the eigenvalue of the Bell operator as a function of the measurement parameters in Matlab.
For a Bell inequality $\sum_{a,b,x,y} \alpha_{abxy} P_{AB|XY}(ab|xy) \leq p_{\mathrm{local}}$ with parameters $\alpha_{abxy}$ and measurement operators $\{M_{a}^{x}\}_{a,x}$ and $\{M_{b}^{y}\}_{b,y}$ the Bell operator is defined as 
\begin{equation*}
	\mathcal{B} = \sum_{a,b,x,y} \alpha_{abxy} M_{a}^{x} \otimes M_{b}^{y} \,.
\end{equation*}

\section{Additional proofs}
\label{apx:proofs}

\begin{proof}[Proof of Lemma~\ref{lma:rewrite}]
	First of all we have
	\begin{align*}
		H(AB|XYFE) \geq H(A|XYFE) = H(A|XFE) \,.
	\end{align*}
	Here the first step follows because $A$ and $B$ are classical registers.
	The second step follows because the non-signalling condition holds between the two components of the device.
	Thus the dependence of $A$ on $Y$ can only be through $X$; i.e., $A$, $X$, and $Y$ form a Markov chain, $A \leftrightarrow X \leftrightarrow Y$.
	Furthermore, it holds that
	\begin{align*}
		H(A|XFE) &= \sum_x \mathrm{Pr}[X = x] \cdot H(A | FE, X = x) \,.
	\end{align*}
	
	Finally we can rewrite
	\begin{align}
		H(A | FE, X = x) &= H(AFE | X=x) - H(FE | X=x) \notag\\
		&= H(A | X=x) + H(FE | A, X=x) - H(FE | X = x) \notag \\
		&= H(A | X = x) - \chi (A : F E | X = x) \notag \\
		&= 1 - \chi (A : F E | X = x) \,, \label{eq:single_entropy_bound}
	\end{align}
	where we used the fact that the outputs are symmetrized (Step~\ref{step:MDL-exp-symm}) and we introduced the Holevo quantity $\chi (A : F E | X = x) = H(F E|X=x) - H(F E|A, X=x)$.
	
	Combining everything, the result follows.
\end{proof}

\bibliographystyle{alphaurl}
\bibliography{refs.bib}

\end{document}